\documentclass[12pt,english]{amsart}
\usepackage{ae,aecompl}
\usepackage[T1]{fontenc}
\usepackage[latin9]{inputenc}
\usepackage[dvips,letterpaper]{geometry}
\usepackage{babel}

\newcommand{\Ham}{\mathcal H}
\newcommand{\bD}{\tilde B}
\newcommand{\Pf}{{\rm Pf}}
\newcommand{\bI}{\widetilde{\rm bott}}
\newcommand{\bott}{{\rm bott}}

\usepackage{amsthm}
\usepackage{amsmath}
\usepackage{amstext}
\usepackage{amssymb}
\usepackage[unicode=true, pdfusetitle,
 bookmarks=true,bookmarksnumbered=false,bookmarksopen=false,
 breaklinks=false,pdfborder={0 0 1},backref=false,colorlinks=false]
 {hyperref}

\numberwithin{equation}{section} 
\numberwithin{figure}{section} 
\theoremstyle{plain}
\newtheorem{thm}{Theorem}[section]
  \theoremstyle{remark}
  \newtheorem{rem}[thm]{Remark}
 \theoremstyle{definition}
  
  \theoremstyle{plain}
  \newtheorem{lem}[thm]{Lemma}
  \theoremstyle{definition}
  \newtheorem{defn}[thm]{Definition}
  \newtheorem{problem}[thm]{Problem}
  \newtheorem{corollary}[thm]{Corollary}
\newtheorem{conjecture}[thm]{Conjecture}

\newcommand{\vlr}{v_{LR}}

\newcommand{\be}{\begin{equation}}
\newcommand{\ee}{\end{equation}}

\usepackage{graphicx}

\usepackage[all,dvips]{xy}

\providecommand{\tabularnewline}{\\}

\usepackage{array}

\begin{document}

\title{Topological Insulators and $C^*$-Algebras: Theory and Numerical Practice}

\author{Matthew~B.~Hastings}

\author{Terry A.~Loring}



\begin{abstract}
We apply ideas from $C^*$-algebra to the study of disordered topological insulators.  We
extract certain almost commuting matrices from the free Fermi Hamiltonian, describing band projected coordinate matrices.
By considering topological obstructions to approximating these matrices by exactly commuting matrices, we are able to compute invariants quantifying
different topological phases.  We generalize previous two dimensional results to higher dimensions; we give a
general expression for the topological invariants for arbitrary dimension and several symmetry classes, including chiral symmetry classes,
and we present a detailed
$K$-theory treatment of this expression for time reversal invariant three dimensional systems.  
We can use these results to show non-existence of localized Wannier functions for these systems.

We use this approach to calculate the index for
time-reversal invariant systems with spin-orbit scattering in three dimensions, on sizes up to $12^3$, averaging over a  large number
of samples.  The results
show an interesting separation between the localization transition and the point at
which the average index (which can be viewed  as an ``order parameter'' for the topological insulator) begins to fluctuate
from sample too sample, implying the existence of an unsuspected quantum phase transition separating two
different delocalized phases in this system.
One of the particular advantages of the $C^*$-algebraic technique that we present is that it is significantly faster in practice
than other methods of computing the index, allowing the study of larger systems.  In this paper, we present a detailed discussion of numerical implementation of our method.
\end{abstract}

\maketitle

\tableofcontents{}
\section{Topological Insulators and Wannier Functions}
Consider a free Fermi Hamiltonian, described by a matrix $\Ham_{ij}$, where $i,j$ label sites on some lattice.  Suppose the
Hamiltonian $H$ is local, so that $\Ham_{ij}$ decays rapidly in the spacing between $i$ and $j$, and  that $H$ has a gap in its spectrum.  Then,
the system can be in either topologically trivial or topologically nontrivial phases.  An example of a topologically nontrivial phase is provided by
the quantum Hall effect in two dimensions; if the Hamiltonian $H$ is a lattice Hamiltonian in a magnetic field with a non-zero Hall conductance, then there
is a topological obstruction to continuing the Hamiltonian $H$ to a trivial Hamiltonian without either making the gap small or violating
locality.  Here, by a ``trivial Hamiltonian'', we mean a Hamiltonian in which all sites are decoupled, so that $\Ham_{ij}$ is diagonal, while the magnitude
of the gap is related to the system size later.
Equivalently, such an obstruction is an obstruction to having localized Wannier functions in the system\cite{hastingsloring}, as discussed
further below.

We can consider also such a system with symmetries imposed.  For example, we can require that the Hamiltonian have time reversal
invariance.  Such time reversal invariant topological insulators were considered in \cite{topo2d}.  Such a system is topologically
trivial when considered as a time reversal non-invariant system; that is, there is no obstruction to continuing the system to a trivial
system {\it if we do not require that the path connecting the Hamiltonian to a trivial Hamiltonian also be time reversal invariant}.  However, there is an obstruction to
continuing such systems to time reversal invariant trivial systems along time reversal invariant paths.

Many other symmetry classes can be considered, and in addition one can consider systems in higher dimensions\cite{3drefs}.
One important concept in the classification of topological systems in higher dimensions is the idea of ``stable equivalence''.
We consider two Hamiltonians $\Ham_0,\Ham_1$ to be connected if we can add some number of additional trivial degrees of freedom to $\Ham_0$, then follow
a continuous path in parameter space maintaining locality and spectral gap, and then arrive at a final Hamiltonian which is equal to $\Ham_1$
plus possibly additional trivial degrees of freedom.  Using this definition
of stable equivalence, we define a system to be in a topologically nontrivial
phase if it is not stably equivalent to a trivial Hamiltonian.
A general table of topological obstructions in different symmetry classes and dimensions was presented in 
\cite{ludwig,kitaev}.

In \cite{hastingsloring,loringhastings,fermicontinue}, an alternative approach to studying these topological insulators was presented, based on
$C^*$-algebras.  We now describe the approach for systems in the three classical universality classes, 
where the Hamiltonian is either an arbitrary Hermitian matrix, a real symmetric matrix, or a self-dual matrix, respectively.  These classes are
sometimes referred to as A, AI, AII.  We use the names GUE,GOE,GSE (Gaussian unitary, orthogonal, and symplectic) for these classes; this is not intended to imply that the Hamiltonian is drawn from some particular Gaussian distribution but is simply shorthand for the time-reversal
symmetries imposed on the Hamiltonian.  
Later in this paper we give the appropriate
generalization for the chiral cases (which we refer to as chiral, chiral real, and chiral-self-dual) but we begin with the three classical classes.  First, consider a system on a $d$-dimensional sphere (the case of a system on a torus or on other topologies is presented in a later section of the paper; the sphere case is the simplest to explain first).  That is, we imagine each lattice
site $i$ as located somewhere on the surface of a $d$-dimensional sphere, with coordinate $x_1(i),x_2(i),...,x_{d+1}(i)$.  We normalize the radius of the sphere to unity, so that $\sum_a x_a(i)^2=1$.
Let $X_1,...,X_{d+1}$ be coordinate
matrices.  These are diagonal matrices, with diagonal entries $(X_a)_{ii}=x_a(i)$.  In the GSE case, there
are two different spin-states per site, and so these coordinate matrices $X_a$ are self-dual matrices.  Thus,
\be
\sum_{a=1}^{d+1} X_a^2=I,
\ee
where $I$ is the identity matrix.
We then compute the projector $P$ onto the space of occupied states.  This is the space of eigenvalues of the Hamiltonian with
energy less than the Fermi energy, $E_F$.  We will see later that the properties of the system will depend on the value of $E_F$ chosen,
and we will observe phase transitions as a function of $E_F$, as the system changes from an ordinary insulator to a diffusive metal, to a topological insulator (experimental variation of the Fermi energy is possible via gating in some such systems\cite{konig}).

We then define a set of band projected matrices, $H_r$.  Let $PX_a P$ be the ``band projected position matrix''.
We write
\be
P X_a P = \begin{pmatrix} 0 & 0 \\ 0 & H_r
\end{pmatrix},
\ee
where the two blocks in the above matrix correspond to the space spanned by the kernel of $P$ and the range of $P$ (the space of empty states
and occupied states, respectively).
Thus, the size of the matrix $H_r$ is equal to the number of occupied states, which will be important in reducing the numerical effort later.
Note that if $H$ is GUE,GOE, or GSE, then the operators $H_r$ are also GUE, GOE, or GSE respectively, so that they inherit the symmetries of $H$.

Suppose the commutators $[P,X_a]$ are small.  This occurs if the Fermi energy $E_F$ is in a spectral gap or in a mobility gap.  In particular, to determine the size of the gap needed to make $[P,X_a]$ sufficiently small, we have to
consider the ratio between the range of the Hamiltonian and the system
size.
Above we have normalized distances so that $\sum_a x_a(i)^2=1$ and the
sphere has radius unity.  
In this case, if there are many sites on the sphere, then the distance between each site is small, tending to zero as the
number of sites tends to infinity.
For clarity in the present discussion, we prefer a more general normalization of distances, so that we can instead normalize the distance between sites to a constant, independent of the number of sites, and change the linear size of the system together with the number
of sites. 
To do this, let us instead normalize so that $\sum_a x_a(i)^2=L^2$, for
some length scale $L$, and let $H_r=P X_a P/L$, so that the previous discussion
corresponds to the choice $L=1$.  
Let us assume we have finite range interactions, so that $\Ham_{ij}$ vanishes if the distance between
$i$ and $j$ is larger than some interaction range $R$, which is held fixed for all system sizes (our numerical studies below correspond to
a choice $R=1$ since they involved nearest neighbor hopping, with the distance between sites normalized to $1$).  Let us assume that
$\Vert H \Vert \leq J$ for some constant $J$.  Then, if $E_F$ is
is separated by a gap $\Delta E$ from the spectrum of $H$, one can show that
the commutator
$\Vert [P,X_a] \Vert$ is bounded by a constant times $R J/(L \Delta E)$.

Since $[P,X_a]$ is bounded, 
the band projected position matrices almost commute with each other and almost square to the identity:
\be
\label{soft1}
\Vert [H_r,H_s] \Vert \approx 0,
\ee
\be
\label{soft2}
\sum_r H_r^2 \approx I.
\ee
We refer to this as a ``soft'' representation of the sphere $S^d$.
We sometimes quantify the approximation in the above equations, saying
that a set of matrices form a $\delta$-representation of the sphere if
\be
\Vert [H_r,H_s] \Vert \leq \delta,
\ee
\be
\Vert \sum_r H_r^2 -I \Vert \leq \delta.
\ee
One can show that we have a $\delta$-representation of the sphere with
$\delta$ of order\cite{hastingsloring}
\begin{eqnarray}
\label{deltabound}
\delta & \leq &
{\rm const.} \times (RJ/L\Delta E)^2 \\ \nonumber
& \equiv &
{\rm const.} \times (\vlr/L\Delta E)^2,
\end{eqnarray}
where the quantity $\vlr$ has units of velocity (we chose to introduce this
velocity because this allows one to also describe Hamiltonians which
do not have finite range interactions but instead have interactions that decay
sufficiently rapidly with distance; for such Hamiltonians one can prove the same bound with only a small amount of extra work).

A well-studied problem in $C^*$-algebra is whether
such a set of almost commuting matrices can be approximated by a set of exactly commuting matrices.  In the absence of symmetry, this problem
is completely understood in the case of the two-sphere and the two-torus.  The
answer is that the approximation is possible if and only
if a certain topological invariant, discussed in the next section,
vanishes.  This topological invariant is an integer, and may be (see
the section on mathematical problems for more detail on this question)
identified with the Hall conductance.  In other symmetry
classes and other dimensions, we have identified other topological
invariants, presented in the next section.  In one particular
case (the case of fermionic systems with charge
conservation but no other symmetries, which corresponds to the
case in which the $X_a$ are arbitrary Hermitian matrices with
no other symmetries), it is possible to determine the complete
set of topological invariants of matrices in the stable limit\cite{dadarlat}.

Our main claim is that the topological invariants of these almost commuting matrices can in general be identified with the topological
invariants of free Fermi system; we have not proven  this in all cases, but we have observed several cases fitting in the general pattern
and have a proof in some cases of this identification.

In this paper we begin by describing the relationship between the topological
invariants of the almost commuting matrices and the existence of Wannier functions
in the free fermion problem, and sketch the outline of our numerical procedure.  We then describe how to compute invariants of the
matrices for a spherical geometry in several cases using the so-called ``Bott
matrix''.
We summarize various mathematical questions, and then discuss how to compute invariants in other geometries.  We then describe chiral classes,
where invariants of almost commuting matrices can again be used to classify topological
phases, but where the procedure of constructing the almost commuting matrices
is different, so that the relevant matrices are not the band projected
position matrices.  We then present numerical results on a time-reversal invariant topological insulator with disorder in three
dimensions.  Finally, we give additional mathematical details on the K-theory
to compute invariants in more general cases, and we describe our numerical
implementation in detail.

\subsection{Wannier Functions and Trivial Hamiltonians}
\label{stableequiv}
The topological classification of free Fermi Hamiltonians we are interested in is a stable classification as discussed above.
The question to ask is whether a given gapped Hamiltonian $\Ham_0$ can be connected by a continuous path of gapped local Hamiltonians to
the trivial Hamiltonian, where the trivial Hamiltonian is a diagonal Hamiltonian so that in the trivial Hamiltonian each site
is either occupied or empty.  In the GSE case, we should instead have two states per site, corresponding to spin up and
down, with the two states having the same energy.  Then, such a diagonal Hamiltonian is a member of the appropriate symmetry
class, either GUE,GOE, or GSE.

We now show that this is equivalent to the question of whether $\Ham_0$ has localized
Wannier functions, where by localized Wannier functions we mean that we can find an orthonormal set of functions which are localized
in space and which span the range of the projector $P$ onto the occupied states of $\Ham_0$.  In the case of the GOE or GSE we require
that the Wannier functions respect the symmetry of the Hamiltonian.  Thus, the Wannier functions must be real in the GOE
case and the Wannier functions must occur in time-reversal invariant pairs in the GSE case.

Suppose a given Hamiltonian can be connected to a trivial Hamiltonian $\Ham_1$ by a smooth path of Hamiltonians $\Ham_s$.  Clearly, the trivial Hamiltonian has localized
Wannier functions: for each site on which the diagonal entry of $\Ham_1$ is negative, corresponding to an occupied site, we have one Wannier function localized on that site.  Let this set of Wannier functions be $\{v_a(1)\}$, where $a$ is a discrete index labeling the different
Wannier functions.  We now show that it is possible to construct a set of Wannier functions $\{v_a(s)\}$ for all Hamiltonians along the path, using
quasi-adiabatic continuation to continue the set of Wannier functions along the path\cite{fermicontinue}.  Define
\be
\partial_s v_a(s)=i{\mathcal D}_s v_a(s),
\ee
where ${\mathcal D}_s$ is a quasi-adiabatic continuation operator.  
This gives a set of Wannier
functions for all Hamiltonians along the path.
With an appropriate choice of quasi-adiabatic continuation operator, these Wannier functions are superpolynomially localized: for
all $s\in [0,1]$, the
amplitude of a given Wannier functions decays superpolynomially away from the site on which the function is localized at $s=1$.

Thus, if a Hamiltonian is connected to the trivial Hamiltonian, it has Wannier functions, so an obstruction to finding
localized Wannier functions implies an obstruction to continuing to the trivial Hamiltonian.  Conversely, if
a Hamiltonian $\Ham_0$ has a set of localized Wannier functions, $\{v_a(0)\}$, we can continue this Hamiltonian to the trivial
Hamiltonian as follows.
First, recalling that we are only interested in stable equivalence, add additional sites to the system, one for each Wannier
function (so that the number of added sites is equal to the number of occupied states).  Let these sites be decoupled from each other
and the rest of the Hamiltonian, with an energy $+1$ for each site, so that we consider the Hamiltonian
\be
\Ham_0\oplus I \equiv \begin{pmatrix} \Ham_0 \\ & I \end{pmatrix},
\ee
where the dimension of the second block is equal to the number of added sites.  Let $i_a$ denote the added site corresponding to
a given Wannier  function $v_a$, and let $|i_a\rangle$ denote the basis vector corresponding to this site.
Then smoothly continue the first block from $\Ham_0$ to the spectrally flattened Hamiltonian $1-2P$.  Then, note that
\be
P=\sum_a |v_a\rangle\langle v_a|.
\ee
Thus, the spectrally Hamiltonian is equal to $1-2\sum_a |v_a\rangle\langle v_a|$.  Follow the continuous path of Hamiltonians
\be
1-2\sum_a |x_a(\theta)\rangle\langle x_a(\theta)|,
\ee
where
\be
|x_a(\theta)\rangle=\cos(\theta) |v_a\rangle+\sin(\theta) |i_a\rangle.
\ee
from $\theta=0$ to $\theta=\pi/2$.  This gives a continuous path of gapped Hamiltonians to a trivial Hamiltonian, and
because the Wannier functions are localized, the Hamiltonians are local throughout (interaction terms in the Hamiltonian
decay superpolynomially if the Wannier functions decay superpolynomially).

Again, this emphasis on Wannier function is specific to these three classical ensembles.
Later we consider chiral ensembles.
In \cite{hastingsloring}, it is shown that if localized Wannier functions
exist, then it is possible to approximate the
matrices $H_r$ by exactly commuting Hermitian matrices.  Conversely, given the ability to
approximate the $H_r$ by exactly commuting Hermitian
matrices, one can define a set of Wannier functions, albeit ones which may be
localized in a weaker sense than exponential or
even than superpolynomial.

There are topological invariants that can be associated to
soft-representations of the sphere.  By calculating these, we
have a mechanism to prove a Hamiltonian $\mathcal {H} $ cannot
be deformed, via a path that stays gapped and local, to
a trivial Hamiltonian.  The chain of implication used in this
are the contrapositives of the down arrows in figure~\ref{flow1}.  
The invariants that we compute of the soft sphere will be invariant under
any continuous deformation of the matrices $H_r$ so long as the matrices
$H_r$ continue to prove a $\delta$-representation of the sphere for sufficiently
small $\delta$ ($\delta$ less than some numeric constant which depends upon dimension).  Using the relation between the energy gap and $\delta$, this proves
that if a Hamiltonian is continuously deformed to a trivial Hamiltonian
then the gap must
become of order $\vlr/L$ somewhere along the path.

That is,
we prove the existence of an obstruction to deforming some Hamiltonian to a
trivial Hamiltonian by computing some invariant of the matrices $H_r$.
This forms the basis of our numerical algorithm: we consider a given
Hamiltonian and choose a Fermi energy $E_F$.  We compute the projector
$P$ onto states with energy less than $E_F$ using standard techniques of
linear algebra (in section (\ref{numericssection}) we discuss how using symmetries such as time reversal
symmetry can improve the numerical accuracy of this calculation).  We then
compute invariants of those matrices; these invariants are elements of a group,
either $\mathbb{Z}$ or $\mathbb{Z}_2$, describing different invariants of
the original free fermion problem.

\begin{figure}
\label{flow1}
\centerline{
\includegraphics[scale=1.0]{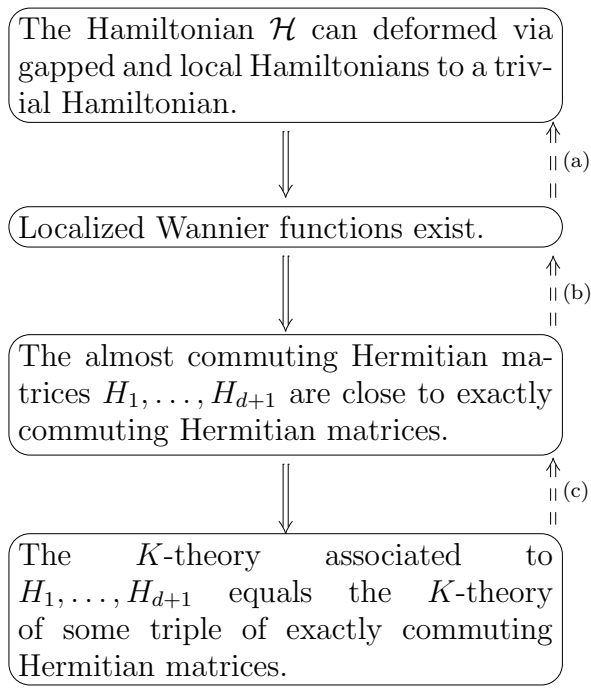}}
\caption{We indicate the directions of implication between various different mathematical concepts.  The bold arrows pointing
downward indicate directions of implication that we prove to be true as stated.  The dashed arrows indicate directions
of implication that are either not proven to be true as stated (but a slightly modified form is proven to be true), or are conjecture.
In particular, a) we prove that the existence of Wannier functions implies stable equivalence of Hamiltonians, a 
weaker statement, in subsection (\ref{stableequiv}).  (b) In \cite{hastingsloring}, we prove that the ability to approximate by exactly
commuting matrices implies some weaker statements about the localization properties of Wannier functions.  In particular, we showed
that this implies that there exist Wannier functions whose variance is small compared to the system size, but we did not prove the
existence of exponentially localized Wannier functions.
(c) This is proven (see \cite{hastingsloring} and references there) in the GUE case for $d=2$.  It is conjectured in the GOE/GSE case.  For the
case of more than 3 matrices (that is, for systems in more than  $2$ dimension), we can only hope to have stable approximation.  See problem (\ref{stableapproxprob}).}
\end{figure}

The absence of localized Wannier functions for systems in
topologically nontrivial phases represents a potential difficulty for
numerical simulation algorithms based on Wannier functions such as
those in \cite{quantumchem}.  It has also interesting implications for
the application of MERA (multi-scale entanglement renormalization
ansatz) schemes such as in \cite{evenblyvidal} based on a hierarchical
construction of wavefunctions for free fermion problems.  It is
claimed in \cite{evenblyvidal} that for gapped systems after a few
rounds of MERA for a gapped system, one converges to a fixed point
describing a product state.  However, such a fixed point corresponds
to a system which does have localized Wannier functions.  Thus, when
MERA is applied to a topologically insulator, one cannot converge to a
product state and must instead keep nontrivial entanglement at all
length scales; converges to such a product state fixed point happens
only for topologically trivial insulators.  We will present this elsewhere, showing that instead
one gets a structure similar to \cite{meratopo}, where a constant
number of degrees of freedom must be kept at each scale.

\subsection{Real, complex and self-dual matrices}

In the GOE and GSE cases, our method requires and benefits from preserving
the needed symmetry throughout the calculations. It is limiting in
the GOE case to work only with real matrices so
at times we use complex symmetric matrices.  For
example, we will consider a complex matrix that is symmetric
and unitary, so $U^{\mathrm{T}}=U$ and $U^{\dagger}=U^{-1}.$ Mathematical
readers should note that $U^{\mathrm{T}}$ denotes just the transpose,
$\overline{U}$ the conjugate, and $U^{\dagger}=\overline{U}^{\mathrm{T}}$the
adjoint.  When we extract the associate sin/cosine pair those matrices will
not just be Hermitian, but real symmetric.

In the GSE case, we need to specify the dual operation on matrices, and
make precise the meaning of ``time reversal invariant pairs'' of vectors.
The dual operation on $\mathbf{M}_{2N}(\mathbb{C})$ is only determined
up to a choice of a specified matrix $Z$ with the properties
\[
Z^{-1}=Z^{\dagger}=Z^{\mathrm{T}}=-Z,
\]
which means it is real orthogonal with eigenvalues $\pm i.$ Unless
specified otherwise, what we are using is 
$Z=\left(\begin{smallmatrix}0 & -I\\
I & 0\end{smallmatrix}\right)$. 
The \emph{dual} $X^{\sharp}$ of a matrix is then
\begin{equation}
X^{\sharp}=-ZX^{\mathrm{T}}Z.
\label{eq:defOfDual}
\end{equation}  
We use the operation $\sharp$ to denote the dual to use a symbol distinct from overline, dagger, star, and so on, which we use for other purposes.
In terms of $N$-by-$N$ blocks, 
\[
\left[\begin{array}{cc}
A & B\\
C & D\end{array}\right]=\left[\begin{array}{cc}
D^{\mathrm{T}} & -B^{\mathrm{T}}\\
-C^{\mathrm{T}} & A^{\mathrm{T}}\end{array}\right].
\]
We need the associated conjugate linear operation 
$\mathcal{T}:\mathbb{C}^{2N}\rightarrow\mathbb{C}^{2N}$
defined by 
\be
\label{defUpsion}
\mathcal{T}\mathbf{v}=-Z\overline{\mathbf{v}}.
\ee
For $\mathbf{v}$ any vector, $\mathcal{T}\mathbf{v}$
is orthogonal to $\mathbf{v}.$

A matrix $W$ is \emph{symplectic }if $W^{\mathrm{T}}ZW=Z.$ 

\begin{lem}
\label{lem:symplecticTests}
Suppose $U$ us a unitary matrix. The following are equivalent:
\begin{enumerate}
\item $U$ is symplectic;
\item $U^{\dagger}=U^{\sharp};$
\item \textup{$Z\overline{U}=UZ;$}
\item If $\mathbf{v}$ is column $j$ of $U$ for $j\leq N$ then column
$N+j$ of $U$ is $\mathcal{T}\mathbf{v}.$
\end{enumerate}
\end{lem}

\begin{lem}
If $U$ is a symplectic unitary and $X$ and $U$ are 
in $\mathbf{M}_{2N}(\mathbb{C})$
then
\[
\left(U^{\dagger}XU\right)^{\sharp}
=U^{\dagger}X^{\sharp}U.
\]
\end{lem}

For low dimensional GSE systems, we derive {\em real} matrices whose
invariants are equivalent to invariants of self-dual matrices.  
In dimension $4$ for a GOE system, we do a conversion the other way,
so end up with {\em self-dual} matrices that encode the invariants
of real matrices. This
is part of how we take advantage of Bott-periodicity in $K$-theory,
avoiding homotopy calculations involving $4$-spheres or $5$-spheres.

The trick we use is classic.  Tensoring together two dual operations
leads to an operation that is equivalent to the transpose.  This has
a simple physical interpretation: the self-dual operation is the time
reversal operation on a system with half-odd-integer spin.  Tensoring
together two systems with half-odd-integer spin gives a system
with integer spin, and the time reversal operation on such a system is
equivalent to the transpose, after an appropriate change of basis.
The next technical lemma specifies the appropriate basis change to make
the time reversal into the transpose.

\begin{lem}
\label{QuaternionsSquared}
Consider
\[
U=\frac {1}{\sqrt {2}}\left ( I - i Z \otimes Z' \right).
\]
For all $X\in\mathbf{M}_{2N}(C)$ and $Y\in\mathbf{M}_{2m}(\mathbb{C}),$
\[
U^\dagger \left (X^{\sharp}\otimes Y^\sharp\right ) U 
=\left ( U^\dagger \left ( X \otimes Y \right ) U\right ) ^{\mathrm{T}}
\]
where 
$\left(A\otimes B\right)^{\mathrm{T}}
=A^{\mathrm{T}}\otimes B^{\mathrm{T}},$
i.e.~is the usual transpose of $4N$-by-$4N$ matrices.
Here $Z$ and $Z'$ are the matrices, of sizes $2N$-by-$2N$ and $2m$-by-$2m$, that
define the two dual operations.
\end{lem}

\begin{proof}
First note that $U$ is unitary and $U = U^{\mathrm{T}}$

Since
\[
UU^{\mathrm{T}}
=\frac {1}{2}\left ( I\otimes I - i Z \otimes Z' \right )^2
=i Z \otimes Z'
\]
and, by conjugation,
$\overline{U}U^\dagger = -i Z \otimes Z' ,$
we find
\begin{align*}
X^{\sharp}\otimes Y^{\sharp}
& =\left(Z\otimes Z'\right)\left(X^{\mathrm{T}}\otimes Y^{\mathrm{T}}\right)\left(Z\otimes Z' \right)\\
& = UU^{\mathrm{T}} \left(X^{\mathrm{T}}\otimes Y^{\mathrm{T}}\right) \overline{U}U^{\dagger}.
\end{align*}
Therefore
\[
U^\dagger\left(X^{\sharp}\otimes Y^{\sharp}\right)U
=U^{\mathrm{T}}\left(X^{\mathrm{T}}\otimes Y^{\mathrm{T}}\right)\overline{U}
=\left(U^{\dagger}\left(X\otimes Y\right)U\right)^{\mathrm{T}}.
\]

\end{proof}

\begin{rem}
We will see later that this lemma really has proven the isomorphism
\[
\mathbf{M}_{N}(\mathbb{H})\otimes\mathbf{M}_{m}(\mathbb{H})
\cong
\mathbf{M}_{2N+2m}(\mathbb{R})
\]
where $\mathbb{H}$ is the algebra of quaternions, with is a
real $C^*$-algebra. In contrast, the next lemma is a
tautology, but a good place to
discuss a critical conventions.
\end{rem}

\begin{lem}
\label{lem:TransposeDual} 
For all $X\in\mathbf{M}_{2N}(\mathbb{C})$
and\textup{ $Y\in\mathbf{M}_{2m}(\mathbb{C}),$} 
\[
X^{\mathrm{T}}\otimes Y^{\sharp}=\left(X\otimes Y\right)^{\sharp}.
\]
\end{lem}

\begin{proof}
We adopt the convention on identifying matrices in $\mathbf{M}_{j}(\mathbb{C})\otimes\mathbf{M}_{k}(\mathbb{C})$
with matrices in $\mathbf{M}_{j+k}$ that makes the association 
\[
T \otimes 
\left[\begin{array}{cc}
a & b\\
c & d\end{array}\right]
\leftrightsquigarrow
\left[\begin{array}{cc}
aT & bT\\
cT & dT\end{array}\right].
\]
when $k=2.$ \emph{Using the other convention can change the argument
of the Pfaffian}, which is the only interesting thing about the Pfaffian.
(The Pfaffian's magnitude is the the square root of the magnitude
of the determinant.) With this convention, 
\[
I\otimes Z\leftrightsquigarrow Z^{\prime}
\]
where the $Z$ are of appropriate size, and so
\[
X^{\mathrm{T}}\otimes Y^{\sharp}
=-\left(I\otimes Z\right)\left(X^{\mathrm{T}}\otimes Y^{\mathrm{T}}\right)\left(I\otimes Z\right)
=-\left(I\otimes Z\right)\left(X\otimes Y\right)^{\mathrm{T}}\left(I\otimes Z\right)
\]
which, under the chosen identification, says 
\[
X^{\mathrm{T}}\otimes Y^{\sharp}=\left(X\otimes Y\right)^{\sharp}.
\]

\end{proof}

\section{Bott periodicity and soft representations of the zero sphere}

\subsection{GUE in all dimensions}
To construct the topological invariants, we define
an operator
\be
\label{Sdef}
B(H_1,...,H_{d+1})=\sum_{r=1}^{d+1} H_r \otimes \gamma_r,
\ee
where the $\gamma_r$ are a set of anti-commuting Hermitian matrices: $\{\gamma_r,\gamma_b\}=0$.
In the GUE case, then we simply choose the $\gamma_r$ to be a set of matrices
of the minimal dimension to provide $d+1$ different $\gamma$-matrices.  That is,
for $d=2$, we can choose the $\gamma$ matrices to be the three different $2$-by-$2$ Pauli spin matrices, while for $d=3,4$, the $\gamma$ matrices need to be
at least $4$-dimensional.  In the GOE,GSE cases, we will choose the $\gamma$ matrices as described later to make certain symmetries of $B(H_1,...,H_{d+1})$ more apparent.

Assuming the matrices $H_r$ are a soft representation of the sphere, as in (\ref{soft1},\ref{soft2}), then
$B$ is a soft representation of the sphere $S^0$.  This simply means that
\be
B^2\approx I \mbox{ and } B^\dagger = B.
\ee
Thus, the eigenvalues of $B$ are close to plus or minus one.

The integer invariant we consider in the GUE case is one-half the difference between
the number of positive and negative eigenvalues of this matrix, which for consistency with the GSE and GOE cases we give a second name,
\[
B(H_1,...,H_{d+1}) = B(H_1,...,H_{d+1}).
\]
When there are no zero eigenvalues we call this quantity
$\bott(H_1,...,H_{d+1})$.
This index, called the Bott index, is a topological
invariant, in that it does not change along any path of matrices $H_r$
which form a $\delta$-representation of the sphere for sufficiently small
$\delta$ (the value of $\delta$ required depends upon dimension; for $d=2$ we
need $\delta<1/4$).  This is proven by showing that the only way for
the Bott index to change is for an eigenvalue of $B=B(H_1,...,H_{d+1})$ to become equal to zero.  If
this happens, then $\Vert B^2-I \Vert=1$.
However, if a set of $H_r$ form a $\delta$-representation of the
sphere for sufficiently small $\delta$, then $\Vert B^2-I \Vert \leq 1$, giving
a contradiction.  
One can similarly show that given two different tuples of matrices $H_r$ and $K_r$ which both form $\delta$-representations of the sphere, then if
$\Vert H_r-K_r \Vert$ is sufficiently small then $\bott(H_1,..,H_{d+1})=\bott(K_1,..,K_{d+1})$ (generalizing lemma 3.5 of \cite{hastingsloring}) by considering
a linear path $(1-t) H_r+tK_r$ for $t\in [0,1]$:

\begin{lem}
Suppose 
$\left(H_1, ..., H_{d+1} \right)$ and
$\left(K_1, ..., K_{d+1} \right)$ are tuples of
self-dual, Hermitian $n$-by-$n$ matrices and suppose
$\left(H_1, ..., H_{d+1} \right)$ is a 
$\delta$-representation
of the sphere with $\delta\leq 1/n$, where
$n$ is defined to be $d(d+1)/2+1$.
If
\[
\sum_{r=1}^{d+1} \left\Vert H_r-K_r\right\Vert \leq \sqrt{1-n\delta}
\]
then 
\[
\bott(K_1,...,K_{d+1})=
\bott(H_1,...,H_{d+1})
\].
\end{lem}

Further, if
\begin{lem}
If $H_1,...,.H_{d+1}$ are exactly commuting and $\bott(H_1,...,H_{d+1})$ is defined, then
$\bott(H_1,...,H_{d+1})=0$.
\end{lem}
These last two lemmas establish that if the Bott index is nontrivial and $\delta$ is sufficiently small then the given tuple of matrices $H_r$ is not close to an exactly commuting tuple.

Note that the Bott index is always equal to zero in the case that $d$ is odd, since the matrix $I \otimes (\gamma_1 \gamma_2 ... \gamma_{d+1})$ then anti-commutes with $B$.

\subsection{2D GSE}
We now consider the case of problems with time-reversal symmetry, where
the Hamiltonian is in the GSE universality class.
In this case, the presence of time-reversal implies that the Bott index above defined in the GUE class is trivial.
However, there is another non-trivial index.  This index represents an obstruction to finding localized Wannier
functions which respect time-reversal symmetry.

The procedure will be to show that in two dimensions, we can construct a unitary transformation that makes $B(H_1,H_2,H_3)$ an
anti-symmetric matrix.  Then, the invariant that we consider is the Pfaffian of this matrix.
The procedure in three dimensions is to construct a unitary transformation that makes $B(H_1,H_2,H_3,H_4)$ a real symmetric chiral 
matrix.  Such a matrix is of the form
\be
\label{chiralorthogdef}
\begin{pmatrix} 0 & A \\ A^T & 0\end{pmatrix}.
\ee
If $B^2 \approx I$, then $A$ is approximately orthogonal.  The invariant that we consider in this case is the determinant of this
matrix.
We notice a general pattern here: starting with a certain number of matrices (3 or 4) in a given symmetry class (GSE) we construct
a single matrix in another symmetry class (anti-symmetric or chiral real symmetric, respectively).  
See table 4 in \cite{ludwig}, where the symmetry classes are arranged in a sequence given by Bott periodicity.

We begin with the two dimensional case, reviewing the construction of \cite{hastingsloring}, used in
\cite{loringhastings}.
From a physical point of view, the
existence of a unitary transformation that makes $B$ anti-symmetric by is not surprising: the self-dual
operation can be regarded as a time-reversal symmetry operation, and a
similar time-reversal symmetry can be applied to the $\gamma$ matrices
used to construct $B(H_1,H_2,H_3)$.  We choose a time-reversal symmetry operation
that makes those $\gamma$ matrices odd under time-reversal.  Then, 
under these combined time reversal
symmetries, $B(H_1,H_2,H_3)$ changes sign; however, since there are
two spin-$1/2$s, the time reversal symmetry operator squares to unity and hence,
up to a basis change, is equivalent to transposition.  This is simply lemma \ref{QuaternionsSquared} above, as
we will see.

\begin{defn}
Let $H_r$ be self-dual and Hermitian.  Define the matrix
$\bD(H_1,H_2,H_3)$ by
\be
\bD(H_1,H_2,H_3)=  U^\dagger B(H_1,H_2,H_3) U,
\ee
where the unitary $U$ is defined by
\begin{eqnarray}
U=\frac{1}{\sqrt{2}} (I+Z \otimes \sigma_2).
\end{eqnarray}
\end{defn}

This matrix $\bD$ will be anti-symmetric, and so $i$ times
a real matrix, by lemma~\ref{QuaternionsSquared}, since
$H_r^\sharp = H_r$ while $\sigma_r^\sharp = -\sigma_r.$  
Still following \cite{hastingsloring},
we now define the index by taking the Pfaffian of this matrix $\tilde{B}.$
Since $\tilde{B}$ is Hermitian and pure imaginary so
has real eigenvalues that occur in pairs, symmetric across zero.  Thus
its determinant is nonnegative and its Pfaffian is real.

\begin{defn}
We define the index $\bI(H_1,H_2,H_3)$ for self-dual matrices $H_r$ by
\be
\bI(H_1,H_2,H_3)={\rm sgn}(
\Pf(\bD(H_1,H_2,H_3))),
\ee
where $\Pf$ is the Pfaffian and ${\rm sgn}(x)=1$ for $x>0$ and ${\rm sgn}(x)=-1$ 
for $x<0$.  If $\Pf(\bD(H_1,H_2,H_3))=0$, the index
$\bI(H_1,H_2,H_3)$ is not defined.
\end{defn}

In \cite{hastingsloring}, it is also shown that this index is a topological
invariant, in analogy to the GUE case:
\begin{lem}
\label{lem:pathstable}
Consider any continuous path of self-dual matrices, $H_r(s)$, where $s$ is a real number, $0\leq s \leq 1$.  Suppose that for all $s$,
the matrix $B(H_1,H_2,H_3)$ has non-vanishing determinant.  Then, $\bI(H_1(0),H_2(0),H_3(0))=\bI(H_1(1),H_2(1),H_3(1))$.
\begin{proof}
The determinant of $B(H_1,H_2,H_3)$ is equal to $\Pf(\bD(H_1,H_2,H_3))^2$.  As long as the determinant does not vanish, the Pfaffian
does not vanish and hence cannot change sign.
\end{proof}
\end{lem}
and
\begin{lem}
If $H_1,H_2,H_3$ are self-dual and exactly commuting and $\bI(H_1,H_2,H_3)$ is defined, then
$\bI(H_1,H_2,H_3)=1$.
\end{lem}
and
\begin{lem}
\label{lem:indexDIsStable}
Suppose 
$\left(H_1, H_2, H_3 \right)$ and
$\left(K_1, K_2, K_3 \right)$ are triples of
self-dual, Hermitian $n$-by-$n$ matrices and suppose
$\left(H_1, H_2, H_3 \right)$ is a 
$\delta$-representation
of the sphere with $\delta<1/4$.
If
\[
\left\Vert H_1-K_1\right\Vert 
+\left\Vert H_2-K_2\right\Vert
+\left\Vert H_3-K_3\right\Vert 
\leq 
\sqrt{1-4\delta}
\]
then 
\[
\bI(K_1,K_2,K_3)=
\bI(H_1,H_2,H_3).
\]
\end{lem}

This last lemma justifies calling this index an invariant.

Given a soft-representation of the two-sphere with arbitrary
complex matrices we can double these with their transposes to
get self-dual soft-representation.  The resulting $Z_2$ index
is determined by the parity of the original $A$ index. 

\begin{thm}
\label{doubledparity}
For $H_{1},H_{2},H_{3}$ a soft representation of the two-sphere,
in $\mathbf{M}_{N}(\mathbb{C}),$
\[
\left[\begin{array}{cc}
H_{1} & 0\\
0 & \overline{H_{1}}\end{array}\right],\left[\begin{array}{cc}
H_{2} & 0\\
0 & \overline{H_{2}}\end{array}\right],\left[\begin{array}{cc}
H_{3} & 0\\
0 & \overline{H_{3}}\end{array}\right]
\]
is a self-dual soft representation of the two-sphere and 
\[
\widetilde{\mathrm{bott}}\left(\left[\begin{array}{cc}
H_{1} & 0\\
0 & \overline{H_{1}}\end{array}\right],\left[\begin{array}{cc}
H_{2} & 0\\
0 & \overline{H_{2}}\end{array}\right],\left[\begin{array}{cc}
H_{3} & 0\\
0 & \overline{H_{3}}\end{array}\right]\right)
=(-1)^{\mathrm{Bott}(H_{1},H_{2},H_{3})}.
\]
\end{thm}

\begin{proof}
Recall
\[
U=\frac{1}{\sqrt{2}}\left(I\otimes I+\sigma_{2}\otimes Z\right)
=\frac{1}{\sqrt{2}}
\left[\begin{array}{cccc}
I & 0 & 0 & -iI\\
0 & I & iI & 0\\
0 & iI & I & 0\\
-iI & 0 & 0 & I\end{array}\right]
\]
and the sign of $\mathrm{Pf}\left(U^{\dagger}BU\right)$ gives the
Pfaffian-Bott index, where
\[
B=B\left(\left[\begin{array}{cc}
H_{1} & 0\\
0 & H_{1}^{\mathrm{T}}\end{array}\right],\left[\begin{array}{cc}
H_{2} & 0\\
0 & H_{2}^{\mathrm{T}}\end{array}\right],\left[\begin{array}{cc}
H_{3} & 0\\
0 & H_{3}^{\mathrm{T}}\end{array}\right]\right).
\]
Let 
\[
Q=\frac{1}{\sqrt{2}}\left[\begin{array}{cccc}
I & 0 & -I & 0\\
0 & iI & 0 & iI\\
0 & I & 0 & -I\\
-iI & 0 & -iI & 0\end{array}\right]
\]
Notice 
\[
\det(Q)=\det\left(\frac{1}{\sqrt{2}}\left[\begin{array}{cc}
I & -I\\
iI & iI\end{array}\right]\right)\det\left(\frac{1}{\sqrt{2}}\left[\begin{array}{cc}
iI & iI\\
I & -I\end{array}\right]\right)=1
\]
so
$
\mathrm{Pf}\left(U^{\dagger}BU\right)
=\mathrm{Pf}\left(Q^{\mathrm{T}}U^{\dagger}BUQ\right).
$
Since
\[
UQ=\left[\begin{array}{cccc}
0 & 0 & -I & 0\\
0 & iI & 0 & 0\\
0 & 0 & 0 & -I\\
-iI & 0 & 0 & 0\end{array}\right]\]
and
\[
Q^{\mathrm{T}}U^{\dagger}=\left[\begin{array}{cccc}
I & 0 & 0 & 0\\
0 & 0 & I & 0\\
0 & 0 & 0 & -iI\\
0 & iI & 0 & 0\end{array}\right]
\]
we find
\begin{align*}
Q^{\mathrm{T}}U^{\dagger}BUQ & =Q^{\mathrm{T}}U^{\dagger}\left[\begin{array}{cccc}
H_{3} & 0 & H_{1}-iH_{2} & 0\\
0 & H_{3}^{\mathrm{T}} & 0 & H_{1}^{\mathrm{T}}-iH_{2}^{\mathrm{T}}\\
H_{1}+iH_{2} & 0 & -H_{3} & 0\\
0 & H_{1}^{\mathrm{T}}+iH_{2}^{\mathrm{T}} & 0 & -H_{3}^{\mathrm{T}}\end{array}\right]UQ\\
 & =\left[\begin{array}{cccc}
H_{3} & 0 & H_{1}-iH_{2} & 0\\
H_{1}+iH_{2} & 0 & -H_{3} & 0\\
0 & -iH_{1}^{\mathrm{T}}+H_{2}^{\mathrm{T}} & 0 & iH_{3}^{\mathrm{T}}\\
0 & iH_{3}^{\mathrm{T}} & 0 & iH_{1}^{\mathrm{T}}+H_{2}^{\mathrm{T}}\end{array}\right]UQ\\
 & =\left[\begin{array}{cccc}
0 & 0 & -H_{3} & -H_{1}+iH_{2}\\
0 & 0 & -H_{1}-iH_{2} & H_{3}\\
H_{3}^{\mathrm{T}} & H_{1}^{\mathrm{T}}+iH_{2}^{\mathrm{T}} & 0 & 0\\
H_{1}^{\mathrm{T}}-iH_{2}^{\mathrm{T}} & -H_{3}^{\mathrm{T}} & 0 & 0\end{array}\right]
\end{align*}
and so
\[
\mathrm{Pf}\left(U^{\dagger}BU\right)=\det(-I)\det\left[\begin{array}{cc}
H_{3} & H_{1}-iH_{2}\\
H_{1}+iH_{2} & -H_{3}\end{array}\right]=\det S\left(H_{1},H_{2},H_{3}\right).
\]
If $\mathrm{bott}(H_{1},H_{2},H_{3})=m$ then the spectrum of the
$2N$-by-$2N$ matrix $B\left(H_{1},H_{2},H_{3}\right)$ will have
$2N-m$ negative eigenvalues, and so the 
sign of $\mathrm{Pf}\left(U^{\dagger}BU\right)$
will be negative exactly when $m$ is odd.
\end{proof}

\subsection{3D/4D GSE and 4D/6D/7D GOE}

We now turn to the three-dimensional case.
In this case, we choose a particular representation of the $\gamma$ matrices as
\begin{eqnarray}
\gamma_1 &=& I \otimes \sigma_x, \\ \nonumber
\gamma_2 &=& \sigma_x \otimes \sigma_y, \\ \nonumber
\gamma_3 &=& \sigma_y \otimes \sigma_y, \\ \nonumber
\gamma_4 &=& \sigma_z \otimes \sigma_y,
\end{eqnarray}
where $\sigma_r$ are the Pauli spin matrices.  This gives $4$-by-$4$ $\gamma$ matrices.
Writing $B=\sum_r H_r \otimes \gamma_r$, we have $4$ different matrices
textured together: the orbital and spin degrees of freedom of $H_r$, and the
two two-dimensional spaces used to define $\gamma_r$.

The existence of a transformation making $B$ chiral and real is not so surprising from the following physical point of view, again using
lemma\ref{QuaternionsSquared} and, similar to the two dimensional case, using the trick of picking a set of $\gamma$ matrices with
appropriate behavior under an appropriate duality operation.
The $H_r$ are
self-dual, so that $H_r^T=-Z H_r Z$. Now, consider the matrix $Z' = i\sigma_y \otimes  I$. The $\gamma_r$ are self-dual using $Z'$ to define
the duality: $\gamma_r^T = -Z' \gamma_r Z'$. The matrix $B$ is then a sum of tensor products of two self-dual matrices, $H_r$ and $\gamma_r$.  So, $B$ should be symmetric under a time reversal operation. However, given two spin-$1/2$s (one spin-$1/2$ corresponding to the spin degree of freedom of the $H_r$ and the other being the first of the two sigma matrices used to define the $\gamma_r$), we form a system with spin $0$ or spin $1$, so up to a unitary the time-reversal operation is the same as transposition. This means that unitarily conjugating $B$ gives us a symmetric matrix, and symmetric plus Hermitian means real symmetric.  Next will then see why it is chiral.

We see the change of basis to make this in in lemma~\ref{QuaternionsSquared}.
By that lemma,
\[ U^\dagger B(H_1,H_2,H_3,H_4) U,
\]
where the unitary $U$ is defined by
\begin{eqnarray}
U=\frac{1}{\sqrt{2}} (I-iZ \otimes Z').
\end{eqnarray}
is real symmetric.

We now show that this is chiral in an appropriate basis.  First
\begin{lem}
The matrix $U^\dagger B(H_1,H_2,H_3,H_4) U$ anti-commutes with the matrix $I\otimes I \otimes I \otimes \sigma_z$, where
the first two matrices refer to the orbital and spin degrees of freedom used to define $H_r$, and the last
two refer to the two two-dimensional space used to define $\gamma_r$.
\begin{proof}
The matrix
$I\otimes I \otimes I \otimes \sigma_z$ anti-commutes with $U$, so we must show that $B$ anti-commutes
with
$I\otimes I \otimes I \otimes \sigma_z$.  However, this follows since 
$I\otimes I \otimes I \otimes \sigma_z$ anti-commutes with all the matrices $\gamma_r$.
\end{proof}
\end{lem}
Thus, if we write $U^\dagger B(H_1,H_2,H_3,H_4) U$ as a block matrix, with the two blocks corresponding to the
positive and negative eigenvalues of
$I\otimes I \otimes I \otimes \sigma_z$, we have a symmetric chiral real matrix.

For numerical purposes, of course, it is convenient just to compute the upper-right-hand block  of
Eq.~(\ref{chiralorthogdef}), since the lower-left hand block is related by transposition. This block can be expressed as
\[
\frac{1}{2}(I-i\sigma_y) H_1 (I+i\sigma_y)+i\bD(H_2,H_3,H_4),
\]
and it is essentially this block that we take as $\tilde{B}$ in this case.
To be consistent with what we coded, we define this as follows.
\begin{defn}
Let $H_r$ be four self-dual and Hermitian matrices.  Define the matrix
\be
\bD(H_1,H_2,H_3,H_3)=
U^\dagger \left( \sum_{r=1}^4 H_r \otimes \nu_r \right) U,
\ee
where the unitary $U$ is defined by
\begin{eqnarray}
U=\frac{1}{\sqrt{2}} (I+Z \otimes \sigma_2).
\end{eqnarray}
and
\be
\nu_1 = I,\ \nu_2 = i\sigma_x,\ \nu_3 = i\sigma_y,\ \nu_4 = i\sigma_z.\  
\ee
\end{defn}

We now define the index for this problem as:
\begin{defn}
We define the index $\bI(H_1,H_2,H_3,H_4)$ for self-dual matrices $H_r$ by
\be
\bI(H_1,H_2,H_3,H_4)={\rm sgn}(
{\rm det}(\bD(H_1,H_2,H_3,H_4))).
\ee
\end{defn}

We summarize in
Table~\ref{tab:GSE-reps} our method to deal
with approximate
representations of $S^{d}$ by self-dual matrices, meaning
$H_r^\dagger=H_r,$ $H_r^\sharp = H_r$
and 
\[
{ \sum_{r=1}^{d+1}H_r \approx I},
\]
in dimensions $2,$ $3$ and $4.$ In all cases,
\[
U=\frac {1}{\sqrt{2}} \left( I+Z \otimes Z^\prime \right )
\]
where
\[
Z^\prime = 
\left [
\begin{array}{cc}
0 & I\\
-I & 0
\end{array}
\right]
\]
in the appropriate size. 
Table~\ref{tab:GOE-reps} deals with  GOE case.
Notice that in some dimensions the Bott matrix $B$ has
different symmetries than the original matrix, so
real $H_r$ can lead to self-dual $B.$
Tables \ref{tab:KtheoryReal} and \ref{tab:KtheoryQuaternions} show
how these relations lead to $K$-theory elements for
an algebra of matrices over the reals of the quaternions.

Table~\ref{tab:GUE-reps} shows the way to deal with approximate
representations of $S^{d}$ in the GUE case, in dimension $2$
and higher even dimensions.
Table~\ref{tab:K-theoryComplexCase}
shows how these relations lead to $K$-theory elements.

\begin{table}
\begin{tabular}{|>{\centering}m{0.4in}|>{\centering}m{1in}|>{\centering}m{2.7in}|>{\centering}m{0.7in}|>{\centering}m{0.9in}|}
\hline 
\noalign{\vskip\doublerulesep}
 & Relations on $\nu_{1},\ldots,\nu_{d+1}$ & Definition of $B=B(H_{1},\ldots,H_{d+1})$ and choices for $\nu_{1},\ldots,\nu_{d+1}$ & Relations on $B$ & Scalar valued function\tabularnewline[\doublerulesep]
\hline
\hline 
\noalign{\vskip\doublerulesep}
all $d$ & $\begin{array}{c}
\nu_{r}^{\dagger}=\nu_{r}^{-1}\\
\nu_{r}^{\dagger}\nu_{s}=-\nu_{s}^{\dagger}\nu_{r}\end{array}$\\
for $r\neq s.$  & $B=\left({\displaystyle \sum_{r=1}^{d+1}H_{r}\otimes\nu_{r}}\right)$ & $B^{\dagger}B\approx I$ & \tabularnewline[\doublerulesep]
\hline 
\noalign{\vskip\doublerulesep}
$d=2$ & $\begin{array}{c}
\nu_{r}^{\dagger}=\nu_{r}\end{array}$ & $\sigma_{x},\sigma_{y},\sigma_{z}$ & $\begin{array}{c}
B^{\mathrm{\dagger}}=B\end{array}$ & $\frac{1}{2}$Signature\tabularnewline[\doublerulesep]
\hline 
\noalign{\vskip\doublerulesep}
$d=4$ & $\begin{array}{c}
\nu_{r}^{\dagger}=\nu_{r}\end{array}$ & $\left[\begin{array}{cc}
\sigma_{x} & 0\\
0 & \sigma_{x}\end{array}\right],\left[\begin{array}{cc}
\sigma_{y} & 0\\
0 & -\sigma_{y}\end{array}\right],\left[\begin{array}{cc}
\sigma_{z} & 0\\
0 & \sigma_{z}\end{array}\right],$ ~~ $\left[\begin{array}{cc}
0 & -i\sigma_{y}\\
i\sigma_{y} & 0\end{array}\right],\left[\begin{array}{cc}
0 & \sigma_{y}\\
\sigma_{y} & 0\end{array}\right]$ & $\begin{array}{c}
B^{\mathrm{\dagger}}=B\end{array}$ & $\frac{1}{2}$Signature\tabularnewline[\doublerulesep]
\hline 
\noalign{\vskip\doublerulesep}
$d=6$ & $\begin{array}{c}
\nu_{r}^{\dagger}=\nu_{r}\end{array}$ & --- & $\begin{array}{c}
B^{\mathrm{\dagger}}=B\end{array}$ & $\frac{1}{2}$Signature\tabularnewline[\doublerulesep]
\hline 
\noalign{\vskip\doublerulesep}
$d=8$ & $\begin{array}{c}
\nu_{r}^{\dagger}=\nu_{r}\end{array}$ & --- & $\begin{array}{c}
B^{\mathrm{\dagger}}=B\end{array}$ & \multicolumn{1}{c|}{$\frac{1}{2}$Signature}\tabularnewline[\doublerulesep]
\hline
\end{tabular}
\caption{GUE approximate representations of $S^{d}.$ The signature of $B$
leads to an integer invariant. \label{tab:GUE-reps}}
\end{table}

\begin{table}
\begin{tabular}{|>{\centering}m{0.4in}|>{\centering}m{1in}|>{\centering}m{2.7in}|>{\centering}m{0.7in}|>{\centering}m{0.9in}|}
\hline 
\noalign{\vskip\doublerulesep}
 & Relations on $\nu_{1},\ldots,\nu_{r}$ & Definition of $B(H_{1},\ldots,H_{d+1})$ and choices for $\nu_{1},\ldots,\nu_{d+1}$ & Relations on $B$ & Scalar valued function\tabularnewline[\doublerulesep]
\hline
\hline 
\noalign{\vskip\doublerulesep}
all $d$ & $\begin{array}{c}
\nu_{r}^{\dagger}=\nu_{r}^{-1}\\
\nu_{r}^{\dagger}\nu_{s}=-\nu_{s}^{\dagger}\nu_{r}\end{array}$\\
for $r\neq s.$  & $B=U^{*}\left({\displaystyle \sum_{r=1}^{d+1}H_{r}\otimes\nu_{r}}\right)U$ & $B^{\dagger}B\approx I$ & \tabularnewline[\doublerulesep]
\hline 
\noalign{\vskip\doublerulesep}
$d=2$ & $\begin{array}{c}
\nu_{r}^{\dagger}=\nu_{r}\\
\nu_{r}^{\sharp}=-\nu_{r}\end{array}$ & $\sigma_{x},\sigma_{y},\sigma_{z}$ & $\begin{array}{c}
B^{\mathrm{\dagger}}=B\\
B^{\mathrm{T}}=-B\end{array}$ & Pfaffian\tabularnewline[\doublerulesep]
\hline 
\noalign{\vskip\doublerulesep}
$d=3$ & $\begin{array}{c}
\nu_{r}^{\sharp}=-\nu_{r}^{\dagger}\end{array}$ & $I,i\sigma_{x},i\sigma_{y},i\sigma_{z}$ & $\begin{array}{c}
B^{\mathrm{\dagger}}=B^{\mathrm{T}}\end{array}$ & Determinant\tabularnewline[\doublerulesep]
\hline 
\noalign{\vskip\doublerulesep}
$d=4$ & $\begin{array}{c}
\nu_{r}^{\dagger}=\nu_{r}\\
\nu_{r}^{\sharp}=\nu_{r}\end{array}$ & $\left[\begin{array}{cc}
\sigma_{x} & 0\\
0 & \sigma_{x}\end{array}\right],\left[\begin{array}{cc}
\sigma_{y} & 0\\
0 & -\sigma_{y}\end{array}\right],\left[\begin{array}{cc}
\sigma_{z} & 0\\
0 & \sigma_{z}\end{array}\right],$ ~~ $\left[\begin{array}{cc}
0 & -i\sigma_{y}\\
i\sigma_{y} & 0\end{array}\right],\left[\begin{array}{cc}
0 & \sigma_{y}\\
\sigma_{y} & 0\end{array}\right]$ & $\begin{array}{c}
B^{\mathrm{\dagger}}=B\\
B^{\mathrm{T}}=B\end{array}$ & $\frac{1}{2}$Signature\tabularnewline[\doublerulesep]
\hline 
\noalign{\vskip\doublerulesep}
$d=8$ & $\begin{array}{c}
\nu_{r}^{\dagger}=\nu_{r}\\
\nu_{r}^{\mathrm{T}}=\nu_{r}\end{array}$ & --- & $\begin{array}{c}
B^{\mathrm{\dagger}}=B\\
B^{\sharp}=B\end{array}$ & $\frac{1}{4}$Signature\tabularnewline[\doublerulesep]
\hline
\end{tabular}
\caption{GSE approximate representations of $S^{d}.$ The signature of $B$
leads directly to an integer invariant, while we must take the sign
of the determinant and Pfaffian to get $\pm1$ and so a $Z_{2}$ invariant.
\label{tab:GSE-reps}}
\end{table}

\begin{table}
\begin{tabular}{|>{\centering}m{0.4in}|>{\centering}m{1in}|>{\centering}m{2.7in}|>{\centering}m{0.7in}|>{\centering}m{0.9in}|}
\hline 
\noalign{\vskip\doublerulesep}
 & Relations on $\nu_{1},\ldots,\nu_{r}$ & Definition of $B(H_{1},\ldots,H_{d+1})$ and choices for $\nu_{1},\ldots,\nu_{d+1}$ & Relations on $B$ & Scalar valued function\tabularnewline[\doublerulesep]
\hline
\hline 
\noalign{\vskip\doublerulesep}
all $d$ & $\begin{array}{c}
\nu_{r}^{\dagger}=\nu_{r}^{-1}\\
\nu_{r}^{\dagger}\nu_{s}=-\nu_{s}^{\dagger}\nu_{r}\end{array}$\\
for $r\neq s.$  & $B={\displaystyle \sum_{r=1}^{d+1}H_{r}\otimes\nu_{r}}$ & $B^{\dagger}B\approx I$ & \tabularnewline[\doublerulesep]
\hline 
\noalign{\vskip\doublerulesep}
$d=4$ & $\begin{array}{c}
\nu_{r}^{\dagger}=\nu_{r}\\
\nu_{r}^{\sharp}=\nu_{r}\end{array}$ & $\left[\begin{array}{cc}
\sigma_{x} & 0\\
0 & \sigma_{x}\end{array}\right],\left[\begin{array}{cc}
\sigma_{y} & 0\\
0 & -\sigma_{y}\end{array}\right],\left[\begin{array}{cc}
\sigma_{z} & 0\\
0 & \sigma_{z}\end{array}\right],$ ~~ $\left[\begin{array}{cc}
0 & -i\sigma_{y}\\
i\sigma_{y} & 0\end{array}\right],\left[\begin{array}{cc}
0 & \sigma_{y}\\
\sigma_{y} & 0\end{array}\right]$ & $\begin{array}{c}
B^{\mathrm{\dagger}}=B\\
B^{\sharp}=B\end{array}$ & $\frac{1}{4}$Signature\tabularnewline[\doublerulesep]
\hline 
\noalign{\vskip\doublerulesep}
$d=6$ & $\begin{array}{c}
\nu_{r}^{\dagger}=\nu_{r}\\
\nu_{r}^{\mathrm{T}}=-\nu_{r}\end{array}$ & --- & $\begin{array}{c}
B^{\mathrm{\dagger}}=B\\
B^{\mathrm{T}}=-B\end{array}$ & Pfaffian\tabularnewline[\doublerulesep]
\hline 
\noalign{\vskip\doublerulesep}
$d=7$ & $\begin{array}{c}
\nu_{r}^{\mathrm{T}}=-\nu_{r}^{\dagger}\end{array}$ & --- & $\begin{array}{c}
B^{\mathrm{\dagger}}=B^{\mathrm{T}}\end{array}$ & Determinant\tabularnewline[\doublerulesep]
\hline 
\noalign{\vskip\doublerulesep}
$d=8$ & $\begin{array}{c}
\nu_{r}^{\dagger}=\nu_{r}\\
\nu_{r}^{\mathrm{T}}=\nu_{r}\end{array}$ & --- & $\begin{array}{c}
B^{\mathrm{\dagger}}=B\\
B^{\mathrm{T}}=B\end{array}$ & $\frac{1}{2}$Signature\tabularnewline[\doublerulesep]
\hline
\end{tabular}
\caption{GOE approximate representations of $S^{d}.$ The signature of $B$
leads directly to an integer invariant, while we must take the sign
of the determinant and Pfaffian to get $\pm1$ and so a $Z_{2}$ invariant.
\label{tab:GOE-reps}}
\end{table}

The following theorem tells us that we a dealing with
a real, or purely imaginary, invertible matrix with the
correct symmetry so that we can compute a real Pfaffian,
a real determinant, or perform a positive/negative
eigenvalue count of the real eigenvalues.

\begin{thm}
Suppose $H_1,\ldots,H_{d+1}$ are $2N$-by-$2N$ matrices
with $H_r^\dagger = H_r^\sharp=H_r,$
\[
 \left\Vert \sum_{r=1}^{d+1} H_r - I \right\Vert 
\leq \delta
\]
and
$ \left\Vert \left [ H_r, H_s \right] \right\Vert 
\leq \delta $
for $r \neq s$ and some $\delta > 0.$ Suppose $\nu_1,\ldots,\nu_{d+1}$
are matrices in $\mathbf{M}_{2m}(\mathbb{C})$ so that
\begin{align*}
\nu_{r}^{\dagger} & =\nu_{r}^{-1}\\
\nu_{r}^{\dagger}\nu_{s} & =-\nu_{s}^{\dagger}\nu_{r}\quad(r\neq s)
\end{align*}
$ $ for $r\neq s.$ Let $Z_{k}$ be the matrices defining the dual
operations on the $\mathbf{M}_{2k}(\mathbb{C})$ and define
\[
U=\frac{1}{\sqrt{2}}\left(I+Z_{N}\otimes Z_{m}\right),
\]
\[
B = B(H_{1},\ldots,H_{d+1}) =
U^\dagger
\left ( \sum_{r=1}^{d+1} H_r  \otimes \nu_r  \right)
U.
\]
\begin{enumerate}
\item We have
\[
\left\Vert B^{\dagger}B-I\right\Vert \leq(d+2)\delta.
\]
\item If $\nu_r^{\dagger}=\nu_r$ for all $r$ then $B^{\dagger}=B.$ 
\item If $\nu_r^{\sharp}=\nu_r$ for all $r$ then $B^\mathrm{T}=B.$
\item If $\nu_r^{\sharp}=-\nu_r$ for all $r$ then $B^\mathrm{T}=-B.$
\end{enumerate}
\end{thm}

\begin{proof}
(1) We use the fact that $U$ is a unitary and we find
\[
UB^{\dagger}BU^{\dagger}
=
\sum_r H_r^2 \otimes I +
\sum_{r<s} \left [ H_r,H_s \right ] \otimes \nu_r^\dagger \nu_s
\]
so
\[
UB^{\dagger}BU^{\dagger}-I
=
\left(\sum_rH_r^{2}-I\right)\otimes I
+\sum_{r<s}\left[H_r,H_s\right]\otimes\nu_r^{\dagger}\nu_s
\]
and
\[
\left\Vert B^{\dagger}B-I\right\Vert 
\leq
\left\Vert \left(\sum_rH_r^{2}-I\right)\right\Vert 
+\sum_{r<s}\left\Vert \left[H_r,H_s\right]\right\Vert \left\Vert \nu_r^{\dagger}\nu_s\right\Vert .
\]
The $\nu_r$ are unitary, so $\left\Vert \nu_r^{\dagger}\nu_s\right\Vert =1$
and we have the desired upper bound. 

(2) From the definition of $B$ we quickly obtain
\[
B^{\dagger}
=
U^\dagger
\left( \sum_{r=1}^{d+1}H_r\otimes\nu_r^{\dagger}\right)
U
\]
so $\nu_r^{\dagger}=\nu$ implies $B^{\dagger}=B.$ 

(3) and (4) Lemma~\ref{QuaternionsSquared} implies
\[
B^{\mathrm{T}}
=
U^\dagger
\left( \sum_{r=1}^{d+1}H_r\otimes\nu_r^{\sharp} \right)
U
\]
and so $\nu_r^{\sharp}=\pm \nu_r$ implies
$B^\mathrm{T}= \pm B.$  

\end{proof}

\begin{thm}
Suppose $H_{1},\ldots,H_{d+1}$ are $n$-by-$n$ real symmetric matrices
with \[
\left\Vert \sum_{r=1}^{d+1}H_{r}-I\right\Vert \leq\delta\]
and $\left\Vert \left[H_{r},H_{s}\right]\right\Vert \leq\delta$ for
$r\neq s$ and some $\delta>0.$ Suppose $\nu_{1},\ldots,\nu_{d+1}$
are matrices in $\mathbf{M}_{2m}(\mathbb{C})$ so that \begin{align*}
\nu_{r}^{\dagger} & =\nu_{r}^{-1}\\
\nu_{r}^{\dagger}\nu_{s} & =-\nu_{s}^{\dagger}\nu_{r}\quad(r\neq s)\end{align*}
for $r\neq s.$ Let \[
B=B(H_{1},\ldots,H_{d+1})=\sum_{r=1}^{d+1}H_{r}\otimes\nu_{r}.\]

\begin{enumerate}
\item We have $\left\Vert B^{\dagger}B-I\right\Vert \leq(d+2)\delta.$
\item If $\nu_{r}^{\dagger}=\nu_{r}$ for all $r$ then $B^{\dagger}=B.$
\item If $\nu_{r}^{\mathrm{T}}=\nu_{r}$ for all $r$ then $B^{\sharp}=B.$
\item If $\nu_{r}^{\mathrm{T}}=-\nu_{r}$ for all $r$ then $B^{\sharp}=-B.$
\end{enumerate}
\end{thm}

\begin{proof}
The proof is nearly the same as before, but we now use lemma~\ref{lem:TransposeDual}.
\end{proof}

In addition to the Pfaffian, we use the determinant and
signature. 

\begin{defn}
\label{signatureDef}
If $X$  is an invertible, Hermitian matrix, its {\em signature}
$
\mathrm{Sig}(X)
$
is the number of positive eigenvalues of $X$ minus the number
of  negative eigenvalues of $X.$ 
\end{defn}

We will see below that the sign of the Pfaffian is identifying
a class in $K_2(\mathbb{R}),$ the sign of the determinant is 
identifying a class in $K_1(\mathbb{R}),$ while the signature
is  identifying a class in $K_0(\mathbb{R}).$  The reason this
$K$-group appear is that the interesting $K$-theory of the
$d$-sphere ends up in $K_{8-d}$ which results in an invariant
in $K_{8-d}(\mathbb{H}).$  Rather than attempt to understand
directly $K_6,$ $K_5$ and $K_4$ of the quaternions $\mathbb{H},$
we use isomorphisms
\[
K_{n}(\mathbb{H}) \cong K_{n-4}(\mathbb{H} \otimes \mathbb{H}) \cong K_{n-4}(\mathbb{R}) 
\]
which means that in the GSE case, a system in dimension $d$
has an invariant in $K_{4-d}(\mathbb{R}).$  As the
$K$-theory of $\mathbb{R}$ is
\[
\mathbb{Z}, \mathbb{Z}_2, \mathbb{Z}_2, 0,  \mathbb{Z} , 0, 0, 0
\]
in degrees $0,1,\dots$ we are getting, as expected, invariants
in $\mathbb{Z}_2$ in dimensions $2$ and $3,$ and an invariant
in $\mathbb{Z}$ in dimension $4.$

\begin{table}
\begin{tabular}{|>{\centering}m{0.7in}|>{\centering}m{0.8in}|>{\centering}m{1in}|>{\centering}m{1.6in}|}
\hline 
\noalign{\vskip\doublerulesep}
$K_{n}(A)$ & \multicolumn{2}{c|}{For $K_{n}(\mathbb{C}).$ } & \tabularnewline[\doublerulesep]
\hline 
\noalign{\vskip\doublerulesep}
 & Restrictions as given & Equivalent conditions & Numerical invariant\tabularnewline[\doublerulesep]
\hline 
\noalign{\vskip\doublerulesep}
$n=0,2,\dots$ & $\begin{array}{c}
X^{\dagger}=X\end{array}$ & Hermitian & $\frac{1}{2}\mathrm{Sig}(X)$ in $Z$ \tabularnewline[\doublerulesep]
\hline
\end{tabular}
\caption{$K_{n}(A)$ built using invertible elements in $\mathbf{M}_{n}(A)$
for a $C^{*}$-algebra $A.$ Here $\mathrm{Sig}(X)$ denotes the inverible
matrix $X,$ the difference between the number of its positive eigenvalues
and the number of its negative eigenvalues. The rows repeat in the
pattern as Bott Peridicity has order two in the complex case. 
\label{tab:K-theoryComplexCase}}
\end{table}

\begin{table}
\begin{tabular}{|>{\centering}m{0.7in}|>{\centering}m{0.8in}|>{\centering}m{1in}|>{\centering}m{1.6in}|}
\hline 
\noalign{\vskip\doublerulesep}
$K_{n}\left(\Re(A)\right)$ & \multicolumn{2}{c|}{For $K_{n}(\mathbb{R})=K_{n}(\Re(\mathbb{C},\mathrm{id})).$ } & \tabularnewline[\doublerulesep]
\hline 
\noalign{\vskip\doublerulesep}
 & Restrictions as given & Equivalent conditions & Numerical invariant\tabularnewline[\doublerulesep]
\hline 
\noalign{\vskip\doublerulesep}
$n=0$ & $\begin{array}{c}
X^{\mathrm{T}}=X\\
X^{\dagger}=X\end{array}$ & real and symmetric & $\frac{1}{2}\mathrm{Sig}(X)$ in $Z$ \tabularnewline[\doublerulesep]
\hline 
\noalign{\vskip\doublerulesep}
$n=1$ & $X^{\mathrm{T}}=X^{\dagger}$ & real orthogonal & $\mathrm{sgn}\left(\det(X)\right)$ in $Z_{2}$ \tabularnewline[\doublerulesep]
\hline 
\noalign{\vskip\doublerulesep}
$n=2$ & $\begin{array}{c}
X^{\mathrm{T}}=-X\\
X^{\dagger}=X\end{array}$ & pure-imaginary and hermitian & $\mathrm{sgn}\left(\mathrm{Pf}(X)\right)$ in $Z_{2}$ \tabularnewline[\doublerulesep]
\hline
\end{tabular}
\caption{The first three $K$-theory groups $\mathbf{M}_{k}(\mathbb{R})$ of
the are summarized. Here $A$ refers to a generic real $C^{*}$-algebra.
\label{tab:KtheoryReal}}
\end{table}

\begin{table}
\begin{tabular}{|>{\centering}m{0.7in}|>{\centering}m{0.8in}|>{\centering}m{1in}|>{\centering}m{1.6in}|}
\hline 
\noalign{\vskip\doublerulesep}
$K_{n}\left(\Re(A)\right)$ & \multicolumn{2}{c|}{$K_{n}(\mathbb{H})=K_{n}(\Re(\mathbf{M}_{2}(\mathbb{C}),\sharp)).$} & \tabularnewline[\doublerulesep]
\hline 
\noalign{\vskip\doublerulesep}
 & Restrictions as given & Equivalent conditions & Numerical invariant\tabularnewline[\doublerulesep]
\hline 
\noalign{\vskip\doublerulesep}
$n=0$ & $\begin{array}{c}
X^{\mathrm{\sharp}}=-X\\
X^{\dagger}=X\end{array}$ & pure-imaginary and hermitian & $\frac{1}{4}\mathrm{Sig}(X)$ in $Z$ \tabularnewline[\doublerulesep]
\hline 
\noalign{\vskip\doublerulesep}
$n=1$ & $X^{\sharp}=X^{\dagger}$ & real orthogonal & $0$\tabularnewline[\doublerulesep]
\hline 
\noalign{\vskip\doublerulesep}
$n=2$ & $\begin{array}{c}
X^{\sharp}=X\\
X^{\dagger}=X\end{array}$ & real and symmetric & $0$\tabularnewline[\doublerulesep]
\hline
\end{tabular}
\caption{The first three $K$-theory groups of the the
quaternions $\mathbf{M}_{N}(\mathbb{H})$
(realized as matrices respecting the dual operations) are summarized.
The $\tfrac{1}{4}$ is needed instead of $\tfrac{1}{2}$ to account
for the doubling of eigenvalues in self-dual, Hermition matrices.
\label{tab:KtheoryQuaternions}}
\end{table}

\section{Mathematical Problems}
There are several interesting mathematical problems which arise from this construction.  This section will be useful even if one 
is only interested in numerical applications of this techniques,
since an understanding of what has been proved and which directions of the implications are known will be very useful.
The first five problems are purely problems in $C^*$-algebra, having to do with matrices, while the last problem involves fermionic
Hamiltonians.

In some of these problems, we consider the question of stable limits for matrices.  This is very similar to the idea of
the stable limit for free fermion systems.  The approximation problem in the stable limit,
given matrices $H_r$ which are a soft representation of $S^d$, is whether one can find exactly commuting matrices $J_r$, with
$\sum_r J_r^2=I$, such that the set of matrices $\{H_r \oplus J_r \}$
can be approximated by a set of exactly commuting matrices $H_r'$.

We begin with
\begin{problem}
\label{matinvar}
Consider $d+1$ matrices $H_r$, belonging to one of the 10 different random matrix theory universality classes\cite{10fold},
giving a soft sphere $S^d$.  Determine the topological invariants in the stable limit.
\end{problem}
We conjecture that
\begin{conjecture}
If the $H_r$ belong to classes GUE,GOE, or GSE, then the solution to the above problem is given by the same results as in \cite{kitaev,ludwig} in the free fermion
case.  See, for example, table 4 in \cite{ludwig}, where the dimension along the top is the dimension $d$ of the sphere, and the
universality class along the left is the universality class of the matrices $H_r$.
The precise conjecture is that for soft representations of the sphere in
dimensions $2$ through $9,$ the
only topological invariants in the stable limit are invariants described in Tables
\ref{tab:GUE-reps}, \ref{tab:GSE-reps} and \ref{tab:GOE-reps}.
\end{conjecture}
In fact, as we stress, only in some cases is this conjecture what we need to study matrix invariants of free fermion systems.
Given free fermion systems in the GOE,GUE,GSE classes, by forming projected position matrices $H_r$ we obtain matrices $H_r$
which are in the same universality class as the original Hamiltonian.  In the previous section, we explained how to
construct invariants in two particular cases, the GSE case in $d=2$ and $d=3$.  Later, we consider this from a more general
K-theoretic point of view, giving a calculation of invariants of these three
classes in all dimensions.  The conjecture above is essentially equivalent to the
conjecture that the invariants we have computed comprise all of the invariants
in these classes in the stable limit.

\begin{problem}
Show that a nontrivial value of the invariant computed from the $B(H_1,...,H_{d+1})$ (either a nontrivial integer or $Z_2$ invariant)
implies that the matrices $H_1,...,H_{d+1}$
cannot be approximated by exactly commuting Hermitian matrices of the given symmetry class, GUE, GOE, or GSE.
\end{problem}
  This result has long been known in the case of $d=2$ without symmetries, and we
have given a proof\cite{hastingsloring} in the time reversal invariant case in $d=2$.  The proof in other cases (other symmetry classes or
dimensions) is very simple and completely analogous to the proof in the other known cases: any other set of matrices $H_r'$ with $\Vert H_r'-H_r \Vert$ small gives rise to an $B(H_1',...,H'_{d+1})$ with
$B(H_1',...,H'_{d+1})-
B(H_1,...,H_{d+1})$ small.  However, the matrix $B(H_1,...,H_{d+1})$ is close to a projector, so its eigenvalues are far from zero, so
invariants of the matrix such as number of positive eigenvalues, determinant, or Pfaffian, do not change under small changes in $B$.
So, this problem is solved.

\begin{problem}
\label{stableapproxprob}
Show that a trivial value of the invariant computed from $B(H_1,...,H_{d+1})$ implies that the matrices can be
approximated by exactly commuting Hermitian matrices of the given symmetry class, either GUE, GOE, or GSE, possibly in the stable limit.
\end{problem}
This problem has been solved in the case of Hermitian matrices
with no other symmetries.  In the case of $d=2$, this is solved in \cite{LoringWhen} and the solution is reviewed
in \cite{hastingsloring}.  In the case of $d=2$, this result holds without going to the stable limit; that is, it is {\it not} necessary
to add on trivial degrees of freedom.  This relies heavily on the fact that two almost commuting Hermitian matrices can be approximated by
two exactly commuting Hermitian matrices, known as Lin's theorem\cite{lin,hastingslin}.  In more than two dimensions, this has been
solved in the stable limit\cite{dadarlat}.

\begin{problem}
Show that for each dimension and each symmetry class and for any given choice of the invariant in that symmetry class, we can construct, for any sufficiently large integer $N$, $d+1$ different $N$-by-$N$ matrices $H_r$
in the appropriate symmetry class, with
$\Vert [H_r,H_s] \Vert$ converging to zero as $N$ tends to infinity and $\Vert \sum_a H_r^2-I \Vert$ also converging to zero
as $N$ tends to infinity and such that, for any $N$, the matrices have the given value of the invariant.
\end{problem}
In fact, in the case that the invariant is an integer, it suffices to find matrices for which the invariant is equal to $+1$ and to $-1$,
as any larger integer invariant can be obtained by the direct sum of matrices with invariant $+1$ or $-1$.
This problem has been solved in the case of $d=2$ with time-reversal symmetry or with arbitrary Hermitian matrices\cite{hastingsloring}.
A general solution would be obtained if one could solve this problem for every dimension.  One route to obtaining this mathematical result is, of course, to
construct the desired matrices by constructing an appropriate free fermion
system on a sphere and then constructing the band projected position matrices 
following our procedure.  This would require solving the next problem, that
the matrix invariants agree with the fermionic invariants:
\begin{problem}
Show that the matrix invariants we compute agree with the classification of free fermions of \cite{kitaev}, so that
for each fermi  Hamiltonian we obtain the same invariant.
\end{problem}
We now sketch a possible approach to showing this in the $\mathbb{Z}$ case (the $\mathbb{Z}_2$ case would be similar, so we just sketch
one case).
We would begin by showing that the invariants agree for any single system $H_{+1}$ with the invariant equal to $+1$ and any other system $H_{-1}$ with the invariant equal to $-1$.  We use the term ``matrix invariant'' to denote the invariant we compute and we use
the term ``free fermion invariant'' to denote the invariants from the classification of \cite{kitaev}.
Once this is done, we can show that the invariants agree for all systems
as follows: let $H_0$ be a given free fermi Hamiltonian.  Suppose 
$H_0$ has a free fermion invariant equal to $m$.
Suppose $m>0$, without loss of generality.
Then, consider the Hamiltonian $H=H_0\oplus H_{-1} \oplus ... \oplus H_{-1}$, where we add $m$ copies of $H_{-1}$.
Then, by definition of the invariants of \cite{kitaev}, this Hamiltonian can be continued to a trivial Hamiltonian, so this
Hamiltonian $H$ has localized Wannier function.  Thus, $H$ has a matrix invariant equal to $0$.  However, the matrix
invariants we study are additive under direct sum.
That is, let $B(\{H_r\})$ denote the index computed from 
a given set of matrices $H_r$.  Given any matrices $H_r$ and $H_r'$, we have
\be
\label{add}
B(\{H_r\})+ B(\{H_r'\})=B(\{H_r\oplus H_r'\}).
\ee
Choose $H_r'$ to be the band projected matrices obtained from the Hamiltonian $H_{-1} \oplus .. \oplus H_{-1}$ and let
$H_r$ be the band projected matrices obtained from the Hamiltonian $H_0$.  Then, the right-hand side of this equation
is equal to zero since $H$ has matrix invariant equal to zero, so Eq.~(\ref{add}) implies that $B(\{H_r\})=m$.
So, it suffices to show equality in only one non-trivial case.

\section{1D Systems and the Polar of a matrix}

\subsection{Soft one-torus and one-sphere representations}
In one-dimensional GOE, GUE or GSE systems (non-chiral) systems, there
are no obstructions to gapped and local Hamiltonians being deformed
to trivial Hamiltonians. We can use a basic
matrix function to fix the one band compressed periodic observable
to make it an actual unitary, diagonalize that unitary and so find
local Wannier functions. We go through this simple case in detail
as we develop machinery used in the higher dimensional cases. 

The simpler nature of 1D systems is related to the fact that in $C^*$-algebras
and matrix theory, the relations associated to a one-dimensional space
are ``stable'' in a way that fails for the relations associated to
higher-dimensional spaces. For example, consider a matrix $X$ such
that
\begin{equation}
\left\Vert X\right\Vert \leq1,\ X^{\dagger}\approx X.
\label{eq:fuzzyLine}
\end{equation}
It is trival to see that with $Y=\tfrac{1}{2}\left(X^{\dagger}+X\right)$
we have $Y\approx X$ and
\begin{equation}
\left\Vert Y\right\Vert \leq1,\ Y^{\dagger}=Y.
\label{eq:lineEq}
\end{equation}
This same trick works in any $C^{*}$-algebra. As the eigenvalues
of $Y$ are real and of magnitude at more $1,$ we think of equation
(\ref{eq:squareEq}) as being associated to a line segment. It is so
easy to fix ``approximate Hermitian'' operators to be Hermitian that
generally one only studies inexact relations between Hermitian matrices,
or relations between general matrices.

This approximation of almost Hermitian $X$ by exactly Hermitian $Y$ is an
example of a class of problems called ``stable relations''.  We have a set of matrices that
almost obey certain constraints, such as a matrix being almost Hermitian, and we ask whether
the matrices can be approximated by matrices that exactly obey the given constraints.
Our discussion in most of this paper is centered on the question of stable relations
involving almost commuting matrices: giving a set of almost commuting matrices, can they
be approximated by exactly commuting matrices?  However, in parallel with this question of
almost commuting matrices, we often have to deal with problems such as replacing almost unitary
matrices by exactly unitary matrices, so it is worth mentioning the more general problem of
stable relations.

Moving to the square, it is a difficult theorem \cite{lin} that given
a matrices $X_{1}$ and $X_{2}$ such that
\begin{equation}
\left\Vert X_{r}\right\Vert \leq1,\ X_{r}^{\dagger}=X_{r},\ X_{1}X_{2}\approx X_{2}X_{1}
\label{eq:fuzzySquare}
\end{equation}
there will exist $Y_{r}$ matrices with $Y_{r}\approx X_{r}$ and
\begin{equation}
\left\Vert Y_{r}\right\Vert \leq1,\ Y_{r}^{\dagger}=Y_{r},\ Y_{1}Y_{2}=Y_{2}Y_{1}.
\label{eq:squareEq}
\end{equation}
Properly stated, with uniform norm conditions, this is \emph{false}
in general $C^{*}$-algebras.

Some relations are not obviously associated
with any topological space, but give interesting stability results,
even for matrices. For example, given a matrix $X$ with 
\[
\left\Vert X\right\Vert \leq1,\ X^{n}\approx0
\]
there is a matrix $Y$ with $Y\approx X$ and
\[
\left\Vert Y\right\Vert \leq1,\ Y^{n}=0.
\]
This is true, but not for easy reasons, in general $C^{*}$-algebras
\cite{shulman}.

What we need for the study of 1D systems is a way to find for matrices
$X_{1}$ and $X_{2}$ with
\begin{equation}
X_{r}^{\dagger}=X_{r},\ X_{1}X_{2}\approx X_{2}X_{1},\ X_{1}^{2}+X_{2}^{2}\approx I
\label{eq:fuzzyCircle}
\end{equation}
matrices $Y_{1}$ and $Y_{2}$ with $Y_{r}\approx Y_{r}$ and 
\begin{equation}
Y_{r}^{\dagger}=Y_{r},\ Y_{1}Y_{2}=Y_{2}Y_{1},\ Y_{1}^{2}+Y_{2}^{2}=I.
\label{eq:circleEq}
\end{equation}
As the scalars that satisfy equation (\ref{eq:circleEq}) are just
$\sin(\theta)$ and $\cos(\theta)$ we consider equation (\ref{eq:circleEq})
as being associated to the circle.
Such a 1D ``stable relations'' problem has an easy solution.  We will
verify that this solution respects the needed symmetries so
that when we fix the one band compressed periodic observable and so
find local Wannier functions, those will have the correct symmetry as
well.

The only geometry to consider in 1D is the one-circle/one-torus.
There is a distinction in the equations used to describe the
same space that becomes important when we compress the position
operators. 
We consider a lattice on a torus of radius $L$ and we use $\theta$
to denote angle on the torus. We define $\theta(i)$ to be the angle
$\theta$ of site $i,$ and let $n$ be the number of sites. In analogy
to the two-sphere case, we define $\Theta$ to be a diagonal matrix,
with $\Theta_{ii}=\theta(i),$ in $\mathbf{M}_{n}(\mathbb{C}^{n}).$
We define band projected position matrices with 
\[
P\exp(i\Theta)P=Q\begin{pmatrix}U & 0\\
0 & 0 \end{pmatrix}Q^{\dagger},
\]
where $Q$ is a unitary (suppressed above) for the change of basis
that puts $U$ in a canonical block where it can be extracted numerically.
As before, $P$ almost
commutes with $\exp(i\Theta)$ so that $U$ is almost unitary. We
call this $U$ a \emph{soft-representation of the one-torus,}
meaning only that
\[
\left\Vert U^{\dagger}U-I\right\Vert \approx 0.
\]
(In the infinite dimensional case we would need to add the relation
$\left\Vert UU^{\dagger}-I\right\Vert \leq 1.$)

If we were in the GSE case, we would be working with $2N$ sites,
where $\mathbf{e}_{j}$ and $\mathbf{e}_{N+j}$ occupy the same point
but typically correspond to spin up and spin down. We would insist
on $\mathcal{H}$ and $\Theta$ being self-dual, so $\Theta$ would
be diagonal with diagonal
\[
\theta(1),\ldots,\theta(N),\theta(N+1),\ldots,\theta(2N).
\]

Observables that are Hermitian are more standard, so returning to
the GUE case we introduce 
\[
X_{1}
=L\cos\left(\Theta\right),\quad X_{2}
=L\sin\left(\Theta\right).
\]
Then we can band-compress these and define $H_{1}$ and $H_{1}$ by
\begin{align*}
PX_{1}P & =LQ\begin{pmatrix}0 & 0\\
0 & H_{1}\end{pmatrix}Q^{\dagger},\\
PX_{2}P & =LQ\begin{pmatrix}0 & 0\\
0 & H_{2}\end{pmatrix}Q^{\dagger}
\end{align*}
These are what we call a \emph{soft-representation of the one-sphere}
which means 
\[
H_{1}^\dagger=H_{1},\ 
H_{2}^\dagger=H_{2},\ 
H_{1}^{2}+H_{2}^{2}\approx I,\ 
\left[H_{1},H_{2}\right] \approx 0.
\]
Since
\[
Q \begin{pmatrix}0 & 0\\
0 & 2H_{1}\end{pmatrix} Q^\dagger
=2P\cos(\Theta)P
=Q \begin{pmatrix}0 & 0\\
0 & U+U^{\dagger}\end{pmatrix} Q^\dagger,
\]
and using a similar equations regarding $H_{2},$ we have the expected
equations
\[
H_{1}=\tfrac{1}{2}\left(U+U^{\dagger}\right),\ H_{2}=\tfrac{1}{2i}\left(U-U^{\dagger}\right)
\]
that allow for an easy translation between the two situations.

We will obtain Wannier functions by a three step process: adjust $H_1$
and $H_2$ so they exactly commute and exactly square-sum to the
identity; jointly diagonalize these new matrices to produce
a basis of common eigenvectors; use $Q$ to move these common eigenvectors
to the larger space where they become a basis of the low-energy subspace
$P\mathbb{C}^{n}.$  As they will be approximate eigenvectors of
the original observables, they will be somewhat localized.

If the Wannier functions are to have the correct symmetries in the
GSE and GOE cases, then we must choose $Q$ carefully and do the simultaneously
diagonalizing in a way that respects the needed symmetries. We address
this after explaining how soft representations are adjusted to be
exact representations.

\begin{thm}
\label{thm:fix1D} 
Suppose $H_{1}$ and $H_{2}$ are matrices such
that $H_{1}^{*}=H_{1},$ $H_{2}^{*}=H_{2},$ 
\[
\left\Vert H_{1}^{2}+H_{2}^{2}-I\right\Vert \leq \delta
\]
and
\[
\left\Vert \left[H_{1},H_{2}\right]\right\Vert \leq \delta.
\]
If $\delta\leq0.6$ then there are matrices $K_{1}$ and $K_{2}$
so that $K_{r}^{*}=K_{r}$ and 
\[
\left\Vert K_{r}-H_{r}\right\Vert \leq 2\delta
\]
for $r=1,2$, \[
K_{1}^{2}+K_{2}^{2}-I \leq \delta
\]
and\[
\left[K_{1},K_{2}\right] \leq \delta.
\]
If the matrices $H_{r}$ are self-dual, then the matrices $K_{r}$ may be chosen
to be self-dual. If the matrices $H_{r}$ are real, then the matrices
$K_{r}$ may be chosen to be real.
\end{thm}

The proof, given in full below, works with the approximate unitary
$U=H_{1}+iH_{2}.$ This seems an odd choice in the case of real matrices,
but there are enough symmetries at hand to be sure we end with real
matrices in that case. An approximate unitary can be dealt with using
what is perhaps the most important of all matrix functions (the idea of a function
of a matrix is sometimes referred to as ``functional calculus''),
the mapping of an invertible matrix to its polar
part.

\subsection{The polar part of an invertible matrix}
It is well known that for X an invertible matrix, there is a unique
way to express it as a product $X=UP$ with $U$ unitary and $P$
positive (meaning what applied mathematicians call positive semidefinite),
and if $X$ is real we find $P$ is real and $U$ is real orthogonal.
Lesser known is the fact that if $X$ is self-dual, then we find $U$
is self-dual. A formula for U is 
$U=X(X^{\dagger}X)^{-\frac{1}{2}}.$

\begin{defn}
\label{defPolar}
Given an invertible matrix $X,$  the {\em polar part} of $X$ is
\[
\mathrm{polar}(X)
=X(X^{\dagger}X)^{-\frac{1}{2}}.
\]
\end{defn}

A quick explanation for the name is 
\[
\mathrm{polar}\left(\mathrm{diag}\left(\strut\lambda_{1},\ldots,\lambda_{n}\right)\right)
=\mathrm{diag}\left(\frac{\lambda_{1}}{|\lambda_{1}|},\ldots,\frac{\lambda_{1}}{|\lambda_{1}|}\right).
\]
We can calculate $\mathrm{polar}\left(X\right)$ several ways. In the implementation section we discuss how Newton's method can be
used to quickly compute the polar part of a matrix.  Another method
is to diagonalize $X,$ so $X = W D W^\dagger,$  and use the formula
\[
\mathrm{polar}(W X W^\dagger) = W \mathrm{polar}( X ) W^\dagger
\]
which holds for all unitaries, and the apply the scalar function
$\lambda \mapsto \lambda / |\lambda|$ to the diagonal elements of
$D.$
 
\begin{lem}
\label{lem:DistanceOfPolar} 
Suppose $X$ is an invertible matrix
in $\mathbf{M}_{n}(\mathbb{C})$ with 
\[
\left\Vert X^{\dagger}X-I\right\Vert \leq\delta.
\]
If $0\leq\delta\leq0.6$ then
\[
\left\Vert \mathrm{polar}(X)-X\right\Vert \leq\delta.
\]
\end{lem}

\begin{proof}
If $X$ is invertible and  has polar decomposition $X=UP$ then 
\begin{align*}
\left\Vert \mathrm{polar}(X)-X\right\Vert  
& =\left\Vert I-\left(X^{\dagger}X\right)^{-\frac{1}{2}}\right\Vert \\
& =\max\left\{ |\lambda^{-\frac{1}{2}}-1|\,\left|\strut\,\lambda\in\sigma\left(X^{\dagger}X\right)\right.\right\} \\
& \leq\max\left(\left(1-\delta\right)^{-\frac{1}{2}}-1,1-\left(1+\delta\right)^{-\frac{1}{2}}\right).
\end{align*}
For the range of $\delta$ under consideration, this quantity is less
that $\delta.$ 
\end{proof}

\begin{lem}
\label{lem:symmetriesOfPolar} 
If $X$ is an invertible, self-dual
matrix in $\mathbf{M}_{2N}(\mathbb{C})$ then $\mathrm{polar}(X)$
is a self-dual unitary. If $X$ is an invertible 
matrix in $\mathbf{M}_{2N}(\mathbb{R})$
then $\mathrm{polar}(X)$ is real orthogonal.
\end{lem}

\begin{proof}
All parts of this are standard except the claims about the symmetries.
We prove these, and more, in Theorem~\ref{polarSymmetries}.
\end{proof}

Now we prove Theorem~\ref{thm:fix1D}:
\begin{proof}
Suppose $H_{1}$ and $H_{2}$ are self-adjoint, 
\[
\left\Vert H_{1}^{2}+H_{2}^{2}-I\right\Vert \leq
\delta
\]
and
$
\left\Vert \left[H_{1},H_{2}\right]\right\Vert \leq \delta.
$
Let $X=H_{1}+iH_{2}$ so that $H_{1}
=\tfrac{1}{2}\left(X^{\dagger}+X\right)$
and $H_{2}=-\tfrac{i}{2}\left(X^{\dagger}-X\right).$ Let 
\[
U=\mathrm{polar}\left(X\right)
\]
and define $K_{1}$ and $K_{2}$ by
\[
K_{1}=\tfrac{1}{2}\left(U^{\dagger}+U\right)
\]
and 
\[
K_{2}=-\tfrac{i}{2}\left(U^{\dagger}-U\right).
\]
These are evidently self-adjoint. Since $U$ is unitary it commutes
with $U^{\dagger}=U^{-1}$ and so $K_{1}$ commutes with $K_{2}.$
Also
\[
K_{1}^{2}+K_{2}^{2}
=\tfrac{1}{4}\left(U^{\dagger}U^{\dagger}+2U^{\dagger}U+UU\right)
-\tfrac{1}{4}\left(U^{\dagger}U^{\dagger}-2U^{\dagger}U+UU\right)
=I
\]
so we have an exact representation of the one-circle. As to the amount
we have moved the $H_{r}$ to get the $K_{r}$ we estimate
\[
\left\Vert K_{r}-H_{r}\right\Vert \leq\frac{1}{2}\left\Vert U^{\dagger}-X^{\dagger}\right\Vert +\frac{1}{2}\left\Vert U-X\right\Vert 
=\left\Vert U-X\right\Vert .
\]
Lemma~\ref{lem:DistanceOfPolar} gives us the estimate
\begin{alignat*}{1}
\left\Vert U-X\right\Vert 
& \leq\left\Vert X^{\dagger}X-I\right\Vert \\
& =\left\Vert H_{1}^{2}+H_{2}^{2}-I+iH_{1}H_{2}-iH_{2}H_{1}\right\Vert \\
& \leq 2\delta
\end{alignat*}
and so $\left\Vert K_{r}-H_{r}\right\Vert \leq 2\delta.$

If $H_{1}^{\sharp}=H_{1}$ and  $H_{1}^{\sharp}=H_{1}$ then $X^{\sharp}=X.$
By lemma~\ref{lem:symmetriesOfPolar}, $U^{\sharp}=U$ and so $K_{r}^{\sharp}=K_{r}.$
If $H_{1}$ and $H_{2}$ are real, then as they are self-adjoint they
are symmetric. By lemma~\ref{lem:symmetriesOfPolar}, $U^{\mathrm{T}}=U$
and so $K_{r}^{\mathrm{T}}=H_{r}.$ As the $K_{r}$ are Hermitian,
they are real.
\end{proof}

\subsection{A structured spectral theorem}
There are various versions of the spectral theorem, involving commuting
matrices, that or Hermitian or at least normal, real or complex. We
need a version that is not well known, involving self-dual matrices
that are self-adjoint.

\begin{lem}
\label{lem:DualEigenEQ}
If $X$ in $\mathbf{M}_{2N}(\mathbb{C})$ is normal and $X\mathbf{v}=\lambda\mathbf{v}$
then 
\[
X^{\sharp}\left(\mathcal{T}\mathbf{v}\right) = \lambda\left(\mathcal{T}\mathbf{v}\right).
\]
If $X$ is any matrix in $\mathbf{M}_{2N}(\mathbb{C})$
and $X\mathbf{v}=\lambda\mathbf{v}$
then 
\[
\left(\mathcal{T}\mathbf{v}\right)^{\dagger} X^{\sharp} = \lambda \left(\mathcal{T}\mathbf{v}\right)^{\dagger}.
\]
\end{lem}

\begin{proof}
The ordinary spectral theorem tells us  $X^\dagger \mathbf{v} = \overline{\lambda}\mathbf{v}.$  Conjugating
this we discover
$X^\mathrm{T} \overline{\mathbf{v}}= \lambda \overline{\mathbf{v}}.$
Since $-Z X^\sharp Z = X ^\mathrm{T}$ this means
$
-ZX^{\sharp}Z\overline{v}=\lambda\overline{\mathbf{v}}
$
which is equivalent to
$
X^{\sharp}Z\overline{v}=\lambda Z\overline{\mathbf{v}}.
$
Recalling from (\ref{defUpsion}) the definition of $\mathcal{T},$ we finish by
multiplying by $-1.$

For the second claim, we start with the transpose of the
eigenequation,
$ \mathbf{v}^\mathrm{T} X^\mathrm{T} = \lambda  \mathbf{v}^\mathrm{T}.$
This implies
$- \mathbf{v}^\mathrm{T} Z X^\sharp Z = \lambda  \mathbf{v}^\mathrm{T}$
which solves to
$
\left( \mathbf{v}^\mathrm{T} Z \right) X^\sharp = \lambda \left( \mathbf{v}^\mathrm{T} Z \right).
$
\end{proof}

The following is the self-dual finite-dimensional version of the spectral
theorem. We assume familiarity with the complex and real versions.

\begin{thm}
\label{thm:SelfDualSpectralThm} 
If $H_{1},\ldots,H_{k}$ are commuting self-dual,
Hermitian matrices in $\mathbf{M}_{2N}(\mathbb{C})$ there is
a symplectic matrix $U$ so that $U^{\dagger}H_{r}U$ is diagonal
for all $r,$ where the diagonal matrices are of the form
\[
\left[\begin{array}{cc}
\Lambda_{r} & 0\\
0 & \Lambda_{r}\end{array}\right].
\]
\end{thm}

\begin{proof}
A finite set of commuting matrices will has a common eigenvector,
so let $\mathbf{v}$ be a unit vector so that 
\[
H_{r}\mathbf{v}=\lambda_{r}\mathbf{v}
\]
 for all $r.$ By lemma~\ref{lem:DualEigenEQ} we conclude 
\[
H_{r}\left(\mathcal{T}\mathbf{v}\right)=\lambda_{r}\left(\mathcal{T}\mathbf{v}\right)
\]
as well. Choose any symplectic unitary $U_{1}$ so that 
$
U_{1}\mathbf{e}_{1}=\mathbf{v}.
$
By lemma~\ref{lem:symplecticTests} we also have
$
U_{1}\mathbf{e}_{N+1}=\mathcal{T}\mathbf{v}.
$
Let $K_{r}=U_{1}^{\dagger}H_{r}U_{1}.$ Then 
\[
K_{r}\mathbf{e}_{1}=\lambda_{r}\mathbf{e}_{1}
\]
and
\[
K_{r}\mathbf{e}_{N+1}=\lambda_{r}\mathbf{e}_{N+1}.
\]
As the $K_{r}$ are self-adjoint, we conclude that in terms of $N$-by-$N$
blocks, 
\[
Y_{r}=\left[\begin{array}{cccc}
\lambda_{r} & 0 & 0 & 0\\
0 & A_{r} & 0 & C_{r}\\
0 & 0 & \lambda_{r} & 0\\
0 & B_{r} & 0 & D_{r}\end{array}\right].
\]
Unless $N=1,$ and we are done, we form 
\[
Z_{j}=\left[\begin{array}{cc}
A_{r} & C_{r}\\
B_{r} & D_{r}\end{array}\right].
\]
These form a self-dual commuting family of Hermitian matrices. As simple
induction now finishes the proof.
\end{proof}

\subsection{Structured band-compressed position operators}
Recall we had commuting Hermitian matrices $X_{1}$ and $X_{2}$ representing
the position observables, and we adjusted their band-compressed versions
$PX_{r}P$ in two ways to define $H_{1}$ and $H_{2}.$ We rescaled
them to account the physical size of the lattice, and we changed basis
so as to remove blocks of zeros:
\[
PX_{r}P=LQ \left[\begin{array}{cc}
0 & 0\\
0 & H_r  \end{array}\right]Q^\dagger.
\]
Abstractly $P$ is $f(\mathcal{H})$ for $\mathcal{H}$ the Hamiltonian
and $f$ the indicator function for the set $(-\infty,E_{F}].$ However,
we describe this more concretely to facilitate the later discussion
of the numerical method, and to clarify a subtle point that arises
when selecting $Q$ in the GSE case. 
We use the eigensolver described in Section
\ref{sub:Factorization-of-self-dual}
so that the self-duality of the Hamiltonian is
reflected in the diagonalization and the computed
matrix $P$ will also be self-dual.  The numerical error that accumulates in other
eigensolvers could destroy the expected self-duality of $P$ 
if the Hamiltonians has multiple eigenvalues very close to each other.

In the GUE case we can work with any spectral decomposition of the
Hamiltonian $\mathcal{H}$ with eigenvalues 
\[
\lambda_{1}\geq \lambda_{2} \geq \cdots \geq \lambda_{n}
\]
and associated eigenvectors 
\[
\mathbf{q}_{1},\ldots,\mathbf{q}_{n}.
\]
We set $Q=[\mathbf{q}_{1},\ldots,\mathbf{q}_{n}]$ to obtain the unitary
the diagonalizes 
$\mathcal{H}$ to $\mathrm{diag}(\lambda_{1},\ldots,\lambda_{n}).$
If $I$ represents the identity of size $n_{0}$-by-$n_{0},$ where
$n - n_{0}$ is the largest index with $\lambda_{n - n_{0}} > E_{F},$ then
\[
P=Q\left[\begin{array}{cc}
0 & 0\\
0 & I\end{array}\right]Q^{\dagger}.
\]
Therefore
\[
PQ=Q\left[\begin{array}{cc}
0 & 0\\
0 & I\end{array}\right]
\]
(which is a partial isometry) and
\[
\left[\begin{array}{cc}
0 & 0\\
0 & H_r\end{array}\right]
=\frac{1}{L}Q^{\dagger}PX_{r}PQ
=\frac{1}{L}\left[\begin{array}{cc}
0 & 0\\
0 & I\end{array}\right]QX_{r}Q\left[\begin{array}{cc}
0 & 0\\
0 & I\end{array}\right].
\]
That is, $H_{r}$ is the bottom-right $n_{0}$-by-$n_{0}$ block in $\tfrac{1}{L}QX_{r}Q.$

We know there are $K_{1}$ and $K_{2}$ close to $H_{1}$ and $H_{2}$
that are exactly commuting and exactly square-sum to one. We discussed
in \cite{hastingsloring} how this produces Wannier functions,
localized in a somewhat weak sense. Here are illuminate the translation
from the $K_{r}$ in $\mathbf{M}_{n_{0}}(\mathbb{C})$ to a basis
for the low-energy subspace of $\mathbb{C}^{n}.$ 

Let $\delta$ be the larger of $\left\Vert H_{1}^{2}+H_{2}^{2}-I\right\Vert $
and $\left\Vert \left[H_{1},H_{2}\right]\right\Vert .$ That $\delta$
will be small is discussed in the first section. By Theorem~\ref{thm:fix1D}
we may assume $K_{r}$ is within $2\delta$ of $H_{r}.$ We can jointly
diagonalize $K_{1}$ and $K_{2},$ with basis
$\mathbf{v}_{1},\ldots,\mathbf{v}_{n_{0}}$ and
\begin{equation}
K_{r}\mathbf{v}_{j}=\lambda_{j}^{r}\mathbf{v}_{j}.
\label{eq:exactEigen}
\end{equation}
 The corresponding basis of $P\mathbb{C}^{n}$ is 
 \begin{equation}
Q\left[\begin{array}{c}
0\\
\mathbf{v}_{1}\end{array}\right],\ldots,Q\left[\begin{array}{c}
0\\
\mathbf{v}_{n_{0}}\end{array}\right].\label{eq:DefOfWannier}
\end{equation}
These are the desired Wannier functions. From (\ref{eq:exactEigen})
we derive
\[
H_{r}\mathbf{v}_{j}\approx\lambda_{j}^{r}\mathbf{v}_{j}\]
and 
\[
PX_{r}PQ\left[\begin{array}{c}
\mathbf{v}_{1}\\
0\end{array}\right]\approx\lambda_{j}^{r}Q\left[\begin{array}{c}
\mathbf{v}_{1}\\
0\end{array}\right]
\]
and finally
\[
X_{r}Q\left[\begin{array}{c}
\mathbf{v}_{1}\\
0\end{array}\right]\approx\lambda_{j}^{r}Q\left[\begin{array}{c}
\mathbf{v}_{1}\\
0\end{array}\right].
\]
 Precise estimates are possible, see \cite{hastingsloring}, but the
point is that an approximate eigenvector for the diagonal matrices
$X_{1}$ and $X_{2},$ with approximate eigenvalues $\alpha$ and
$\beta,$ must have small coefficients at sites far from $(\alpha,\beta).$

In the GOE case we again can use a generic eigensolver to find a real
orthogonal matrix $Q$ that otherwise works as above. We obtain real
Wannier functions, as desired.

In the GSE we need a symplectic diagonalization of $\mathcal{H},$
meaning eigenvalues 
\[
\lambda_{1}\geq\lambda_{2}\geq\cdots\geq\lambda_{N}
\]
and their doubles
\[
\lambda_{N+1}=\lambda_{1},\lambda_{N+2}=\lambda_{2},\ldots,\lambda_{2N}=\lambda_{N}
\]
and associated eigenvectors 
\[
\mathbf{q}_{1},\ldots,\mathbf{q}_{N},\mathcal{T}\mathbf{q}_{1},\ldots,\mathcal{T}\mathbf{q}_{N}.
\]
Let $N_{0}$ be the largest index with $\lambda_{N - N_{0}} > E_{F}.$
We assemble these eigenvectors to form a unitary $Q$ as follows (in essence, we write first all the eigenvectors $q_i$ with
$\lambda_i>E_F$, then write the corresponding  $\mathcal{T} q_i$, then do the same for the eigenvectors
$q_i$ with $\lambda_i\leq E_F$),
\[
Q = \left[\strut 
\mathbf{q}_1, \dots  ,\mathbf{q}_{N - N_0},
\mathcal{T}\mathbf{q}_1,  \dots   ,\mathcal{T}\mathbf{q}_{N - N_0},
\mathbf{q}_{N - N_0 + 1},  \dots  ,\mathbf{q}_N,
\mathcal{T}\mathbf{q}_{N - N_0 + 1},  \dots  ,\mathcal{T}\mathbf{q}_N
\right ].
\]
This matrix $Q$ satisfies a symmetry similar to being symplectic,
\[
Q\left[\begin{array}{cc}
Z\\
 & Z\end{array}\right]=Z\overline{Q}
\]
where the bottom $Z$ is of size $2N_{0}$-by-$2N_{0}.$ 
For matrices
$A$ and $B,$ of appropriate sizes,
\[
\left(Q\left[\begin{array}{cc}
A & 0\\
0 & B\end{array}\right]Q^{\dagger}\right)^{\sharp}=Q\left[\begin{array}{cc}
A^{\sharp} & 0\\
0 & B^{\sharp}\end{array}\right]Q^{\dagger}.
\]
The projection $P$ will be self dual. We can even see this using
$Q$ to compute it:
\[
\mathcal{H}=Q\left[\begin{array}{cccccc}
\lambda_{1}\\
 & \ddots\\
 &  & \lambda_{N}\\
 &  &  & \lambda_{1}\\
 &  &  &  & \ddots\\
 &  &  &  &  & \lambda_{N}\end{array}\right]Q^{\dagger}
 \]
so
\[
P=Q\left[\begin{array}{cc}
\left[\begin{array}{cc}
0\\
 &I \end{array}\right]\\
 & \left[\begin{array}{cc}
0\\
 & I\end{array}\right]\end{array}\right]Q^{\dagger}
\]
and
\[
P^{\sharp}=Q\left[\begin{array}{cc}
\left[\begin{array}{cc}
0\\
 & I\end{array}\right]^{\sharp}\\
 & \left[\begin{array}{cc}
0\\
 & I\end{array}\right]^{\sharp}\end{array}\right]Q^{\dagger}.
\]
Thus also $X_{1}$ and $X_{2}$ are self-dual, and
\begin{align*}
\left(Q\left[\begin{array}{cc}
0 & 0\\
0 & H_{r}^{\sharp}\end{array}\right]Q^{\dagger}\right) & =\left(Q\left[\begin{array}{cc}
0& 0\\
0 & H_{r} \end{array}\right]Q^{\dagger}\right)^{\sharp}\\
& =\frac{1}{L}\left(QQ^{\dagger}PX_{r}PQQ^{\dagger}\right)^{\sharp}\\
& =\frac{1}{L}PX_{r}P\\
& =Q\left[\begin{array}{cc}
0& 0\\
0 & H_{r}^{\sharp} \end{array}\right]Q^{\dagger}
\end{align*}
which at last gives us $H_{r}^{\sharp}=H_{r}.$

We can therefore find $K_{r}$ that are self-dual, and so obtain a
common basis of the form 
\[
\mathbf{v}_{1},\ldots,\mathbf{v}_{N_{0}},\mathcal{T}\mathbf{v}_{1},\dots,\mathcal{T}\mathbf{v}_{N_{0}}.
\]
The Wannier functions are then 
\[
Q\left[\begin{array}{c}
0\\
\mathbf{v}_{1}\end{array}\right],\ldots,Q\left[\begin{array}{c}
0\\
\mathbf{v}_{N_{0}}\end{array}\right],Q\left[\begin{array}{c}
0\\
\mathcal{T}\mathbf{v}_{1}\end{array}\right],\ldots,Q\left[\begin{array}{c}
0\\
\mathcal{T}\mathbf{v}_{N_{0}}\end{array}\right].
\]
These have the desired structure, as
\[
\mathcal{T}\left( Q\left[\begin{array}{c}
0\\
\mathbf{v}_{j}\end{array}\right]\right)
=-Z\overline{Q}\left[\begin{array}{c}
0\\
\overline{\mathbf{v}_{j}}\end{array}\right]
=-Q\left[\begin{array}{cc}
Z\\
 & Z\end{array}\right]\left[\begin{array}{c}
0\\
\overline{\mathbf{v}_{j}}\end{array}\right]=Q\left[\begin{array}{c}
0\\
\mathcal{T}\mathbf{v}_{j}\end{array}\right].
\]

\section{Torus and Other Geometries}
We can consider lattice systems on other topologies such as the torus.  In general, our procedure, discussed below is to map the matrices on the torus to matrices on a sphere and then compute the sphere invariants using the techniques described before.  However,
we begin with a particularly nice technique that is available in the GUE case on the two-torus.  This technique will be useful later in relating the index
to the Hall conductance.
Unfortunately, we do not have access to such a simple formula for the index on the torus in the GSE case or in the GUE case on higher  dimensional
torus and so we have to resort to the mapping.

We begin by describing this index on the torus.  We then explain in the next subsection how to map the matrices from
the torus to the sphere to define an index that way, and  then theorem \ref{sameindex} shows that the two indices are the same in
the GUE case.

We consider a lattice on the torus, and we use $\theta_1,\theta_2$ to denote
angles on the torus.  We define $\theta_1(i)$ to be the angle $\theta_1$ of
site $i$ and $\theta_2(i)$ to be the angle $\theta_2(i)$.  In analogy to
the sphere case, we define $\Theta_{a}$ for $a=1,2$ to be diagonal matrices with
$(\Theta_{a})_{ii}=\theta_a(i)$. 
We define band projected position matrices with
\begin{eqnarray}
P\exp(i\Theta_1) P & = &\begin{pmatrix} 0 & 0 \\ 0 & U_1 \end{pmatrix} \\ \nonumber
P\exp(i\Theta_2) P & = &\begin{pmatrix} 0 & 0 \\ 0 & U_2 \end{pmatrix}.
\end{eqnarray}
As before, $P$ almost commutes with $\exp(i\Theta_a)$, so that $U_1$ and $U_2$
almost commute with each other and are both almost unitary.  We call this
a soft-representation of the torus.

We now define the torus index.
Let
\be
\label{torusindex}
{\rm tr}(\log(U_1 U_2 U_1^\dagger U_2^\dagger)=r+2\pi i m,
\ee
for some real numbers $r,m$.  However, in fact $m$ is an integer, because
${\rm tr}(\log(U_1 U_2 U_1^\dagger U_2^\dagger)=\log({\rm det}(U_1 U_2 U_1^\dagger U_2^\dagger))=\log({\rm det}(|U_1|^2) {\rm det}(|U_2|^2))$,
and 
${\rm det}(|U_1|^2) {\rm det}(|U_2|^2)$ is real.
In taking the log in Eq.~(\ref{torusindex}), we place the branch cut along the negative real axis.  We define the index
to be the integer $m$.

We can join $U_{r}$ to $\mathrm{polar}(U_{r})$ by a continuous path,
and if
\[
\left\Vert U_{r}^{\dagger}U_{r}-I\right\Vert \leq\delta,
\ 
\left\Vert \left[U_{1},U_{2}\right]\right\Vert 
\leq\delta
\]
for sufficiently small $\delta$, then this torus index cannot jump
to the next integer, and the commutator of the polar parts will stay
small. This is important because the prior work on this index \cite{ExelLoring}
was on almost commuting unitary matrices, while here we consider the problem
of almost commuting matrices which are almost unitary.  The point is that
the index previously considered for almost commuting exactly unitary matrices
also works for almost commuting almost unitary matrices.

\subsection{From the soft torus to the soft sphere}
We now describe various possible functions which can be used to map
from the torus to the sphere, and how to use those to give an alternate way
of computing the index for torus systems.
Any continuous function $\gamma:\mathbb{T}^2\rightarrow S^2$
can, in theory, be used to determine a conversion
\[
\left(U_1,U_2\right)\mapsto\left(H_1,H_2,H_3\right)
\]
 that takes an approximate representation of the torus to an approximate
representation of the sphere, so that we can then compute the index of the
resulting representation of the sphere.

A minor problem is that such a conversion
is not uniquely determined by the function $\gamma.$ More substantial
issues have to do with $K$-theory, numerical efficiency and the preservation
of self-duality. The good news is that the Hermitian matrices will
commute when the unitary matrices commute, and the mapping will be continuous,
even Lipschitz if we are careful. This means that if
$\left(H_1,H_2,H_3\right)$
is bounded away from commuting triples, then $\left(U_1,U_2\right)$
is bounded away from commuting pairs.

The issue with numerical efficiency is that if the matrices $U$ almost commute, then the matrices $H$
also almost commute.  However, the norm of the commutator may increase under the mapping.  By choosing
suitable, sufficiently well-behaved, maps, we can avoid a large increase in the commutator while still obtaining
an efficient numerical procedure.

The map we use for this purpose is, in terms of periodic coordinates
on $\mathbb{T}^{2}$ and real coordinates on $S^{2}\subseteq\mathbb{R}^{3},$
\[
\gamma(\theta,\phi)=(x_{1},x_{2},x_{3})\]
where\begin{align*}
x_{1} & =f(\phi)\\
x_{2} & =g(\phi)+h(\phi)\cos(2\pi i\theta)\\
x_{3} & =h(\phi)\sin(2\pi i\theta)\end{align*}
and $f,$ $g$ and $h$ are continuous, periodic scalar functions
that satisfy\begin{align*}
f^{2}+g^{2}+h_{2} & =1\\
gh & =0.\end{align*}
To ensure the mapping $\mathbb{T}^{2}\rightarrow S^{2}$ is one-to-one
one all but a one-dimensional set, ensuring later the correct $K$-theory,
we need to be sure $f$ moves just once from $-1$ to $1$ on the
set where $h\neq0.$ A choice that works well is
\be
\label{eq:f_1}
f(\phi)
=
\tfrac{150}{128}\sin(2\pi\phi)+\tfrac{25}{128}\sin(6\pi\phi)+\tfrac{3}{128}\sin(10\pi\phi)
\ee
\[
g(\phi)
=\begin{cases}
0 & \tfrac{1}{4}\leq\theta\leq\tfrac{3}{4}\\
\sqrt{1-\left(f(\phi)\right)^{2}} & -\tfrac{1}{4}\leq\theta\leq\tfrac{1}{4}\end{cases}
\]
\[
h(\phi)=\begin{cases}
\sqrt{1-\left(f(\phi)\right)^{2}} & \tfrac{1}{4}\leq\theta\leq\tfrac{3}{4}\\
0 & -\tfrac{1}{4}\leq\theta\leq\tfrac{1}{4}\end{cases}
\]
although the earlier mathematical work on the Bott index uses
$f$ that was piecewise linear.

\begin{figure}
\begin{centering}
\hfill{}
\includegraphics[bb=69bp 210bp 545bp 592bp,clip,scale=0.38]{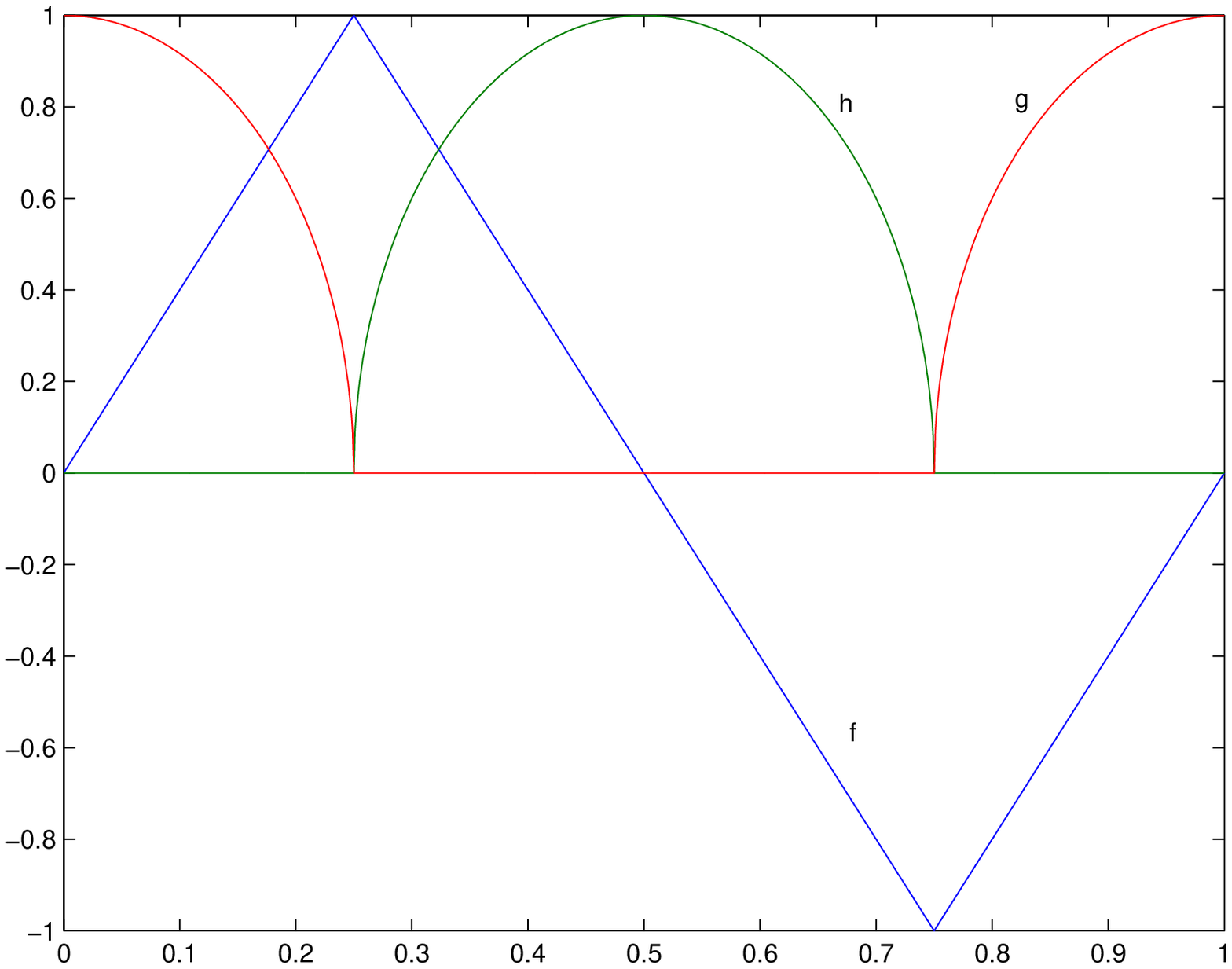}
\hfill{}
\includegraphics[bb=0bp 0bp 476bp 376bp,clip,scale=0.38]{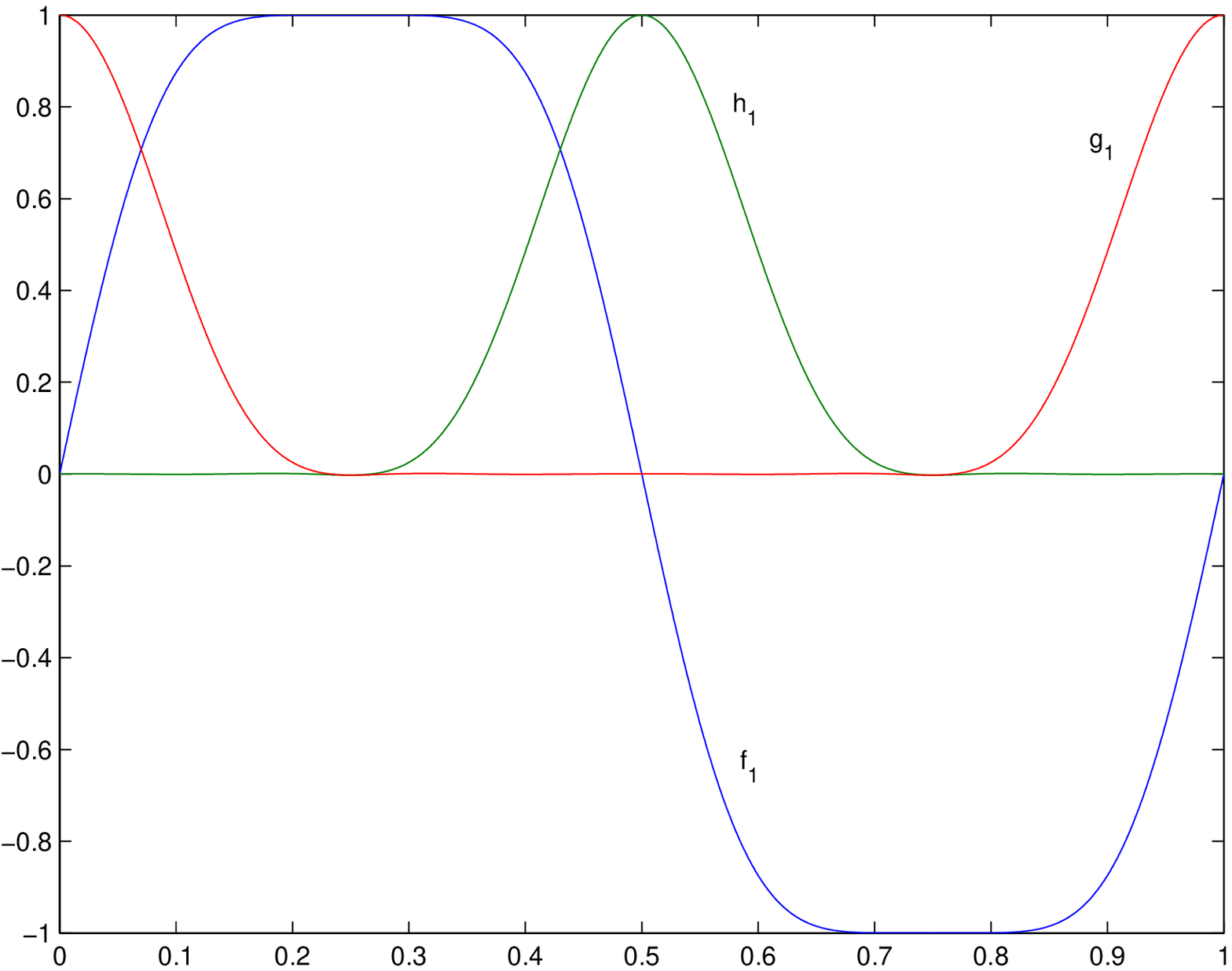}
\hfill{}\\
(a)\hfill{}(b)\hfill{}
\par\end{centering}
\caption{\label{fig:Three-functions-both} Two sets of functions for mapping
the torus the sphere. (a) Here we take $f$ to be piecewise linear.
(b) Here $f_{1}(x)$ is as in (\ref{eq:f_1}).}
\end{figure}

\begin{figure}
\begin{center}
\includegraphics[scale=0.37, trim = 40mm 00mm 5mm 15mm]{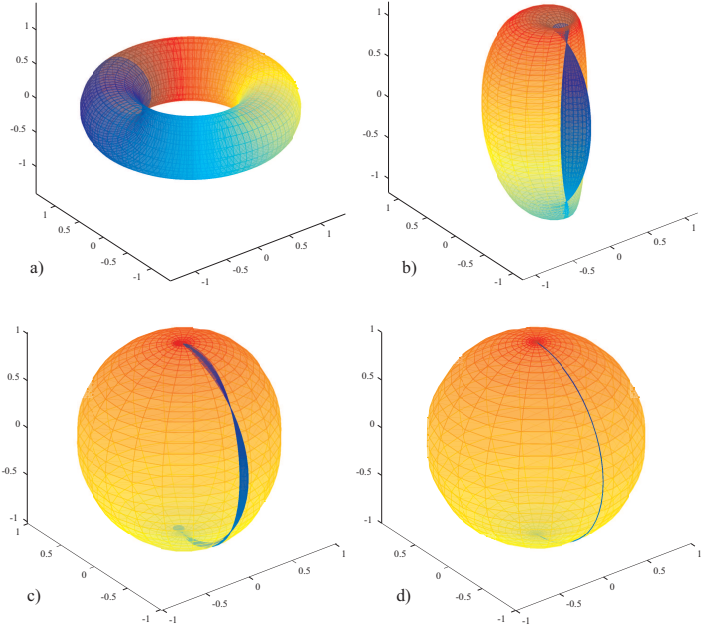}
\caption{Panel d) indicates using $f,$ $g$ and $h$ to map the torus
(panel a)) close to the sphere. Panels b) and c) show the result of
truncating the three functions to
degrees one and three. \label{spherePic}}
\end{center}
\end{figure}

We can interpret this at the level of almost commuting unitary matrices
$U_{1}$ and $U_{2}$ in several ways. For example
\[
h(\phi)\cos(2\pi i\theta)
=
\frac{1}{2}\left(h(\phi)e^{-2\pi i\theta}+h(\phi)e^{2\pi i\theta}\right),
\]
but also
\[
h(\phi)\cos(2\pi i\theta)
=
\frac{1}{4}\left(h(\phi)e^{-2\pi i\theta}+e^{-2\pi i\theta}h(\phi)
+
h(\phi)e^{2\pi i\theta}+e^{2\pi i\theta}h(\phi)\right).
\]
As we have many matrix symmetries to worry about, we prefer this second
formula, and define
\begin{align}
H_1 
& =f(U_2)
\nonumber \\
H_{2} 
& =g(U_2)
+\frac{1}{4}\left\{ h(U_2),U_1^\dagger\right\} 
+\frac{1}{4}\left\{ h(U_2),U_1\right\} 
\label{eq:threeFun}\\
H_3 
& =\frac{i}{4}\left\{ h(U_2),U_1^\dagger\right\} 
-\frac{i}{4}\left\{ h(U_2),U_1\right\} .
\nonumber 
\end{align}

Computing $f(U_2),$ $g(U_2)$ and $h(U_2)$ is simple enough
in theory, as we can diagonalize $U_2$ by a unitary $W,$ so
\[
U_2=
W
\left[\begin{array}{ccc}
e^{2\pi ix_{1}}\\
 &   \ddots\\
 &   & e^{2\pi ix_{n}}
\end{array}\right]
W^\dagger
\implies
f(U_{2})=
W
\left[\begin{array}{ccc}
f(x_{1})\\
 &  \ddots\\
 &   & f(x_{n})
\end{array}\right]
W^{\dagger}.
\]

We shall refer mapping the torus to the sphere as in (\ref{eq:threeFun})
as mapping by the {\em polynomial map} as we implement it by calculating nearby polynomials in $U_1,U_1^\dagger,U_2,U_2^\dagger.$  There is an alternate method
that is slower to compute but appears to show phase transitions
more clearly.  We shall refer to it as the {\em logarithmic map}
and is defines
\begin{align}
H_1 
& =\ell(U_2)
\nonumber \\
H_{2} 
& =
\frac{1}{4}\left\{ m(U_2),U_1^\dagger\right\} 
+\frac{1}{4}\left\{ m(U_2),U_1\right\} 
\label{eq:logMethod}\\
H_3 
& =\frac{i}{4}\left\{ m(U_2),U_1^\dagger\right\} 
-\frac{i}{4}\left\{ m(U_2),U_1\right\} .
\nonumber 
\end{align}
where
\[
\ell(\phi) = \phi
\]
is discontinuous and $m$ is the continuous function $m=\sqrt{1 - \ell^2}.$

\begin{lem}
With $H_1,$ $H_2$ and $H_3$ defined as in (\ref{eq:threeFun}),
if the unitaries $U_1$and $U_2$ satisfy
$
\left\Vert \left[U_1,U_2\right]\right\Vert \leq \delta
$
then
\[
\left\Vert \sum_{r=1}^{3} H_r - I\right\Vert \leq O(\delta)
\]
and
\[
\left\Vert \left[H_r,H_s \right]\right\Vert \leq O(\delta).
\]
\end{lem}

It is also easy to prove a stability result, to the effect that
if we also have $V_1$  and $V_2$ and define the associated three
Hermitians $K_1,$ $K_2$ and $K_3,$ then
\[
\left\Vert H_r - K_r\right\Vert \leq O(\eta)
\]
where
\[
\eta
=
\max\left(\left\Vert U_1 - V_1\right\Vert ,\left\Vert U_2 - V_2\right\Vert \right).
\]
The upshot is that we can now use the Bott index of the sphere to
determine an index on the torus

\begin{defn}
If $U_1,U_2$ are unitary matrices, define $H_r$ by (\ref{eq:threeFun})
and (\ref{eq:f_1}).
\[
\bott \left(U_1,U_2\right) = \bott \left(H_{1},H_{2},H_{3}\right)
\]
If $U_1,U_2$ are self-dual unitary matrices
\[
\bI \left(U_1,U_2\right) = \bI\left(H_{1},H_{2},H_{3}\right).
\]
\end{defn}

\begin{thm}
\label{sameindex}
Suppose $U_{1}$ and $U_{2}$ are unitary matrices and
\[
\left\Vert \left[U_{1},U_{2}\right]\right\Vert \leq\delta.
\]
For small $\delta,$ the torus index of unitaries $U_{1}$ and $U_{2}$
equals
\[
\mathrm{Bott}\left(H_{1},H_{2},H_{3}\right)
\]
where $H_{r}$ are as above, with either choice for $f,$ $g$ and
$h,$ or where the $H_r$ are defined by the log-method. 
\end{thm}

\begin{proof}
This is essentially the main theorem in \cite{ExelLoring}. The only
new claim is that the smoother choice of $f,$ $g$ and $h$ can be
used. There is a continuous deformation between the two choices of
functions, keeping the $gh=0$ and $f^{2}+g^{2}+h^{2}=1$ relations
at all times, so for small enough $\delta$ the resulting path of
Hermitian matrices will keep a gap in its spectrum at $0.$ Therefore
the eigenvalue counts are not able to vary.
\end{proof}

\begin{conjecture}
For small $\delta$ the index $\bI\left(H_{1},H_{2},H_{3}\right)$
is the same when $f,g,h$ are defined by either the polynomial
map of the logarithmic map.
\end{conjecture}

\subsection{Other geometries}

There are other ways to describe the torus than by two unitaries.  For example, we could consider the torus as embedded in 3-dimensional space.  We
would then have matrices $X_1,X_2,X_3$ obeying appropriate algebraic relations describing this surface.
In general, we can consider other spaces in this way.  For example, we could describe a system on a two-dimensional manifold
with many handles by imagining this manifold embedded in 3 dimensional space.  We would then construct three matrices $X_1,X_2,X_3$
again with appropriate algebraic relations.  We then project these matrices as before.  This gives us a soft representation of the
original manifold.

For each such space, we can construct a group, called the reduced $K_0$, describing
possible topological obstructions in the GUE case.  For example, for a three dimensional torus, we find that this group is equal to
$\mathbb{Z}+\mathbb{Z}+\mathbb{Z}$.  These three integer invariants are in fact lower dimensional invariants, similar to the idea of weak topological
insulators studied in the translation invariant case\cite{weakins}.  These can be understood as follows.  We have
three matrices, $U_1,U_2,U_3$, which are approximately unitary and which approximately commute with each other.  Any pair of
them, such as $U_1$ and $U_2$ can have a nontrivial invariant as in the case of the two torus.  The physical interpretation is that
we simply ignore one of the three directions of the torus.  We have a Hamiltonian which is local on the three torus, and then we
simply map the lattice sites to sites on the two torus by ignoring one of thre three coordinates and we then construct a Hamiltonian
which is local on the two torus.
For other manifolds, just as in the case of the three-dimensional torus, we may see topological obstructions arising
from lower dimensions; these will appear in the reduced $K_0$.

One way of obtaining just the highest dimensional obstruction, is to map the manifold onto $S^{d+1}$.  Thus, we map $T^2$ onto
$S^2$, or $T^3$ onto $S^3$, and so on.  This generalizes the torus to sphere mapping described above.  
In the case of a GUE system, this leads to nontrivial obstructions only for $d$ even.
We will also use this mapping
procedure in the case of three-dimensional time-reversal invariant insulators described below.  We study a system whose
lattice sites live on a three-dimensional torus.  However, after constructing the projected position matrices on the torus,
we then map these matrices to matrices on a sphere, and then study invariants on the sphere.  One reason for this is that we do
not have a ``native'' torus formula for the index for any system other than the complex case on the two-torus.  In all other symmetry classes,
we only know how to compute the index by mapping to the sphere.

In the special case of the three torus, we can define the needed four
almost commuting self-dual Hermitians by one of two methods. Suppose
$U_{1},$ $U_{2}$ and $U_{3}$ are self-dual almost commuting matrices.
We first define
\begin{align*}
M & =f(U_{3})+ig\left(U_{3}\right)+\tfrac{i}{2}h\left(U_{3}\right)f\left(U_{2}\right)+\tfrac{i}{2}f\left(U_{2}\right)h\left(U_{3}\right)\\
N & =h\left(U_{3}\right)g\left(U_{2}\right)+\tfrac{1}{2}h\left(U_{3}\right)h\left(U_{2}\right)U_{1}+\tfrac{1}{2}U_{1}h\left(U_{2}\right)h\left(U_{3}\right)\end{align*}
which gives self-dual matrices that are almost normal and almost commute.
From there we set
\begin{align*}
H_{1} & =\tfrac{1}{2}\left(M^{\dagger}+M\right)\\
H_{2} & =\tfrac{i}{2}\left(M^{\dagger}-M\right)\\
H_{3} & =\tfrac{1}{2}\left(N^{\dagger}+N\right)\\
H_{4} & =\tfrac{i}{2}\left(N^{\dagger}-N\right)
\end{align*}
This is our \emph{polynomial map}.

For the \emph{logarithmic map}, we define
\begin{align*}
M & =\ell(U_{3})+\tfrac{i}{2}m\left(U_{3}\right)f\left(U_{2}\right)+\tfrac{i}{2}f\left(U_{2}\right)h\left(U_{3}\right)\\
N & =h\left(U_{3}\right)g\left(U_{2}\right)+\tfrac{1}{2}h\left(U_{3}\right)h\left(U_{2}\right)U_{1}+\tfrac{1}{2}U_{1}h\left(U_{2}\right)h\left(U_{3}\right)
\end{align*}
and then define the $H_{r}$ as above.

\subsection{Relation of Matrix Invariant to Hall Conductance}
We now show that the matrix invariant of the matrices $U_1,U_2$ that we compute is the same as the usual Chern number invariant.
This will in fact provide a simple proof of Hall conductance quantization for free fermion systems.

Consider
\be
{\rm tr}(\log(U_1 U_2 U_1^\dagger U_2^\dagger)).
\ee
We have shown that this is equal to $r+2\pi i m$, for some integer $m$, with
$r=\log(|det(U_1)|^2 |det(U_2)|^2)$.
\begin{lem}
Assume $U_1,U_2$ are obtained from a free fermion system on a torus topology.  Let the lattice
be a square lattice, of size $L$-by-$L$, with $L^2$ lattice sites.  Let the free fermion
Hamiltonian have hopping distance bounded above (uniformly in $L$)
by a constant, and spectral gap bounded below (uniformly in $L$) by a constant.  Then, $\Vert [P,U_1] \Vert,\Vert [P,U_2] \Vert \leq {\mathcal O}(1/L)$,
and $\Vert [U_1,U_2] \Vert \leq {\mathcal O}(1/L^2)$
and $\Vert U_1 U_1^\dagger -I \Vert,\Vert U_2 U_2^\dagger -I \Vert \leq {\mathcal O}(1/L^2)$.
\begin{proof}
This is a minor variation of lemma 5.1 in \cite{hastingsloring}.
\end{proof}
\end{lem}

This implies that
\be
\Bigl\Vert U_1 U_2 U_1^\dagger U_2^\dagger -I \Bigr  \Vert \leq {\mathcal O}(1/L^2),
\ee
so
\be
\Vert \log(U_1 U_2 U_1^\dagger U_2^\dagger)-\Bigl( U_1 U_2 U_1^\dagger U_2^\dagger -I \Bigr)  \Vert \leq {\mathcal O}(1/L^4).
\ee
Thus,
\be
{\rm Im}\Bigl({\rm tr}(U_1 U_2 U_1^\dagger U_2^\dagger -m)\Bigr) \leq {\mathcal O}(1/L^2),
\ee
using the fact that the dimension of the Hilbert space is $L^2$ and the trace
of an operator is bounded by the dimension of the space times the operator
norm of that operator.

Note that
\begin{eqnarray}&&
{\rm Im}({\rm tr}(U_1 U_2 U_1^\dagger U_2^\dagger))  \\ \nonumber
&=&
{\rm Im}({\rm tr}([U_1,U_2] U_1^\dagger U_2^\dagger))+
{\rm Im}({\rm tr}(U_2 U_1 U_1^\dagger U_2^\dagger))
\\ \nonumber
&=& {\rm Im}({\rm tr}([U_1,U_2] U_1^\dagger U_2^\dagger)),
\end{eqnarray}
where we used the fact that $U_a U_a^\dagger$ is Hermitian, and given any two Hermitian matrices $A,B$, we have
${\rm Im}({\rm tr}(AB))=0$.

For notational convenience, let us define
\be
U=P \exp(i\Theta_1) P,
\ee
\be
V=P\exp(i \Theta_2) P,
\ee
so that
\be
U=\begin{pmatrix}0 &0\\0&U_1 \end{pmatrix},
\ee
\be
V=\begin{pmatrix}0 &0\\0&U_2 \end{pmatrix}.
\ee
Then,
\begin{eqnarray}
\label{uv}
{\rm Im}({\rm tr}([U_1,U_2] U_1^\dagger U_2^\dagger))=
{\rm Im}({\rm tr}([U,V] U^\dagger V^\dagger)).
\end{eqnarray}

Define the current operators $J_1,J_2$ by
\be
J_a=i\frac{\exp(i\Theta_a) H \exp(-i \Theta_a) - \exp(-i \Theta_a) H \exp(i \Theta_a)}{2}.
\ee
Then
\be
\label{Japprox}
J_a=[\exp(i \Theta_a),H] \exp(-i \Theta_a)+{\mathcal O}(1/L^2).
\ee
The first equation for the current operators is Hermitian; the second is slightly more convenient for later use.
Let $\Bigl((1-P)J_aP\Bigr)(t)=\exp(i H t) (1-P) J_a P \exp(-i H t)$, following
the usual Heisenberg evolution of operators.
Then,
\begin{eqnarray}
\label{currentsinO}
\lim_{\epsilon\rightarrow 0^+} \int_{-\infty}^0 {\rm d}t \exp(\epsilon t) (1-P) J_a(t) P &=&
\lim_{\epsilon\rightarrow 0^+} \int_{-\infty}^0 {\rm d}t \exp(\epsilon t) \Bigl( (1-P) J_a P \Bigr) (t)
\\ \nonumber
&=&-i(1-P) \exp(i \Theta_a) P \exp(- i \Theta_a)+{\mathcal O}(1/L^2)
\end{eqnarray}
To obtain the above equation, we first use
use Eq.~(\ref{Japprox})
to approximate $J_a$.  We then set
$(1-P) [\exp(i \Theta_a),H] \exp(-i \Theta_a) P=
(1-P) [\exp(i \Theta_a),H] (1-P)\exp(-i \Theta_a) P+
(1-P) [\exp(i \Theta_a),H] P\exp(-i \Theta_a) P$.  The first term is bounded by ${\mathcal O}(1/L^2)$ since we can bound the norms of
$[\exp(i \Theta_a),H]$ and  $(1-P)\exp(-i \Theta_a) P$ both by ${\mathcal O}(1/L)$.  We then approximate 
$(1-P) [\exp(i \Theta_a),H] P\exp(-i \Theta_a) P=[(1-P) \exp(i \Theta_a) P \exp(-i \Theta_a) P,H]+{\mathcal O}(1/L^2)$.  Finally,
we use that $\
\lim_{\epsilon\rightarrow 0^+} \int_{-\infty}^0 {\rm d}t \exp(\epsilon t) [O,H](t)=-iO$ for any operator $O$ which has vanishing matrix
elements between degenerate eigenstates of $H$.
Similarly,
\begin{eqnarray}
\lim_{\epsilon\rightarrow 0^+} \int_{-\infty}^0 {\rm d}t \exp(\epsilon t) P J_a(t) (1-P)&=&
-iP \exp(i \Theta_a) (1-P)\exp(-i\Theta_a)+{\mathcal O}(1/L^2)\\ \nonumber
-i\exp(-i\Theta_a)P \exp(i \Theta_a) (1-P)+{\mathcal O}(1/L^2).
\end{eqnarray}

Let
\be
O_a^+\equiv (1-P) \exp(i \Theta_1) P \exp(-i \Theta_a),
\ee
and
\be
O_a^-\equiv \exp(-i \Theta_a) P \exp(i \Theta_a) (1-P).
\ee
Using the fact that $[\exp(i \Theta_1),\exp(i \Theta_2)]=0$, we have
$P[\exp(i \Theta_1),\exp(i \Theta_2)]P=0$.  However
$P[\exp(i \Theta_1),\exp(i \Theta_2)]P=[U,V]+
P\exp(i \Theta_1) (1-P) \exp(i \Theta_2) P-
P\exp(i \Theta_2) (1-P) \exp(i \Theta_1) P$.
So, $[U,V]=
P\exp(i \Theta_2) (1-P) \exp(i \Theta_1) P-
P\exp(i \Theta_1) (1-P) \exp(i \Theta_2) P$.
Thus,
\begin{eqnarray}
&&{\rm Im}({\rm tr}([U,V] U^\dagger V^\dagger)).
\\ \nonumber
&=&
{\rm Im}({\rm tr}(P\exp(i \Theta_2) (1-P) \exp(i \Theta_1) P \exp(-i\Theta_1) P \exp(-i \Theta_2)))
\\ \nonumber
&&
-{\rm Im}({\rm tr}(P\exp(i \Theta_1) (1-P) \exp(i \Theta_2) P \exp(-i\Theta_1) P \exp(-i \Theta_2))).
\end{eqnarray}
Since $\Vert [P,\exp(i \Theta_1)] \Vert$
and $\Vert [P,\exp(i \Theta_2)] \Vert$
are both bounded by ${\mathcal O}(1/L)$, the
term inside the first trace in the above equation is equal to
$P\exp(i \Theta_2) (1-P) \exp(i \Theta_1) P \exp(-i\Theta_1) \exp(-i \Theta_2) P$, up to an error in operator norm which is of order $1/L^3$.  Since the dimension of the Hilbert space is equal to $L^2$, this means that the first trace
is equal to
${\rm Im}({\rm tr}(P\exp(i \Theta_2) (1-P) \exp(i \Theta_1) P \exp(-i\Theta_1) \exp(-i \Theta_2)))+{\mathcal O}(1/L)
=-{\rm Im}({\rm tr}(O_1^+ O_2^-)+{\mathcal O}(1/L)$.
Applying the same argument to the second trace in the above equation,
\begin{eqnarray}
&&{\rm Im}({\rm tr}([U,V] U^\dagger V^\dagger)).
\\ \nonumber
&=&
{\rm Im}({\rm tr}(O_1^+ O_2^--O_2^+ O_2^-))+{\mathcal O}(1/L).
\end{eqnarray}
Using the result (\ref{currentsinO}) this becomes
\begin{eqnarray}
\nonumber
\lim_{\epsilon_1\rightarrow 0^+} \int_{-\infty}^0 {\rm d}t_1 \exp(\epsilon_1 t_1) 
\lim_{\epsilon_2\rightarrow 0^+} \int_{-\infty}^0 {\rm d}t_2 \exp(\epsilon_1 t_2) 
{\rm Im}({\rm tr}(
J_2(t_2)(1-P)J_1(t_1)P-
J_1(t_1)(1-P)J_2(t_2)P).
\end{eqnarray}
Recognizing this as the 
Kubo  formula for the Hall conductance, we find two results.
\begin{lem}
On a torus, the matrix invariant is equal to the Hall conductance, up to ${\mathcal O}(1/L)$.
\end{lem}
This has the corollary:
\begin{corollary}
Consider a free fermion system on a torus topology.  Let the lattice
be a square lattice, of size $L$-by-$L$, with $L^2$ lattice sites.  Let the free fermion
Hamiltonian have hopping distance bounded above (uniformly in $L$)
by a constant, and spectral gap bounded below (uniformly in $L$) by a constant.  Then, the Hall conductance is within
${\mathcal O}(1/L)$ of an integer.
\end{corollary}
While better bounds are known and are valid for interacting systems\cite{hallint}, this seems a simple way to show
quantization of Hall conductance for non-interacting systems.  See also \cite{bellis} for the non-commutative geometry approach
to this problem and \cite{prodan} for a review of non-commutative geometry and topological insulators.

This result also implies that
\begin{lem}
Assume $U_1,U_2$ are obtained from a free fermion system on a torus topology.  Let the lattice
be a square lattice, of size $L$-by-$L$, with $L^2$ lattice sites.  Let the free fermion
Hamiltonian have hopping distance bounded above (uniformly in $L$)
by a constant, and spectral gap bounded below (uniformly in $L$) by a constant.  Then, for all sufficiently large $L$, the
matrix invariant agrees with the free fermion invariant in both the GUE and GSE universality classes.
\begin{proof}
We have just shown that the Chern number invariant, which is the same as the free fermion invariant,
agrees with the matrix invariant in the GUE case\cite{kitaev}.
Consider the GSE case.  
We use a similar argument as in the paragraph near Eq.~(\ref{add}),  adapted to the $\mathbb{Z}_2$ case.

Our first step is to construct a free fermion Hamiltonian in the GSE class
with free fermion invariant and matrix invariant both equal to $-1$.
Let $H_{GUE}$ be a free fermion Hamiltonian in the GUE class, with Chern number $+1$.  Define a free fermi Hamiltonian
$H_{GSE}$ by taking two copies of $H_{GUE}$:
\be
H_{GSE}=\begin{pmatrix} H_{GUE} \\ & \overline H_{GUE} \end{pmatrix},
\ee
where the overline denotes the complex conjugation.  By identifying the two copies with spin up and down, the Hamiltonian
$H_{GSE}$ is indeed time-reversal invariant and has free fermion invariant $-1$.  Since $H_{GUE}$ has odd Chern
number, the corresponding band projected position matrices of $H_{GUE}$have
odd matrix invariant.  The band projected position matrices of $H_{GSE}$
are obtained by doubling the band projected matrices of $H_{GUE}$; that is,
given matrices $H_r$ from Hamiltonian $H_{GUE}$, consider the matrices
$H_r'$ defined by
\be
H_r'=\begin{pmatrix} H_r & 0 \\ 0 & \overline{H_r} \end{pmatrix}.
\ee
Given that $H_r$ have odd matrix invariant, the doubled matrices
have $\mathbb{Z}_2$ invariant equal to $-1$, as desired (see theorem (\ref{doubledparity}).

Consider any free fermi Hamiltonian $H$ in the GSE universality class.  If its free fermion invariant is equal to $+1$, then
$H$ is equivalent, up to addition of trivial degrees of freedom, to a Hamiltonian $H'$ with localized Wannier functions (here we use a result of Kitaev in\cite{kitaev}, though the details of that proof are not yet published).  This implies that,
for sufficiently large $L$, 
the matrix invariant of $H'$ is also equal to $+1$.  This implies that the matrix invariant of $H$ plus trivial degrees of freedom is
equal to $+1$, which implies that the matrix invariant of $H$ is equal to $+1$.
Therefore, if the free fermion invariant is equal to $+1$, then the matrix invariant is also equal to $+1$ for sufficiently large $L$.

Suppose instead $H$ has free fermion invariant equal to $-1$.  Consider the Hamiltonian $H\oplus H_{GSE}$.  This
has free fermion invariant equal to $+1$.  Following the same argument as above, $H\oplus H_{GSE}$ has matrix
invariant equal to $+1$.  Using the fact that the matrix invariant of $H \oplus H_{GSE}$ is the product (using
the $\mathbb{Z}_2$ group multiplication rule) of the matrix invariant of $H$ with that of $H_{GSE}$, the
matrix invariant of $H$ is equal to $-1$.
\end{proof}
\end{lem}

\section{Chiral Classes}
The discussion above is entirely concerned with the classical universality classes.  In this section, we consider 
the three chiral symmetry classes (unitary, orthogonal, and symplectic).

We begin with the complex case.
In the complex case, the chiral class consists of Hamiltonians on a bipartite lattice, with non-zero hopping only between sites on
different sublattices.  Such Hamiltonians can be written as
\be
\label{chiraldecomp}
\begin{pmatrix}
0 & A \\
A^\dagger & 0
\end{pmatrix},
\ee
where the two blocks correspond to the odd and even sublattices, respectively.
In the complex case, chiral Hamiltonians have topological obstructions in odd dimensions, while non-chiral Hamiltonians
have topological obstructions in even dimensions\cite{kitaev}.

However, the topological obstructions for chiral Hamiltonian do {\it not}
arise from obstructions to construct localized Wannier functions, unlike
the other problems we study.
Consider a simple example in one-dimension.  We have a system with $N$ sites arranged on a ring, with
$N$ even.   Let Hamiltonian $H_0$ be
\be
H_0=\sum_{i=2k} \Psi^\dagger_i \Psi_{i+1} + h.c.,
\ee
and let $H_1$ be
\be
H_1=\sum_{i=2k+1} \Psi^\dagger_i \Psi_{i+1} + h.c.
\ee
These two Hamiltonians are both gapped, but they are in different topological phases: there is no way to find a continuous
path connecting these Hamiltonians while maintaining locality and while keeping the gap larger than ${\mathcal O}(1/N)$.
However, both of these Hamiltonians have localized Wannier functions.  For the Hamiltonian $H_0$, the localized Wannier functions
for the occupied states (we fix $E_F=0$ for all chiral Hamiltonians) are given by $N/2$ vectors, $v^k$, where $v^k$ has entries
$(v^k)_i=1/\sqrt{2}$ if $i=2k$, $(v^k)_{i+1}=-1/\sqrt{2}$ if $i=2k+1$, and $(v^k)_i=0$ otherwise.  That is, these vectors are
localized on sites $i,i+1$ for even $i$.  For Hamiltonian $H_1$, the Wannier functions are localized on sites $i,i+1$ for odd $i$.
Thus, the difference in the phases does not have to do with one Hamiltonian having localized Wannier functions and the other Hamiltonian
not having localized Wannier functions.

However, we can still quantify topological obstructions in chiral systems using almost commuting matrices in a different way.
We consider the case without time reversal symmetry first.
Let $H$ be a chiral Hamiltonian.  Spectrally flatten $H$ to define a Hamiltonian $H'$ whose eigenvalues are all equal to $\pm 1$.
Using $P$ as the projector onto negative eigenvalues of $H$, we have the relation
\be
H'=1-2P.
\ee

The Hamiltonian $H'$ is still chiral.  Suppose $H$ has a gap in the spectrum near $0$ so that all eigenvalues of $H$ are greater than
$\Delta E$ in absolute value; then $H'$ is still local so that if entries of $H$ decay superpolynomially in the distance between sites then so do entries of $H'$ and if entries of
$H$ decay exponentially then so do entries of $H'$ (this can be shown using
the same techniques used to prove exponential decay of correlation functions\cite{lsmhd}).
Further, since $H'$ is an odd function of $H$, then $H'$ is chiral given
that $H$ is chiral.

For definiteness, let us work
on a $d$-dimensional torus.  We define position matrices $\exp(i \Theta_i)$ for $i=1,...,d$ for sites on the even lattice exactly as
we did in the non-chiral cases.  
We then pair sites in the lattice, choosing pairs of neighboring two sites on opposite sublattices and considering them to be a ``pair'' of sites (we assume
the total number of sites in the lattice is even; if it is odd, then there are zero modes in $H$).  For example, in the
one-dimensional system above, we may choose to pair sites $1$ and $2$, sites $3$ and $4$, and so on.  
Since $H'$ is chiral, we can write $H'$ in the form (\ref{chiraldecomp}).  We use this pairing in writing $H'$ in the form (\ref{chiraldecomp}): if there are a total of $N$ sites in the lattice, we choose the $i+N/2$-th basis vector to be the pair of the $i$-th basis vector.  The
first $N/2$ basis vectors are in one sublattice and the last $N/2$ are in the other.
Since $H'^2=I$, the matrix $A$ is a unitary matrix.

Thus, we have constructed $d+1$ different unitary matrices: $d$ of these are the matrices $\exp(i \Theta_i)$, while the $d+1$-st is the
matrix $A$.  These matrices almost commute with each other, since $[\exp(i \Theta_i),\exp(i \Theta_j)]=0$ and using locality
properties of $P$ we can bound the operator norm of the commutator $[\exp(i \Theta_i),A]$.
Non-trivial topological obstructions can exist for $d+1$ unitaries.  For example, consider the Hamiltonian $H_0,H_1$ given above.
This is a system on a $1$-torus, so we have $2$ unitaries, both of size $N/2$-by-$N/2$.  One unitary is the diagonal matrix
\be
\exp(i \Theta)=\begin{pmatrix}
1 \\
& \exp(i2\pi/(N/2)) \\
&& \exp(i 4\pi/(N/2)) \\
&&&...
\end{pmatrix}
\ee
The other unitary is equal to
\be
A_0=\begin{pmatrix}
1 \\
& 1 \\
&& 1 \\
&&& ...
\end{pmatrix}.
\ee
for $H_0$, but it is equal to
\be
A_0=\begin{pmatrix}
0 & 1 \\
& 0 & 1 \\
&&& ...
1 & 0 &...
\end{pmatrix}
\ee
for $H_1$.
The matrices $\exp(i \Theta)$ and $A_0$ exactly commute.  The matrices $\exp(i \Theta)$ and $A_1$ almost commute (the operator norm
of the commutator is of order $1/N$), but cannot be approximated by exactly commuting matrices, as can be seen by computing the invariant $m$ from Eq.~(\ref{torusindex}).  In fact, these two matrices
$\exp(i \Theta)$ and $A_1$ are a previously considered example of a pair of almost commuting unitaries which cannot be approximated by exactly commuting unitaries\cite{Voiculescu}.

Two almost commuting unitaries are characterized by an integer invariant as described in the previous section on the torus, corresponding to the integer invariant known to describe
chiral systems in one dimension.  Note that this invariant of chiral systems provides an obstruction to connecting two Hamiltonians by
a path of gapped, local Hamiltonians.  Suppose $H_0,H_1$ are gapped, local Hamiltonians, connected by a smooth path $H_s$ of gapped
local Hamiltonians.  Since $H_s$ is gapped and local, the corresponding spectrally flattened Hamiltonian is local, and so the unitary
matrix $A_s$ is local.  Thus, the matrix $A_s$ approximately commutes with $\exp(i \Theta)$ for all $s$.  If the commutator
$[A_s,\exp(i \Theta)]$ is sufficiently small, then the integer invariant described above is indeed invariant under small changes
in $A_s$.  Thus, if the gap is sufficiently large, the invariant does not change along the path of Hamiltonians and so is
the same for $H_0$ and $H_1$.

The system of $d+1$ unitaries describes a soft $d+1$-torus.
In the torus case, we have obstructions for all $d \geq 1$, due to the possibility of lower dimensional obstructions, just as in the
case of weak topological insulators discussed in the non-chiral case, however the highest dimensional obstruction occurs only
for $d$ odd.
We can repeat the exercise on a sphere, instead of a torus.  In this case we obtain $d+1$ Hermitians
$H_i$ obeying the requirement that $\sum_i H_i^2=I$, and one unitary $A$ which almost commutes with the $H_i$. This gives a
soft $S^d \times S^1$.  The reduced $K_0$ of $S^d\times S^1$ is $\mathbb{Z}$ in all dimensions.  For $d$ even, our system is
always in the trivial case, because the topological obstructions for this space in even $d$ occur only if the
$d+1$ matrices describing $S^d$ have a topological obstruction, and in our case these $d+1$ matrices commute exactly.  So, we see integer
obstructions in odd dimensions and no obstructions in even dimensions.

One can also consider the chiral real and self-dual classes in this manner.  In this case, the matrices $A$ are orthogonal or
symplectic matrices, respectively.

The same mathematical problems are present in the case of sublattice symmetry as in the non-chiral case.  We need to show that the index obstructions
we have obtained are the only obstructions, and to construct explicit examples of sequences of almost commuting unitaries displaying the
different obstructions.  Finally, we need to identify this invariant with other known invariants.

\section{Numerical Results in Three Dimensions}
We now describe our numerical results on a three dimensional time reversal invariant topological insulator.
We considered the system on a three dimensional torus, using the polynomial map described previously to map to the sphere.

\subsection{Three Dimensional Hamiltonian}
In previous numerical work in two dimensions\cite{loringhastings},
we used the model of \cite{konig} which includes coupling between up
and down spin components due to breaking of bulk inversion symmetry, plus
an additional coupling to disorder.  The
model we study in three dimensions consists of a Hamiltonian
\be
\Ham=\Ham_0+V,
\ee
where $V$ represents on-site disorder described below.  The Hamiltonian $\Ham_0$ representing
the system without disorder is defined on a three dimensional torus.  Each
site has four different states, corresponding to two different bands and
to two different spins.  We introduce Pauli spin matrices $\sigma_{x,y,z}^{band}$ to describe band desgrees of freedom and $\sigma_{x,y,z}^{spin}$ to
describe spin degrees of freedom.  We set
\be
\Ham_0=A (i\partial_x \sigma_x^{band} \sigma_z^{spin}+i\partial_y \sigma_y^{band}+i\partial_z \sigma_x^{band} \sigma_x^{spin})+(M+B\partial^2) \sigma_z^{band},
\ee
where $A,B,M$ are numerical parameters, and $\partial_{x,y,z}$ is short-hand
notation for a lattice derivative.  That is, $\partial_x$ denotes a matrix whose matrix elements between sites $j$ and $k$ is equal to $1$ if site $j$ is one lattice site away from site $k$
in the $\hat x$-direction, equal to $-1$ if site $j$ is one lattice site
away from site $k$ in the $-\hat x$-direction, and $0$ otherwise.  Similarly,
$\partial^2$ is a lattice second partial derivative: its matrix element
between sites $j$ and $k$ is equal to $-6$ if $j=k$, equal to $+1$ if $j$ and
$k$ are nearest neighbors, and $0$ otherwise.
We chose $A=1,B=-1,M=-2$.  These parameters were chosen to obtain
a topologically nontrivial phase in the absence of disorder, with the
relation between $B$ and $M$ chosen to cancel the leading irrelevant term
$k^2$ in a continuum treatment of the problem.  This Hamiltonian
appears in \cite{3dham}.

The disorder term $V$ is a diagonal matrix, independent of spin and band index.  The disorder on a given site was chosen randomly from
the interval $[-4,4]$.  This is sufficiently strong to close the gap in the middle of the spectrum.  
The spectrum of $H_0$ is symmmetric about zero energy, and the distribution from which we draw $V$ is also symmetric about zero energy.
The statistical properties of the spectrum of $H$ are symmetric about zero energy.

We checked the localization
properties of the eigenvalues as follows.  For each disorder realization, for each eigenfunction, $\psi$, we computed the
variance in $\sin(\theta_i),\cos(\theta_i)$ for each of the three angles $\theta_1,\theta_2,\theta_3$ on the three torus.
We averaged over angles, computing
\begin{eqnarray}
&& \frac{1}{3}\Bigl( \sum_{i=1}^3
\langle \psi, \sin(\theta_i)^2 \psi\rangle+
\langle \psi, \cos(\theta_i)^2 \psi\rangle-
\langle \psi, \sin(\theta_i) \psi\rangle^2-
\langle \psi, \cos(\theta_i) \psi\rangle^2
\Bigr)
\\ \nonumber
&=&
1-\frac{1}{3} \Bigl( \sum_{i=1}^3
\langle \psi, \sin(\theta_i) \psi\rangle^2+
\langle \psi, \cos(\theta_i) \psi\rangle^2
\Bigr).
\end{eqnarray}.
We plot this variance as a function of energy in Fig.~(\ref{figvar}).  One can see crossings in the variance as a function of system size
at $E\approx \pm 0.65$ and $E\approx \pm 10.6$.  For $-10.6<E<-0.65$ and $0.65<E<10.6$, the variance increases as a function of system size, indicating
that the eigenfunctions are delocalized in this energy range, while outside this range the variance decreases, indicating that the eigenfunctions
are localized.  Let $E_{loc}\approx -10.6$ and $E_{loc}'\approx -0.6$ denote the localizations of these two localization transitions.

\begin{figure}[tb]
\label{figvar}
\centerline{
\includegraphics[scale=0.3]{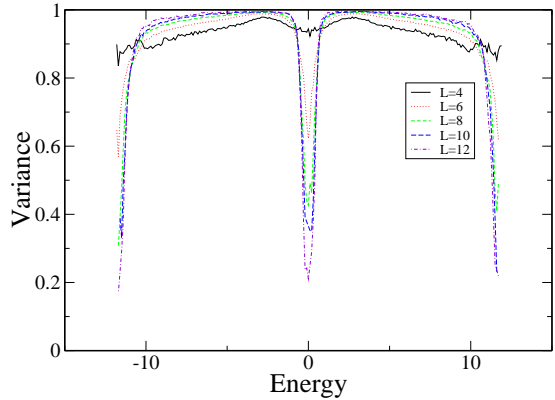}}
\caption{Plot of variance for $L=4,6,8,10,12$.}
\end{figure}

When we turn to the index as a function of energy, we encounter a surprise.  The expected behavior (which was seen in two dimensional systems
in \cite{hastingsloring}), is that the delocalized region of energies should precisely coincide with the region in which the index
fluctuates from sample to sample, while in the localized region the index should be the same in all samples.  The surprise we encounter
in three dimensions is that there is a range of energies over which the index is equal to $+1$ in all samples, but the eigenfunctions
are still delocalized.  In Fig.~(\ref{figpoly}) we plot the index computed using a polynomial mapping from torus to sphere.
For $E_F$ less than approximately $-4$, the index is equal to $+1$ for almost all samples, and it seems likely that in the thermodynamic
limit, the index will be equal to $+1$ for all samples for $E_F<E_c$ for some critical energy $E_c$ which is roughly $-5$.  Conversely,
it also seems likely that in the thermodynamic limit the index will be equal to $-1$ for all samples for $E_c'<E_F<0$, for some critical
energy $E_c'$.  This energy $E_c'$ appears to be at around $-0.6$.

Since for some small fraction of samples the index was not equal to $+1$ for $E_F=-4.9$ (the most negative $E_F$ plotted in the figure), we ran additional checks on a larger number of samples with $E_F=-5,6,7,8,9,10$.  We are able to compute the index more rapidly for the
samples with more negative $E_F$ since the number of occupied states decreases.  We studied $\approx 2700$ samples for $L=10$ and
$\approx 1800$ samples for $L=12$.  We found that for $L=10$, the average index was equal to $0.998...$ for $E_F=-7$ and was equal to $1$ for
$E_F\leq -8$, and for $L=12$ the average index was equal to $0.992...$ for $E_F=-6$ and was equal to $1$ for $E_F\leq -7$.  This provides
strong numerical evidence that in the thermodynamic limit the index is indeed equal to $1$ for $E_F\leq E_c \approx -5$.  Perhaps $E_c$
is slightly smaller than $-5$, but certainly the numerical evidence indicates that it does {\it not} coincide with the localization transition.

Thus, the critical energy $E_c'$ where the index becomes equal to $-1$ appears to coincide with the transition to localized
states at $E_{loc}'$, while the critical energy $E_c$ appears to be in the middle of the delocalized region.
This is surprising, but does not contradict any of the known properties of the index.  
Note that the index must
be trivial ($+1$) if all of the occupied states are localized, because in this event we know that the system has localized Wannier functions:
the energy eigenfunctions themselves supply  a basis of localized Wannier functions.  Hence, the index must  be equal to $-1$ for $E_F<E_{loc}$.  

Similarly, the fact that the index is fluctuating from sample-to-sample for 
$E_c<E_F<E_c'$ implies that the system must be in a delocalized phase in this region.  To state this claim in a  more mathematically
precise fashion, the fact that the index
fluctuates from sample-to-sample and monotonically decreases with $E_F$, implies that for any small interval of energies,
$[E,E+dE]$, with $E_c<E<E_c'$, there is some non-zero probability that a given sample will have a transition in index from $+1$ to $-1$
in that interval, which means that there is some non-zero probability that a state in that interval will be delocalized.

The surprising thing, then, is that the index can be equal to $+1$ for all samples even when the system is delocalized.
However, perhaps this should also not be surprising.  Recall the  phenomenon in the two-dimensional case\cite{loringhastings}.
There, the index always fluctuated from sample-to-sample in a delocalized phase.  However, the average index varied smoothly as a function
of energy in the delocalized phase.  The average was equal to $+1$ exactly at the lower critical energy $E_c$, and the average was close to
$+1$ cose to the the lower critical
energy $E_c$.  Thus, the system was delocalized and yet the index was be equal to $+1$ in {\it almost all} samples.  Thus, perhaps it
should not be surprising that one can have a system in which the index is equal to $+1$ in all samples.

The question that this raises is how to interpret the index in a delocalized phase.  In a localized phase, the index
quantifies topological properties of the system, and we have proven that it does not change under small change in the Hamiltonian
(if the Hamiltonian is localized, small changes in the Hamiltonian only lead to small changes in the projector $P$ and hence to
only small changes in $H_r$ and in $B(H_1,H_{d+1})$).  In the delocalized phase, we do not have a good interpretation of our index.  However, it has been
pointed out to us\cite{kl} that the index is well-defined so long as the
entries of $P$ drop off as a sufficiently rapid power of distance (the
particular power depends on spatial dimension) and hence one can have a delocalized phase with a well-defined index.
However, the calculation of the
index reveals an unsuspected critical point in the system in three dimensions, a transition at $E_c$ which appears not to have
any signature in the localization properties of the system.  Thus, it would be interesting to understand this transition from a field theoretic point of view.

\begin{figure}[tb]
\label{figpoly}
\centerline{
\includegraphics[scale=0.3]{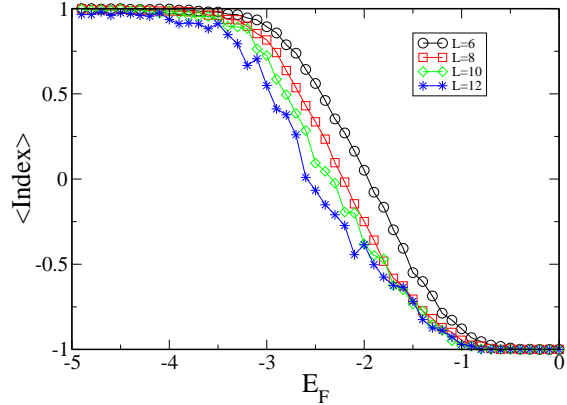}}
\caption{Plot of average index for $L=6,8,10,12$ (circle,square,diamond,star), using polynomial map from the three-torus to the three-sphere.  Each data point
is an average of $1700,1400,600,400$ samples, respectively.}
\end{figure}

\begin{figure}[tb]
\label{figbott}
\centerline{
\includegraphics[scale=0.3]{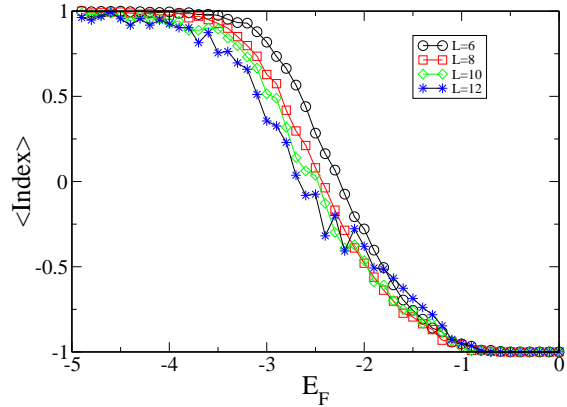}}
\caption{Plot of average index for $L=6,8,10,12$ (circle,square,diamond,star), using logarithmic map from torus to sphere.  Each
data point is an average of $1300,1000,300,270$ samples, respectively.}
\end{figure}

\begin{figure}[tb]
\label{figcompare}
\centerline{
\includegraphics[scale=0.3]{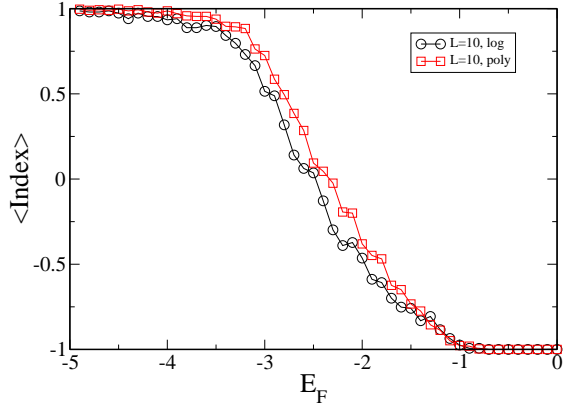}}
\caption{Plot of average index for $L=10$, comparing logarithmic and polynomial maps (circle,square).}
\end{figure}

\subsection{Particle-Hole Symmetry and Index}
In the calculation in this section, we computed the index only for $E_F<0$.  The reason to do this is that the dimension of the matrix
$B$ is equal to twice the number of occupied states.  For $E_F<0$, this number of occupied states is less than half the dimension of
the Hamiltonian, so that the dimension of $B$ is at most equal to the dimension of the Hamiltonian.  This restriction on the
size of the matrix $B$ is important in making it practically possible to consider large system sizes.  We implemented the algorithm
by first diagonalizing the Hamiltonian.  Then, we computed the index for a sequence of values of $E_F$, starting at $E_F=0$ and reducing by
$0.1$ at each step.  As $E_F$ gets more negative at every step, the dimension of $B$ decreases, so that the calculation of
the index is faster for more negative values of $E_F$.

However, one may wonder what happens for $E_F>0$.  Is there new physics?  In this subsection, we discuss the implications of particle-hole
symmetry for $E_F>0$ and argue that we can in fact obtain all the desired information from $E_F<0$. 
Assume that we have some system with
a given $E_F$.  Let $P$ be the projector onto the space of occupied states.  Then, $1-P$ is the projector onto the space of
unoccupied states.
We can form the matrices
\be
H_r=PX_a P,
\ee
and then use them to construct $B$.
However, we could instead form the matrices
\be
H_r'\equiv (1-P) X_a (1-P),
\ee
where $(1-P)$ is the projector onto the space of empty states.
We need the following lemma:
\begin{lem}
For any symmetry class from GUE,GOE,GSE and any dimension, there is a numeric constant $c>0$ such that the following holds.
Assume that the commutator $\Vert [X_a,P] \Vert<c$ for all $a$ and that $\sum_a X_a^2=I$ and $[X_a,X_b]=0$.
Compute the the index from the matrices $H_r'$ and call this $I'$.
Compute the index from the matrices $H_r$ and call this $I$.  Then, 
the product of $I$ with $I'$, using the appropriate group
multiplication law (either $\mathbb{Z}$ or $\mathbb{Z}_2$), is equal to the trivial index.
Equivalent, $I=I'$ in the $\mathbb{Z}_2$ case, while $I=n$ and $I'=-n$ for some integer $n$ in the $\mathbb{Z}$ case.
\begin{proof}
We need two properties of the index, proven elsewhere\cite{hastingsloring}.  First, we use the fact that under direct sum of matrices, the
index combines using multiplication as discussed above Eq.~\ref{add} in the $\mathbb{Z}$ case.
That is, let $I(\{H_r\})$ denote the index computed from 
a given set of matrices $H_r$.  Given any matrices $H_r$ and $H_r'$, we have
\be
I(\{H_r\})\cdot I(\{H_r'\})=I(\{H_r\oplus H_r'\}),
\ee
where the product is group multiplication and where
\be
H_r\oplus H_r'=\begin{pmatrix} H_r \\ & H_r' \end{pmatrix}.
\ee
So, to show the claim it suffices to show that $I(\{H_r \oplus H_r'\})$ is trivial.
However, we have
\be
H_r\oplus H_r'=X_a-PX_a(1-P)-(1-P)X_aP.
\ee
If $\Vert[X_a,P] \Vert<c$, then $\Vert P X_a (1-P) \Vert<c$.  So,
$\Vert H_r\oplus H_r'-X_a \Vert <2c$.  The matrices $X_a$ have trivial index.  The second property of the
index we need is that one can prove a lower bound on the distance in
operator norm from the $X_a$ to the nearest set of matrices with nontrivial index (proven in \cite{hastingsloring} for the case of $d=2$ in
two different symmetry classes, and stated without proof here for other cases since the proof is essentially the same for
all cases).  Thus, if $c$ is smaller than this
lower bound, the desired result follows.
\end{proof}
\end{lem}

Then, suppose the commutator $\Vert[X_a,P]\Vert$ is indeed sufficiently small, as would hold in the localized regime.
Then, using the lemma, we can either compute the index by consider the space of
occupied states or the space of empty states, choosing whichever space has the smaller dimension to make the numerics simpler.
In the case of the three dimensional system given here, for $E_F>0$, we would choose to compute the index from the space of
empty states.  However, using symmetry properties of the Hamiltonian and the disorder distribution, one can show that
the statistical properties of the index computed from the empty states at energy $E_F$ are identical to the
statistical properties of the index computed from the occupied states at energy $-E_F$.  Thus, assuming the commutator
is sufficiently small, indeed we do not need to consider the properties of the system with $E_F>0$, at least in the localized regime.
However, on physical grounds we believe that even outside the localized regime the indices $I(\{H_r\})$ and $I(\{H_r'\})$ are
related in the same way, so long as the projector $P$ is obtained from a local Hamitonian.

\section{K-Theory}

We have described many of our invariants in terms of linear algebra
and matrix functions (functional calculus to pure mathematicians).  This
makes finding efficient algorithms, but hides where these invariants
are coming from.

In finite models, the position observables are bounded Hermitian
operators that commute.  The spectral theorem applies.
This immediately brings into the discussion a commutative algebra
of functions on a space. Let us adopt the notation 
\[ C(X) = \left\{ \left.
f:X  \rightarrow \mathbb{C} 
\strut \ \right | \ 
f \mbox{ is continuous }
\right \} \]
whenever $X$ is a compact Hausdorff space.
Given commuting $n$-by-$n$ Hermitian matrices $H_1,\ldots,H_k,$
of norm at most one,
the spectral theorem states that there is
a unitary $W$  so that $H_r = WD_rW^{\dagger}$ with $D_r$ diagonal,
\[
D_r = \mathrm{diag}(\lambda_{r,1},\ldots,\lambda_{r,n}).
\]
The requirement that the $H_r$ be {\em contractions} (norm at most one)
translates to the restriction $-1 \leq \lambda_{r,j} \leq 1.$
For $f$ in $ C([-1,1]^{k}) $ we now are able to define
\[   f(H_1, \ldots ,H_k) =
  W \mathrm{diag}(\lambda_{r,1},\ldots,\lambda_{r,1}) W^{\dagger}.
\]
On the one hand, this is a nice form of functional calculus (multivariate
matrix functions).  One the other, we get a homomorphism of algebras
\begin{equation}
\label{funCalcMulti}
\Phi : C([-1,1]^{k}) \rightarrow \mathbf {M}_k (\mathbb {C} ) 
\end{equation}
defined by
\[ \Phi (f) = f(H_1, \ldots ,H_k),
\]
where
$ A \subseteq \mathbf{M}_n (\mathbb{C}) $
is the unital complex algebra generated by the $H_r.$ 

If the $H_r$ 	satisfy some equation, that equation allows us
to cut down the space $[-1,1]^{k}$ to something more interesting.
For example, if
\[ H_1^2 + H_2 ^2 = H_3^2 + H_4 ^2 = I  
\]
then we conclude
\[ \lambda_{1,\ell}^2 + \lambda_{3,\ell} ^2 = \lambda_{3,\ell}^2 + \lambda_{4,\ell} ^2 = 1  
\]
and so we get a homomorphism
\[ \Psi :C(\mathbb {T}^2 ) \rightarrow \mathbf {M}_k (\mathbb {C} ).
\]

We often think of this in terms of two commuting unitaries,
$U=H_1 + iH_2$ and $V=H_3 + iH_4.$  From two commuting unitary
matrices $U$ and $V$ in $\mathbf {M}_k (\mathbb {C} )$ we derive
a homomorphism
\[ \Psi :C(\mathbb {T}^2 ) \rightarrow \mathbf {M}_k (\mathbb {C} ).
\]
For Laurent polynomials $\Psi$ operates as expected.  For example
\[
\Psi \left( \strut
\cos(2\pi i \theta_1) + e^{2\pi i \theta_1}e^{-1\pi i \theta_2}
\right)
=
\tfrac{1}{2}(U + U^\dagger) + UV^{-1}.
\]

The algebras $A,$ $\mathbf{M}_n (\mathbb{C})$ and $C([-1,1]^{k})$
are examples of complex $C^*$-algebras.
The category of complex $C^*$-algebras gives us enough room to study
algebras of matrices generated by almost commuting Hermitian matrices.
If we wish to keep track of TR-symmetry via self-dual matrices, we need
to switch to the study of real $C^*$-algebras.  These can be thought
as complex $C^*$-algebras that carry an extra operation, similar to
the matrix transpose or matrix dual.

\subsection{Complex $C^*$-Algebras}

$C^*$-algebras are algebras over the complex scalars that
are equipped with a norm and an involution $x \mapsto x^*.$
Only in the context of real $C^*$-algebra does one generally specify
{\em complex} $C^*$-algebra.

There are axioms relating the algebraic and norm
structures, c.f.~\cite{BlackadarOpAlg}; the idea is that the
norm is modeled on the operator norm, while the involution
is modeled on the adjoint of a complex matrix or of an
operator on Hilbert space.  The adjoint is written as $A^*$ in
mathematics literature and as $A^\dagger$ in physics literature.  We will
use $*$ for the involution in an abstract $C^*$-algebra,
and $\dagger$ when dealing with finite matrices or actual
operators on Hilbert space (for example, when considering matrices above, we used $\dagger$ for the adjoint throughout).

The most basic mappings between $C^*$-algebras are the $*$-homomorphisms.
A $*$-homomorphism $\varphi:A \rightarrow B$ must be linear and satisfy
the axioms
\[
\varphi (ab) = \varphi(a)\varphi(b)   
\]
and
\[
\varphi (a^*) = \varphi(a)^*.
\]
In examples the involution should be represented appropriately.
For example, the involution in $C(X)$ is conjugation, the
value of $\bar{f}$ at $x$ is $\overline {f(x)}.$
The map $\Phi$ in (\ref{funCalcMulti}) is a $*$-homomorphism because
\[
\Phi(\bar{f}) = \Phi(f)^\dagger
\]

Continuity of a $*$-homomorphisms is automatic, as
is the fact that when $\varphi$ is $*$-homomorphism we have
$\|\varphi(a)\| \leq \| a \| .$

\subsection{Complex $K_0$}

There are two groups $K_0(A)$ and $K_1(A)$ that make up the $K$-theory
of a complex $C^*$-algebra $A.$  A straight-forward, common construction
of $K_0(A)$ uses projectors (a.k.a. projections, a.k.a. self-adjoint idempotents)
and homotopy.  The simplest construction of $K_1(A)$ uses homotopy and
unitaries.  The ability to build up matrices is used to gain a commutative
group operation.

Assume $A$ to be a unital $C^*$-Algebra.
Were $A$ not unital, we would need to consider
the so-call reduced $K_0$ group.  We are glossing over such details in
this short summary of $K$-theory in the complex
case.  See~\cite{RordamBook} for details.
An alternate picture can be used, where invertible Hermitian
elements replace projectors.  Given $x$ invertible with $x^* = x$
we can associate the projection
\[
p = \tfrac{1}{2} + \tfrac{1}{2}\mathrm{polar}(x).
\]
When $x$ is self-adjoint, $\mathrm{polar}(x)$
will be self adjoint, which forces $p$  to have spectrum
in $\{0, 1\}$ and to be
a projector.

Let
\[
\mathrm{GL}_n^{\mathrm{odd}}(A) = \mathrm{GL}_{n}(A) 
=
\left\{ x \in  \mathbf{M}_n(A) \left |\, x^-1\mbox{ exists} \right.\right\}
\]
and
\[
\mathrm{GL}_n^{\mathrm{ev}}(A)
=
\left\{ x \in  \mathrm{GL}_n(A) \left |\, x^*=x \right.\right\}
\]
To build $K_0(A),$ one considers the union over all even $k$ of the sets
$
\mathrm{GL}_{k}^{\mathrm{ev}}(A)
$
and then forms equivalence classes denoted $[x].$  The equivalence
relation is defined by requiring
\[  
x \sim y \implies [x] = [y]
\]
where $\sim$ is to indicate homotopy within the set of invertible
Hermitian elements and
\[
\left[\begin{array}{ccc}
x\\
 & 1\\
 &  & -1
\end{array}\right]
=
\left [ x \right] .
\]

The group $K_0(A)$ consists of all the classes $[x]$ and formal differences
of classes $[x] - [y].$  The addition is determined by the rule
\[
[x] + [y] = 
\begin{pmatrix}
x & 0\\
0 & y
\end{pmatrix},
\]
while the inverse is formal,
\[
-([x] - [y]) = [y] - [x]. 
\]

An example of a $K_0$ class when $A$ is the commutative $C^*$-algebra
$A = C(X)$  is that associated to a vector bundle.  If that vector
bundle is a sub-bundle of the trivial bundle $\mathbb{C}^k \times X$
then there is a continuously varying projector $p_x$ so that the fiber
over $x$ is $p_x(\mathbb{C}^k).$  The class in $K_0(C(X))$  associated
to this vector bundle is $[2p-1]$ where $p$ is the function $p(x)=p_x$ which
is a projector in 
\[
\mathbf {M} _k (C(X)) = C(X,\mathbf {M} _k (\mathbb {C})) . 
\]

A more basic example is $K_0(\mathbb {C}).$  In this case we
get an isomorphism
\begin{equation}
\label{TraceForK0}
\xymatrix{
\mathrm{Sig}_* :  K_{0}(\mathbf {M}_{n}(\mathbb{C}))
		\ar[r] ^(0.73){\cong} 
	& \mathbb {Z} 
}
\end{equation}
determined by the signature 
$
[X] \mapsto \mathrm{Sig}\left(X\right)
$
as in Definition~\ref{signatureDef}.

\subsection{Complex $K_1$}

We can build $K_1(A)$ with a little less work.
First consider the union over all $k$ of the sets
$
\mathrm{GL}_{k}(A)
$
and form equivalence classes denoted $[x].$  The equivalence
relation is defined by requiring
\[  
x \sim y \implies [x] = [y]
\]
where $\sim$ is to indicate homotopy within the set of invertible
elements and
\[
\left[\begin{array}{cc}
x\\
 & 1
\end{array}\right]
=
\left [ x \right] .
\]

The inverse is
\[
- [x] = [x^{-1}]
\]
which works because for $x$ and $y$ invertible elements
in the same
$\mathbf{M}_{k}(A)$ the addition has two equivalent formulas
\[
[x] + [y] = 
\begin{pmatrix}
x & 0\\
0 & y
\end{pmatrix}
= [xy]
\]
as is shown here:
\[
\left(\begin{array}{cc}
x & 0\\
0 & y\end{array}\right)
=
\left(\begin{array}{cc}
x & 0\\
0 & 1\end{array}\right)\left(\begin{array}{cc}
0 & 1\\
1 & 0\end{array}\right)^{-1}\left(\begin{array}{cc}
y & 0\\
0 & 1\end{array}\right)\left(\begin{array}{cc}
0 & 1\\
1 & 0\end{array}\right)
\sim
\left(\begin{array}{cc}
xy & 0\\
0 & 1\end{array}\right)
\]

All invertibles in $\mathbb {M} _k (\mathbb {C} )$ are homotopic---just
diagonalize an invertible matrix and push the eigenvalues to $1$---so 
\[
K_1(\mathbb {C}) = \{ 0 \}
\]
where
$
[1] = 0.
$
This stops looking strange after a while.

Without proof, we claim $K_1(C(\mathbb{T} ^2 ) ) \cong \mathbb {Z} ^2 $ with
the two generators being represented by the unitaries
$u$ and $v$ in $C(\mathbb{T} ^2 )$ given by
$u(e^{2\pi i\theta_1},e^{2\pi i\theta_2}) =  e^{2\pi i\theta_1} $
and
$v(e^{2\pi i\theta_1},e^{2\pi i\theta_2}) =  e^{2\pi i\theta_2} .$

\subsection{Pushing forward $K$-theory classes}

If we have a $*$-homomorphism
$\varphi:A \rightarrow B$ between $C^*$-algebras,
it extends via componentwise action to $*$-homomorphisms
$\varphi_k : \mathbb {M}_k (A) \rightarrow \mathbb {M}_k (B).$
These send $\mathrm{GL}_{k}(A)$ to $\mathrm{GL}_{k}(B)$ and
$\mathrm{GL}_{n}^{\mathrm{ev}}(A)$ to $\mathrm{GL}_{n}^{\mathrm{ev}}(B)$ and
so induces two group homomorphisms
\[
\varphi_* :K_d(A) \rightarrow K_d(B).
\]

With the examples at hand, the $C(X)$ and $\mathbb {M}_n (\mathbb {C}),$
this is tremendously boring.  It is only by considering something
less rigid than a $*$-homomorphism that something interesting can
occur.  This dichotomy is the basis for our using $K$-theory as a
means to distinguish trivial an non-trivial phases of models of topological
insulators.

We now give an exact statement about $*$-homomorphisms
$\varphi : C(C) \rightarrow \mathbb {M}_n (\mathbb {C}).$
The proof uses the fact that homotopic homomorphisms
induce the same map on $K$-theory.  By a {\em homotopy of
$*$-homomorphisms} we mean $\varphi_t:A \rightarrow B$ where
each $\varphi_t$ is a $*$-homomorphism and with $a$ fixed in $A,$
the map $t \mapsto \varphi_t(a)$ is continuous.
In the following, since $C(X)$  is involved, it is 
more natural to think of $K_0$ in terms of projectors.
For projectors in $\mathbf {M}_k ( \mathbb {C} )$
their class in $K_0$ is determined by the trace.
 
\begin{thm}
\label{thm:GUEtrivialIndex}
Suppose $X$ is a pathwise connected compact metric space.
Define $I(X)$ as subset of $K_0(C(X))$ all elements
of the form $[p] - [q]$ where
\[
\mathrm{Tr}(p(x)) = \mathrm{Tr}(q(x)) 
\]
for all $x.$  For any unital $*$-homomorphism
\[
\varphi : C(X) \rightarrow \mathbf {M}_k ( \mathbb {C} )
\]
the induced map
\[
\varphi _* : K_0(C(X)) \rightarrow K_0 (\mathbf {M}_k ( \mathbb {C} ))
\]
sends all of $I(X)$  to zero.
\end{thm}

\begin{proof}
The representation theory of $C(X)$ is well known.  By
decomposing $\varphi$ into irreducible representations
we find a unitary $W$ and points $x_1,\ldots,x_k$ in $X$
so that
\[
 \varphi (f) = W \left [ 
\begin{array}{cccc}
f(x_{1})\\
 & f(x_{2})\\
 &  & \ddots\\
 &  &  & f(x_{k})\\
\end{array}
 \right ] W^\dagger .
\]
We can deform $W$ through unitaries to the identity, and deform
the $x_j$ along paths to $x_1$ and so, as far at $K$-theory is
concerned, $\varphi$ might as well be the $*$-homomorphism $\psi$
where
\[
\psi (f) = f(x_1)I.
\]
Then
$\varphi_n(p) = p(x_1)\otimes I$ and $\varphi_n(q) = q(x_1)\otimes I$ so
\[
\mathrm{Tr}(\varphi_n (p)) =  
k\mathrm{Tr}(p(x_1)) =
k\mathrm{Tr}(q(x_1)) =
\mathrm{Tr}(\varphi_n (q)) .
\]
We are done, by (\ref{TraceForK0}).
\end{proof} 

If we replace $*$-homomorphisms by mappings that are ``almost multiplicative''
but otherwise like homomorphisms, we can get richer induced maps on
$K$-theory.  This is the basis for $E$-theory as introduced by Connes and Higson \cite{connesHigson}.  We prefer the simpler approach of defining $C(X)$
by relations and consider our ``soft representations'' as replacements
for $*$-homomorphisms from $C(X)$ to $C^*$-algebras of matrices.  These
soft representations can induce mappings at the level of $K$-theory that
actual $*$-homomorphisms cannot.

The induced mappings work roughly as follows. A detailed description
would require a discussion of generators and relations for $C^{*}$-algebras,
a theory yet to be developed for real $C^{*}$-aglebras. Given $u=(u_{ij})$
in $\mathbf{M}_{k}(C(X))$ with $u^{-1}=u^{*}=u$ we find formulas
for the component function 
\[
u_{ij}=u_{ij}(x_{1},\ldots,x_{d+1})
\]
so that we interpret via functional calculus the meaning of 
\[
u_{ij}(a_{1},\ldots,a_{d+1})
\]
formulas for $\delta$-representation $a_{1},\ldots,a_{d+1}$ of the
$d$-sphere in a $C^{*}$-algebra. When $\delta$ is small, we will
find
\[
U=U(a_{1},\ldots,a_{d+1})
\]
satisfies $U^{-1}\approx U^{*}=U$ and so have a well defined element
of $K_{0}(A).$ In particular, we were working in the GUE case over
the sphere with
\[
u=\left[\begin{array}{cc}
x_{3} & x_{1}-ix_{2}\\
x_{1}+ix_{2} & -x_{3}\end{array}\right]\in\mathbf{M}_{2}(C(S^{2})).
\]

Moving to the $3$-sphere, we have a $K_{1}$-element determined by
the unitary
\[
u=\left[\begin{array}{cc}
x_{3}+ix_{4} & -x_{1}+ix_{2}\\
x_{1}+ix_{2} & x_{3}-ix_{4}\end{array}\right]\in\mathbf{M}_{2}(C(S^{3}))
\]
and given $a_{1},\ldots,a_{4}$ in a $C^{*}$-algebra $A$ that are
a soft-representation of the $3$-sphere we get a representative of
a $K_{1}$-class in $A$ using
\[
\left[\begin{array}{cc}
a_{3}+ia_{4} & -a_{1}+ia_{2}\\
a_{1}+ia_{2} & a_{3}-ia_{4}\end{array}\right].
\]
This leads to nothing of substance when $A=\mathbf{M}_{n}(\mathbb{C}),$
as this has $K_{1}$ equal to zero. This pushing forward become interesting
when we switch to the study of real $C^{*}$-algebras.

\subsection{Real $C^*$-Algebras}

There are two competing
mathematical objects forming the subject called real $C^*$-algebras.
The one we fill emphasize are what we call $C^{*\tau}$-algebras, which
are $C^*$-algebras with an extra operation.

A $C^{*\tau}$\emph{-algebra} is a pair $(A,\tau)$ where $A$ is
a $C^*$-algebra and $\tau$ is a $*$-isomorphism 
\[
\xymatrix{
 \tau : A \ar[r]^(0.52){\cong} & A^{\mathrm{op}}
}
\]
subject to the axiom $\tau(\tau(a))=a.$ We regard $A^\mathrm{op}$
as having the same underlying set as $A.$
The opposite of a $C^*$-algebra is defined as
having the original algebraic operations, except that the
new multiplication applies the original multiplication in
reverse order.
We may just as well treat
$\tau$ as a unary operation $a\mapsto a^\tau$ that satisfies all
the properties of being a $*$-homomorphism, except the multiplicative
rule becomes
\[
\left(ab\right)^{\tau}=b^{\tau}a^\tau,
\]
and with the axiom
\[
\left(a^{\tau}\right)^{\tau}=a.
\]

Since $*$ and $\tau$ commute, we write $a^{*\tau}$ where many authors
write $\overline{a}.$ The advantage of our notation is that one $C^{*}$-algebra
can have $\tau_{1}\neq\tau_{2}$ that make it a $C^{*\tau}$-algebra
and lead to different $\overline{a}$ operators. In the case of
$A=\mathbf{M}_{2N}(\mathbb{C})$ there are two possible $\tau$ operations,
the dual and the transpose, one of which leads
to $\overline{T}$ being the expected entry-wise
conjugate operation, the other does not.

Mathematicians generally work with the {\em real part}
of $A,$ denoted $ \Re_\tau (A)$ and defined as
\[
\Re_{\tau}(A)
=
\left\{ a\in A\left|\, a^{\tau}=a^{*}\right.\right\} . 
\]
This is a $*$-algebra over the scalar field $\mathbb{R}.$ We take
here a definition:  an {\em $R^*$-algebra} is a normed $*$-algebra
over $\mathbb{R}$ that equals $\Re^{\tau}(A)$
for some $C^{*\tau}$-algebra $A.$ Two essential examples in physics
are
\[
\Re^{\mathrm{T}}\left(\mathbf{M}_{2N}(\mathbb{C})\right)
=
\mathbf{M}_{2N}(\mathbb{R})
\]
and
\[
\Re^{\sharp}\left(\mathbf{M}_{2N}(\mathbb{C})\right)
\cong
\mathbf{M}_{N}(\mathbb{H})
\]
where $\mathbb{H}$ is the algebra of quaternions. The isomorphism
is not terribly elucidating, so we do not discuss it. We advocate,
however, calling the condition $X^{\sharp}=X^{\dagger}$
the \emph{quaternion
condition}. It is easy to see the
matrices that satisfy the quaternion condition are all those of the form
\[
\left(\begin{array}{cc}
A & -\overline{B}\\
B & \overline{A}
\end{array}\right).
\]

\begin{rem}
The standard term for an $R^{*}$-algebra is a real $C^{*}$-algebra.
A $C^{*\tau}$-algebra is generally called a ``real'' 
$C^{*}$-algebra, where the quotation marks are part of the
name \cite{kasparov1981operator}, and the additional operation is emphasized
is $a\mapsto a^{*\tau}.$ We think our notation is more fitting in
this context, as the $a\mapsto a^{\tau}$ operation satisfies the
axioms of the transpose $A\mapsto A^{\mathrm{T}}.$
\end{rem}

An essential feature of complex $C^*$-algebras is the
matrices over them are again $C^*$-algebras, in a unique
way.  In the case of $C^{*\tau}$-algebras, we lose uniqueness,
but have a canonical structure.  Given $(A,\tau)$ the $\tau$
operation, also denoted $\tau,$ of $\mathbb{M}_n(A)$ is
$\mathrm{T}\otimes\tau$ on $\mathbf{M}_{n}(\mathbb{R})\otimes A=\mathbf{M}_{n}(A),$
so
\[
\left(a_{jk}\right)^{\tau}=\left(a_{kj}^{\tau}\right).
\]

\subsection{$K_0$ $K_1$ and $K_2$ for real $C^*$-algebras}

For present purposes, we need to define $K_n$ for a
$C^{*\tau}$-algebra for $n=0,1,2$ and $n=4,5,6.$
We have no need for $K_3$ and $K_7$ and omit them.
We deal with $n=0,1,2$ in this
subsection in a manner that makes them most look like
signature, determinant and Pfaffian.

For  $n<0$  or $n>7$ we rely on the order-$8$ Bott
periodicity, and so take $K_{-5}$ to be defined to 
equal $K_{3}.$  In the next subsection we show
how to deal with the cases $n=4,5,6$ 
by an old trick involving the quaternions.

Consider how we constructed $K$-theory in the complex
case, in Table~\ref{tab:ComplexK-theory}.  Within
the sets of invertible elements, we made a choice
on what symmetry to require---none or self-adjoint.
We needed a neutral matrix to know how to embed a
set of smaller matrices into a larger set. 
We define $K_{n}^\tau(A),$ with alternate notation
$K_{n}(\Re_\tau(A))$ or  $K_{n}(A, \tau)$ to fit
in context, by repeating what we did in the complex 
case, but with the symmetries and neutral elements
as indicated in Table~\ref{tab:RealK-theory}.
The required subsets of invertible matrices are
the following:
\[
\mathrm{GL}_{n}^{[0]}(A,\tau)
=
\left\{ x\in\mathrm{GL}_{n}(A)\left|\strut\, x^*=x^\tau, \ x^*=x\right.\right\} 
\]
\[
\mathrm{GL}_{n}^{[1]}(A,\tau)
=
\left\{ x\in\mathrm{GL}_{n}(A)\left|\strut\, x^*=x^\tau \right.\right\} 
\]
\[
\mathrm{GL}_{n}^{[2]}(A,\tau)
=
\left\{ x\in\mathrm{GL}_{n}(A)
\left|\strut\, x^* = x^\tau,\ x^*= -x
\right.\right\} 
\]

\begin{table}
\begin{tabular}{|c|c|c|c|>{\centering}p{1.6in}|}
\hline 
 & Invertibles used & Neutral element & One-dimensional example & Classical form\tabularnewline
\hline
\hline 
$K_{0}(A)$ & $x^{*}=x$ & $\left(\begin{array}{cc}
1 & 0\\
0 & -1\end{array}\right)$ & $K_{0}(\mathbb{C})=\mathbb{Z}$ & $\mathrm{Sig}$\tabularnewline
\hline 
$K_{1}(A)$ & no restrictions & $1$ & $K_{1}(\mathbb{C})=0$ & --\tabularnewline
\hline
\end{tabular}
\caption{
Complex $K$-theory.
\label{tab:ComplexK-theory}
}
\end{table}

\begin{table}
\begin{tabular}{|c|c|c|c|>{\centering}p{1.6in}|}
\hline 
 & $\mathrm{GL}_{n}^{[j]}(A,\tau)$ & Neutral element & One-dimensional example & Classical form\tabularnewline
\hline
\hline 
$K_{0}^{\tau}(A)$ & $x^{*}=x^{\tau},\ x^{*}=x$ & $\left(\begin{array}{cc}
1 & 0\\
0 & -1\end{array}\right)$ & $K_{0}(\mathbb{R})=\mathbb{Z}$ & Signature or trace\tabularnewline
\hline 
$K_{1}^{\tau}(A)$ & $x^{*}=x^{\tau},$ & $1$ & $K_{1}(\mathbb{R})=\mathbb{Z}_{2}$ & Determinent \tabularnewline
\hline 
$K_{2}^{\tau}(A)$ & $x^{*}=x^{\tau},\ x^{*}=-x$ & $\left(\begin{array}{cc}
0 & 1\\
-1 & 0\end{array}\right)$ & $K_{2}(\mathbb{R})=\mathbb{Z}_{2}$ & Pfaffian\tabularnewline
\hline
\end{tabular}
\caption{
Real $K$-theory.
\label{tab:RealK-theory} 
}
\end{table}

For $K_0$  we require
\[
\left[\begin{array}{ccc}
x\\
 & 1 \\
 & & -1
\end{array}\right]
=
\left [ x \right] .
\]
and only use elements of $\mathrm{GL}_{n}^{[0]}(A,\tau)$
for $n$ even.
For $K_1$  we require
\[
\left[\begin{array}{cc}
x\\
 & 1 
\end{array}\right]
=
\left [ x \right] .
\]
These are standard descriptions of $K_0$  and $K_1$ 
of the $R^*$-algebra
\[
\Re_{\tau}(A)
=
\left\{ a\in A\left|\, a^{\tau}=a^{*}\right.\right\} . 
\]

For $K_2$  we require
\[
\left[\begin{array}{ccc}
x\\
 & 0 & 1 \\
 & -1 & 0
\end{array}\right]
=
\left [ x \right] .
\]
and only use elements of $\mathrm{GL}_{n}^{[0]}(A,\tau)$
for $n$ even.  This is not a standard picture of
$K_2,$ so we must prove it is equivalent to the standard
picture where the symmetries required on the invertible
elements of $\mathbf{M}_{2n}(A)$ are
\[
x^\tau = x^*,\ x^2 = -1
\]
and the same neutral element.
This is as described in \cite{woodBott}.  Notice that it
is a feature of Bott periodicity that we avoid higher
homotopy groups by changing symmetries.  Even the first
homotopy group introduces loops, which are bad computationally
as they introduce an infinite object even when studying
finite-dimensional matrices.

\begin{lem}
If $u$ in a unitary in a $C^{*}$-algebra then
\[
u^{*}=-u\iff u^{2}=-1.
\]
\end{lem}

\begin{lem}
\label{lem:polarOfSkewHermitian}
Suppose $x$ is an invertible element in a $C^{*\tau}$-algebra and
\[
x^\tau = x^*,\ x^* = -x.
\]
If $x$ has polar decomposition $x=up$ then $x_t = up^{t}$ satisfies
\[
x_t^\tau = x_t^*,\ x_t^* = -x_t
\]
for all $t$ between $0$ and $1.$ 
\end{lem}

\begin{proof}
With $s=\frac{t-1}{2}$
we have
\[
\left(up^{t}\right)^{*}
=\left(x(x^{*}x)^{s}\right)^{*}
=(x^{*}x)^{s}x^{*}
=-(xx^{*})^{s}x
=-x(x^{*}x)^{s}
=-\left(up^{t}\right)^{*}.
\]
Working with polynomial approximations, one shows that for $f$
applied to a normal element $y$ we have
\[
(f(y))^* = \overline{f}(y^*)
\]
and
\[
(f(y))^\tau = f(y^\tau).
\]
Applied here, the second tells us
\[
\left(up^{t}\right)^{\tau}
=\left(x(x^{*}x)^{s}\right)^{\tau}
=\left((x^{*}x)^{s}\right)^{\tau}x^{\tau}
=\left((x^{\tau}x^{*\tau})^{s}\right)x^{\tau}
=\left((x^{*}x)^{s}\right)x^{*}
=\left(up^{t}\right)^{*}
\]
so $x_{t}^{\tau}=x_{t}^{*}.$ 
\end{proof}

\begin{lem}
\label{lem:PolarOfSqrtNegOne}
Suppose $x$ is an invertible element in a $C^{*\tau}$-algebra and
\[
x^\tau = x^*,\ x^2 = -1.
\]
If $x$ has polar decomposition $x=up$ then $x_t = up^{t}$ satisfies
\[
x_t^\tau = x_t^*,\ x_t^2 = -1
\]
for all $t$ between $0$ and $1.$ 
\end{lem}

\begin{proof}
Again $ up^{t} = x(x^{*}x)^{s}.$ From
\[
\left(x^{*}x\right)\left(xx^{*}\right)=1
\]
we derive
\[
xx^{*}=\left(x^{*}x\right)^{-1}.
\]
Therefore
\[
up^{t}up^{t}
= x(x^{*}x)^{s}x(x^{*}x)^{s}  
=(xx^{*})^s xx(x^{*}x)^s 
=-(xx^{*})^s (x^{*}x)^s  
=-1.
\]
The fact that $x_{t}^{\tau}=x_{t}^{*}$ is as before. 
\end{proof} 

\begin{thm}
Let $A$ be a unital $\tau$-$C^{*}$-algebra. The group $K_{2}^{\tau}(A)$
is isomorphic to the usual $K_{2}$ of the underlying $R^{*}$-algebra
\[
A_0 = \left\{ a\in A\left|\, a^{\tau}=a^{*}\right.\right\} .\]
The natural isomorphism is induced by a homotopy equivalence of
\be
\left\{ x\in\mathrm{GL}_{2n}(A)\left|\strut\, x^{*}=x^{\tau},\ x^{*}=-x\right.\right\} 
\ee
with 
\be
\label{eq:standardGL2}
\left\{ x\in\mathrm{GL}_{2n}(A)\left|\strut\, x^{*}=x^{\tau},\ x^{2}=-1\right.\right\} 
\ee
that sends, in either direction, the class of $x$ to the class of
$\mathrm{polar}(x)$. 
\end{thm}

\begin{proof}
The definition of $K_2(A_0)$ is  $\pi_{1}$ of the union of the $\mathrm{GL}_{n}^{[1]}(A,\tau),$ for all $n.$
Part of Bott periodicity, as in 
\cite{woodBott}, states that there is a natural isomorphism
of $K_2(A_0)$ with $\pi_{0}$
of the union over even $m$ of the sets in eq.~\ref{eq:standardGL2}.

The proof of the stated homotopy equivalence is contained in the
previous three lemmas.  We leave to the reader to verify
that $[x] \mapsto [\mathrm{polar}(x)]$
respects the inclusions as $n$ increases, and is natural.
\end{proof}

\subsection{$K_4$ $K_5$ and $K_6$ for real $C^*$-algebras}

The key isomorphism in this subsections is, 
at the $R^*$-algebra level,
\[
K_n(A \times \mathbb {H} ) \cong 
K_{n+4}(A) .
\]
We will work, however, in the form of $\tau$-operations.
This should be more familiar to physicists.

We define $K_{n}^{\tau}(A),$ for $4\leq n < 8$ by
\[
K_{n-4}^{\tau\otimes\sharp}(A\otimes\mathbf{M}_{2}(\mathbb{C})).
\]
In the alternate notation, which is perhaps more clear in this instance,
we take as definition 
\[
K_{n}(A,\tau)=K_{n-4}(A\otimes\mathbf{M}_{2}(\mathbb{C}),\tau\otimes\sharp).
\]

\subsection{Examples}

For a compact Hausdorff space $X$ the $C^{*}$-algebra $C(X)=C(X,\mathbb{C})$
is equal to its opposite. We could take a variety of order-$2$ homeomorphisms
of $X$ to create a variety of $\tau$ operations on $C(X),$ but
the only one with obvious relevance to the physics in this paper is
the trivial choice, so $\tau=\mathrm{id}.$ The $\tau$-$C^{*}$-algebra
$\left(C(X),\mathrm{id}\right)$ has associated $R^{*}$-algebra $C(X,\mathbb{R}).$ 

The interesting element in the $K$-theory of $C(S^{d},\mathbb{R})$
is lives in $K_{-d}=K_{8-d}.$ For the $2$-sphere, we consider a
$K$-theory class $[b]$ in 
\[
K_{6}(C(S^{2},\mathbb{R}))
=K_{2}(C(S^{2})\otimes\mathbf{M}_{2}(\mathbb{C}),\mathrm{id}\otimes\sharp).
\]
Let $f_{1}(x,y,z)=x,$ $f_{2}(x,y,z)=y$ and $f_{3}(x,y,z)=z,$ where
we regard $S^{2}$ as the unit ball in $\mathbb{R}^{3}.$ Using the
Pauli spin matrices we define
\[
b=\sum f_{r}\otimes i\sigma_{r}.
\]
Since the $\sigma_{r}$ anti-commute, square to $-1,$ 
have $\sigma_{r}^{\sharp}=-\sigma_{r}$
and $\sigma_{r}^{\dagger}=\sigma_{r}$ we discover
\[
b^{\dagger}=-b,\ b^{\tau}=b^{\dagger},\  b^{2}=-1
\]
where $\tau=\mathrm{id}\otimes\sharp.$ Notice $b$ is unitary, and
has the correct symmetries to define an element
\[
[b]\in K_{2}(C(S^{2})\otimes\mathbf{M}_{2}(\mathbb{C}),\mathrm{id}\otimes\sharp)
\]
but we cannot use  $b$  as it is in $\mathrm{GL}_{1}^{[0]}$ of
$C(S^{2})\otimes\mathbf{M}_{2}(\mathbb{C}).$  Instead we must
use
\[
b \oplus i\sigma_x 
=
\left[\begin{array}{cc}
b\\
 & i\sigma_{x}\end{array}\right] .
 \]
There are topological ways to see $[b \oplus i\sigma_x]$ is not the trivial element
in this group, but it also follows from our example in \cite{hastingsloring},
since triviality of $[b \oplus i\sigma_x]$ would have forced our Pfaffian-Bott index
to be zero for triples of self-dual Hermitian matrices with small
commutators.

The finite dimensional example, that was essential in our study of
2D TR invariant systems, is 
\[
K_{6}\left(\mathbf{M}_{2N}(\mathbb{C}),\sharp\right)
=K_{2}\left(\mathbf{M}_{2N}(\mathbb{C})\otimes\mathbf{M}_{2}(\mathbb{C}),\sharp\otimes\sharp\right).
\]
Using the isomorphism above, and some standard facts in real $K$-theory
\cite{schroederKtheory}, we find
\[
K_{2}\left(\mathbf{M}_{4N}(\mathbb{C}),(\mbox{--})^{\mathrm{T}}\right)
\cong 
K_{2}\left(\mathbf{M}_{4}(\mathbb{C}),(\mbox{--})^{\mathrm{T}}\right)
\cong 
K_{2}\left(\mathbf{M}_{4}(\mathbb{R})\right)
\cong
\mathbb{Z}/2.
\]
We need more than just an example of an $K_{2}$-class here. We need
an explicit isomorphism
\[
K_{2}\left(\mathbf{M}_{4N}(\mathbb{C}),(\mbox{--})^{\mathrm{T}}\right)
\rightarrow
\left\{ \pm 1 \right\} .
\]
That is, given an invertible, skew-Hermitian and skew-symmetric matrix,
we must determine which of the two classes in $K_2$ it is in.
The mapping is
\[
\left[X\right]
\mapsto
\mathrm{sgn}\left(\mathrm{Pf}(X)\right)
\]
The isomorphism
\[
K_{2}\left(\mathbf{M}_{2N}(\mathbb{C})\otimes\mathbf{M}_{2}(\mathbb{C}),\sharp\otimes\sharp\right)\rightarrow\left\{ \pm1\right\} 
\]
then is 
\[
\left[X\right]\mapsto\mathrm{sgn}\left(\mathrm{Pf}(U^{\dagger}XU)\right)
\]
where
\[
U=\frac{1}{\sqrt{2}}\left(I+Z\otimes\sigma_{2}\right),
\]
and where $X$ is invertible with
\[
X^{\sharp\otimes\sharp}=X^{\dagger}=-X.
\]

These Pfaffians can be effectively computed using
the following theorem.  Notice skew-Hermitian plus
skew-symmetric implies the matrix is real.

\begin{thm}
\label{thm:factorSkeySymm}
If $A$ in $\mathbf{M}_{2N}(\mathbb{R})$ is invertible 
and skew-symmetric, then there is an orthogonal matrix $U$ and positive
real numbers $a_{1},\ldots,a_{N}$ so that $\det(U)=1$ and
$A=UDU^{\mathrm{T}}$with
\[
D=\left[
\begin{array}{cccccccc}
0 & \epsilon a_{1}\\
-\epsilon a_{1} & 0 & 0\\
 & 0 & 0 & a_{2}\\
 &  & a_{2} & 0 & 0\\
 &  &  & 0 & \ddots & \ddots\\
 &  &  &  & \ddots & 0 & 0\\
 &  &  &  &  & 0 & 0 & a_{N}\\
 &  &  &  &  &  & a_{N} & 0
\end{array}\right]
\]
where $\epsilon=\pm1.$ Moreover
\[
\mathrm{Pf}\left(X\right)=\epsilon a_{1}a_{2}\cdots a_{N}.
\]
The eigenvalues of $X$ are $\pm a_{1},\pm a_{2},\ldots,\pm a_{N}$
with associated eigenvectors
\[
U
\left[\begin{array}{cccccccccccc}
0 & 0 & \cdots & 0 & 0 & \frac{1}{\sqrt{2}} & \frac{\pm i}{\sqrt{2}} & 0 & 0 & \cdots & 0 & 0
\end{array}\right]^{\mathrm{T}}.
\]
\end{thm}

\begin{proof}
The Schur factorization
of real matrices gives us an orthogonal matrix $U$ and a block upper-triangular
matrix $D,$ with $1$-by-$1$ or $2$-by-$2$ blocks on the diagonal,
and $A=UDU^{\mathrm{T}}.$ Possibly by altering a single row in $U,$
and a single column and row in $D$ we can assure $\det(U)=1.$ Since
$U$ is orthogonal, $D$ is real, invertible and skew-symmetric, so
$B$ is actually block diagonal and no $1$-by-$1$ blocks and has
only $2$-by-$2$ blocks of the form 
\[
\left[\begin{array}{cc}
0 & -a_{j}\\
a_{j} & 0\end{array}\right].
\]
We are mostly finished. The signs of the $a_{j}$
are arbitrary, but conjugating by a diagonal matrix with $\pm1$
on the diagonal, and with an even number of negative elements, we
can get the signs as indicated. 

The rest of the statements are standard, using in particular the formulas
\[
\mathrm{Pf}(YXY^{\mathrm{T}})=\det(Y)\mathrm{Pf}(X)
\]
\[
\mathrm{Pf}\left(X\right)=\left[\begin{array}{cc}
0 & z\\
-z & 0\end{array}\right]=z
\]
and
\[
\mathrm{Pf}\left(
\left[\begin{array}{cc}
X_{1} & 0\\
0 & X_{2}
\end{array}\right]
\right)
=
\mathrm{Pf}\left(
\begin{array}{c}
X_{1}\end{array}\right)
\mathrm{Pf}\left(
\begin{array}{c}
X_{2}\end{array}
\right)
\]
which holds for all $Y$ and all skew-symmetric $X,$ $X_{1}$ and
$X_{2}.$ Taken together, these formulas are a good definition of
the Pfaffian. 
\end{proof}

\subsection{Pushing forward real $K$-theory}

We need a version of Theorem~\ref{thm:GUEtrivialIndex} to apply in the GSE and GOE cases,
to show that when we start with commuting self-dual Hermitians, the
$Z_{2}$ index is trivial. As our numerical studies are for 2D and
3D GSE lattices, we state as plainly as possible a result that covers
those two cases. We state for future reference a theorem to cover
the GSE and GOE in all dimensions in as brief a form as possible.

The maps between $C^{*\tau}$-algebras that induce homomorphism of
$K$-theory groups are the $*$-$\tau$-homomorphisms, meaning the
algebra homomorphisms that preserve both the $*$ and $\tau$ operations.
\begin{thm}
Suppose $X$ is a pathwise connected compact subspace of $\mathbb{R}^{d}.$
For any unital $*$-homomorphism 
\[
\varphi:C(X)\rightarrow\mathbf{M}_{2N}(\mathbb{C})
\]
such that 
\[
\varphi(f)=\left(\varphi(f)\right)^{\sharp}
\]
the induced homorphisms 
\[
\varphi_{*}:K_{5}(C(X), \mathrm{id})\rightarrow K_{5}(\mathbf{M}_{2N}(\mathbb{C}),\sharp)
\]
and
\[
\varphi_{*}:K_{6}(C(X), \mathrm{id})\rightarrow K_{6}(\mathbf{M}_{2N}(\mathbb{C}),\sharp)
\]
are the trivial maps.
\end{thm}

\begin{proof}
Using the structured spectral theorem, applied to $\varphi(h_{r})$
for $h_{r}$ the various coordinate functions of Euclidean space,
restricted to $X,$ we find a symplectic unitary $U$ and points $x_{1},\ldots,x_{N}$
so that
\[
\varphi(f)=U\left[\begin{array}{cccccc}
f(x_{1})\\
 & \ddots\\
 &  & f(x_{N})\\
 &  &  & f(x_{1})\\
 &  &  &  & \ddots\\
 &  &  &  &  & f(x_{N})\end{array}\right]U^{\dagger}.
\]
All symplectic unitary matrices are homotopic, and all points in $X$
connected by paths, so we may assume $U=I$ and $x_{j}=x_{1}$ and
so\[
\varphi(f)=f(x_{1})I.
\]
If we let $\varphi_{0}$ denote the inclusion $\mathbb{C}\hookrightarrow\mathbf{M}_{2N}(\mathbb{C})$
and let $\psi:C(X)\rightarrow\mathbb{C}$ denote the map $f\mapsto f(x_{1})$
then we have shown that, as far as $K$-theory is concerned, $\psi$
is the composition $\iota\circ\psi,$ i.e. 
$\varphi_{*}=\iota_{*}\circ\psi_{*}.$  Since 
\[
K_{5}(\mathbb{C},\mathrm{id})=K_{5}(\mathbb{R})=K_{1}(\mathbb{H})=0
\]
and 
\[
K_{6}(\mathbb{C},\mathrm{id})=K_{6}(\mathbb{R})=K_{2}(\mathbb{H})=0
\]
(see \cite{schroederKtheory}, for example) we see that $\varphi_{*}$ is the
trivial map. 
\end{proof}

\begin{thm}
Suppose $X$ is a pathwise connected compact subspace of $\mathbb{R}^{d}.$
For any unital $*$-homomorphism 
\[
\varphi:C(X)\rightarrow\mathbf{M}_{n}(\mathbb{C})
\]
such that  
\[
\varphi(f)=\left(\varphi(f)\right)^{\tau}
\]
the induced homomorphism 
\[
\varphi_{*}:K_{q}(C(X),\mathrm{id})\rightarrow K_{q}(\mathbf{M}_{n}(\mathbb{C}),\tau))
\]
must send to the trivial element the kernel of the map
\[
\delta_{x}:K_{q}(C(X),\mathrm{id})\rightarrow K_{q}(\mathbb{C},\bar{\ })
\]
where $\delta_{x}$ sends $f$ in $C(X)$ to $f(x),$ for some chosen
base point $x$ in $X.$
\end{thm}

\begin{proof}
The proof is essentially as above, except that if $\tau$ is the transpose,
then we use a simultaneous diagonalization of commuting symmetric
matrices by a real orthogonal matrix.
\end{proof}

We now return to the discussion of pushing forward $K$-theory 
by ``almost homomorphisms.''
Consider once more
\[
u=\left[\begin{array}{cc}
x_{3}+ix_{4} & -x_{1}+ix_{2}\\
x_{1}+ix_{2} & x_{3}-ix_{4}\end{array}\right]
\in
\mathbf{M}_{2}(C(S^{3}))=C(S^{3})\otimes\mathbf{M}_{2}(\mathbb{C})
\]
but now notice it has symmetry 
\[
u^{*}=u^{\mathrm{id}\otimes\sharp}.
\]
This is defining an element of 
\[
K_{1}(C(S^{3},\mathrm{id})\otimes(\mathbf{M}_{2}(\mathbb{C}),\sharp))
\cong
K_{-3}(C(S^{3},\mathrm{id}))
\]
Now when we form 
\[
U=U(a_{1},a_{2},a_{3},a_{4})=\left[\begin{array}{cc}
a_{3}+ia_{4} & -a_{1}+ia_{2}\\
a_{1}+ia_{2} & a_{3}-ia_{4}\end{array}\right]\in A\otimes\mathbf{M}_{2}(\mathbb{C})
\]
for a $\delta$-representation of the $3$-sphere in $(A,\tau),$
we obtain $U$ that is invertible and has
\[
U^{\tau\otimes\sharp}=U^{*}.
\]
We have our induced element in 
\[
K_{1}(A\otimes(\mathbf{M}_{2}(\mathbb{C}),\sharp))\cong K_{-3}(A).
\]

\section{Numerical Implementation}
\label{numericssection}

The main computations needed for the 2D and 3D GSE studies 
involve matrix functions of self-dual matrices and either the Pfaffian 
or determinant of a real matrix. The Pfaffian and most of the matrix 
functions depend on a structured factorization so we consider 
that topic first. Next we consider Newton's method for 
computing the polar of a matrix, and finally matrix 
functions of Hermitian and unitary matrices.

\subsection{Factorization of self-dual matrices \label{sub:Factorization-of-self-dual}}

Taking advantage of the block structure a self-dual matrix can create
a block-diagonal matrix and so reduce the complexity of further computations.
The algorithm we use is known as the Paige/van Loan algorithm.

\begin{thm}
\label{thm:(van-Loan)}
(Paige/van Loan) If $X$ is a self-dual, $2N$-by-$2N$
matrix, there is an order $N^{3}$ algorithm that produces a symplectic
unitary $U$ so that
\[
U^{\dagger}XU=
\left[\begin{array}{cc}
A & C\\
0 & A^{\mathrm{T}}
\end{array}\right]
\]
for some $A$ and $C$ in $\mathbf{M}_{N}(\mathbb{C})$ with $A$
in upper Hessenberg form and real numbers on the sub-diagonal.
\end{thm}

\begin{proof}
We know 
\[
X=\left[\begin{array}{cc}
A & C\\
B & A^{\mathrm{T}}\end{array}\right]
\]
with $B^{\mathrm{T}}=-B.$ A standard fact in numerical linear algebra
is that there is a Householder unitary $V$ so that $VB$ is a matrix
with zeros in the first column below the first subdiagonal and such
that $V$ fixes $\mathbf{e}_{1}.$ It follows that $B^\prime=VBV^{\mathrm{T}}$
also has zeros in the first column below the subdiagonal, so we consider
the symplectic conjugation
\[
\left[\begin{array}{cc}
\overline{V} & 0\\
0 & V\end{array}\right]\left[\begin{array}{cc}
A & C\\
B & A^{\mathrm{T}}\end{array}\right]\left[\begin{array}{cc}
V^{\mathrm{T}} & 0\\
0 & V^{\dagger}\end{array}\right]=\left[\begin{array}{cc}
A^{\prime} & C^{\prime}\\
B^{\prime} & A^{\prime^{\mathrm{T}}}\end{array}\right].
\]
This unitary is symplectic, so we retain the self-dual structure and
so $B^{\prime}$ is skew-symmetric, hence zero on the diagonal and
so the first column of $B^{\prime}$ will be a scalar multiply of
\textbf{$\mathbf{e}_{2}.$ }Again by standard techniques we find $\alpha$
and $\beta$ in $\mathbb{C}$ so that $\alpha^{2}+\beta^{2}=1$ and
\[
W=\alpha\left|\mathbf{e}_{2}\right\rangle \left\langle \mathbf{e}_{2}\right|+\beta\left|\mathbf{e}_{2}\right\rangle \left\langle \mathbf{e}_{N+2}\right|-\overline{\beta}\left|\mathbf{e}_{N+2}\right\rangle \left\langle \mathbf{e}_{2}\right|+\alpha\left|\mathbf{e}_{N+2}\right\rangle \left\langle \mathbf{e}_{N+2}\right|
\]
will be a unitary (a so-called Givens rotation) so that
\[
W\left[\begin{array}{cc}
A^{\prime} & C^{\prime}\\
B^{\prime} & A^{\prime\mathrm{T}}\end{array}\right]W^{\dagger}=\left[\begin{array}{cc}
A^{\prime\prime} & C^{\prime\prime}\\
B^{\prime\prime} & A^{\prime\prime^{\mathrm{T}}}\end{array}\right]
\]
will produce $B^{\prime\prime}$ with the entire first column equal
to zero. The construction of $W$ is such that it will be a symplectic
unitary so $B^{\prime\prime}$ will be skew-symmetric and have first
row zero as well. Another Householder unitary $Q$ can be found so
that $Q$ fixes $\mathbf{e}_{1}$ and $QA^{\prime\prime}$ has first
column with a real scalar in the subdiagonal and zeros below that.
This fact that $Q\mathbf{e}_{1}=\mathbf{e}_{1}$ implies that $\overline{Q}B^{\prime\prime}Q^{\dagger}$
will also have first row and column equal to zero, so
\[
\left[\begin{array}{cc}
Q & 0\\
0 & \overline{Q}\end{array}\right]
\left[\begin{array}{cc}
A^{\prime\prime} & C^{\prime\prime}\\
B^{\prime\prime} & A^{\prime\prime^{\mathrm{T}}}\end{array}\right]
\left[\begin{array}{cc}
Q^{\dagger} & 0\\
0 & Q^{\mathrm{T}}\end{array}\right]
=
\left[\begin{array}{cc}
A^{\prime\prime\prime} & C^{\prime\prime\prime}\\
B^{\prime\prime\prime} & A^{\prime\prime\prime^{\mathrm{T}}}
\end{array}\right]
\]
will have $B^{\prime\prime\prime}$ with zeros in the first row and
column and $A^{\prime\prime\prime}$ having zeros in the first column
below the first subdiagonal and a real scalar on the first subdiagonal.
We use the symplectic unitary
\[
U_{1}=\left[\begin{array}{cc}
Q & 0\\
0 & \overline{Q}\end{array}\right]W\left[\begin{array}{cc}
\overline{V} & 0\\
0 & V\end{array}\right]
\]
to find $U_1^{\dagger}XU_{1}$ is the row $1$ and column
$N+1$ in the correct form and $U_{1}$ fixes both $\mathbf{e}_{1}$
and $\mathbf{e}_{N+1}.$ This last condition allows for an iterative
solution.
\end{proof}

If we apply the Paige/van Loan algorithm to a unitary or Hermitian
matrix, we get an even better outcome, a block-diagonal matrix.

\begin{thm}
If 
\[
X=\left[\begin{array}{cc}
A & C\\
0 & D\end{array}\right]
\]
is self-dual and normal, then $A$ is normal, $D=A^{\mathrm{T}}$ and
$C=0.$
\end{thm}

\begin{proof}
Since $X$ is self dual, $D=A^{\mathrm{T}}.$ Therefore
\[
X^{\dagger}X=\left[\begin{array}{cc}
A^{\dagger} & 0\\
C^{*} & \overline{A}\end{array}\right]\left[\begin{array}{cc}
A & C\\
0 & A^{\mathrm{T}}\end{array}\right]=\left[\begin{array}{cc}
A^{\dagger}A & A^{\dagger}C\\
C^{\dagger}A & \overline{A}A^{\mathrm{T}}+C^{*}C\end{array}\right]
\]
and
\[
XX^{\dagger}=\left[\begin{array}{cc}
A & C\\
0 & A^{\mathrm{T}}\end{array}\right]\left[\begin{array}{cc}
A^{\dagger} & 0\\
C^{\dagger} & \overline{A}\end{array}\right]=\left[\begin{array}{cc}
AA^{\dagger}+CC^{\dagger} & C\overline{A}\\
A^{\mathrm{T}}C^{\dagger} & A^{\mathrm{T}}\overline{A}\end{array}\right].
\]
The normality of $X$ tells us $AA^{\dagger}\leq A^{\dagger}A.$ For
matrices, hyponormal implies normal so $A^{\dagger}A=AA^{\dagger}$
and $C^{\dagger}C=0.$ Thus $C=0.$ 
\end{proof}

We have now a method of diagonalizing a unitary or Hermitian self-dual
matrix $X.$ From the Paige/van Loan algorithm, we obtain a symplectic
unitary $U$ so that
\[
X=U\left[\begin{array}{cc}
Y & 0\\
0 & Y^{\mathrm{T}}\end{array}\right]U^{\dagger}
\]
with $Y$ being unitary or Hermitian. In the unitary case use use
standard algorithms to obtain a Schur factorization $Y=VDV^{\dagger}$
with $D$ upper triangular and $V$ unitary. As $D$ is unitary, it
is also diagonal. In the Hermitian case, we use a standard eigensolver
to diagonalize $Y.$ In either case,
\[
X=U\left[\begin{array}{cc}
V & 0\\
0 & \overline{V}\end{array}\right]\left[\begin{array}{cc}
D & 0\\
0 & D\end{array}\right]\left[\begin{array}{cc}
V^{\dagger} & 0\\
0 & V^{\mathrm{T}}\end{array}\right]U^{\dagger}
\]
so the desired symplectic unitary is
\[
U\left[\begin{array}{cc}
V & 0\\
0 & \overline{V}\end{array}\right].
\]

In practice, when we diagonalize a self-dual self-adjoint matrix,
specifically a Hamiltonian in a GSE system, we store only the left
half of the all intermediate matrices that are either Hermitian self-dual
or symplectic unitary. This allows the structured eigensolver to be
competitive in speed with unstructured eigensolvers. The timing of
such algorithms is architecture dependent, and very sensitive to the
test matrices, but Table~\ref{tab:self-dualEigenTest} gives some
ideal of accuracy and speed.

\begin{table}
\begin{tabular}{|c|c|c|c|c|c|c|c|}
\hline 
\noalign{\vskip\doublerulesep}
$N$ & \multicolumn{2}{c|}{Time (sec.)} & \multicolumn{2}{c|}{$\tfrac{1}{10^{12}}\left\Vert \strut H-UDU^{\dagger}\right\Vert $} & \multicolumn{2}{c|}{$\max\left|\lambda_{2j+1}-\lambda_{j}\right|$} & $\tfrac{1}{10^{12}}\left\Vert \strut H-UDU^{\dagger}\right\Vert $\tabularnewline[\doublerulesep]
\hline 
\noalign{\vskip\doublerulesep}
 & stand. & struct. & stand. & struct. & stand. & struct. & adjusted standard\tabularnewline[\doublerulesep]
\hline 
500 & 7.6200 & 6.4800 & 0.0176 & 0.0129 & 0.0066 & 0.0000 & 0.7516\tabularnewline
\hline 
1000 & 62.530 & 60.020 & 0.0259 & 0.0206 & 0.0112 & 0.0000 & 2.0157\tabularnewline
\hline 
1500 & 207.68 & 200.17 & 0.0307 & 0.0225 & 0.0115 & 0.0000 & 6.0742\tabularnewline
\hline 
2000 & 488.13 & 472.22 & 0.0360 & 0.0299 & 0.0162 & 0.0000 & 10.796\tabularnewline
\hline
\end{tabular}
\caption{
\label{tab:self-dualEigenTest}
Comparison of eigensolvers for $H^{\dagger}=H=H^{\sharp}$
in $\mathbf{M}_{2N}(\mathbb{C}).$ In the standard case, using ZHEEV,
$D$ is diagonal, and $U$ is unitary. In the structured case, $U$
is a symplectic unitary and $D$ is diagonal with doubled eigenvalues.
The last column reflects the error arising from adjusting eigenvector
pairs in the output of ZHEED to force the expected symmetry. Each
reported average is over 10 test matrices.
}
\end{table}

We need to deal with approximate unitary matrices as well, so have
the following variation.

\begin{thm}
If 
\[
X=\left[\begin{array}{cc}
A & C\\
0 & D\end{array}\right]
\]
is self-dual and $\left\Vert X^{\dagger}X-I\right\Vert \leq\delta$
then $D=A^{\mathrm{T}}$ and $\left\Vert A^{\dagger}A-I\right\Vert \leq\delta$
and $\left\Vert C\right\Vert \leq\sqrt{2\delta}.$
\end{thm}

\begin{proof}
Since $X$ is self dual, $D=A^{\mathrm{T}},$ and 
\[
X^{\dagger}X-I=\left[\begin{array}{cc}
A^{\dagger} & 0\\
C^{\dagger} & \overline{A}\end{array}\right]\left[\begin{array}{cc}
A & C\\
0 & A^{\mathrm{T}}\end{array}\right]-I=\left[\begin{array}{cc}
A^{\dagger}A-I & A^{\dagger}C\\
C^{\dagger}A & \overline{A}A^{\mathrm{T}}+C^{\dagger}C-I\end{array}\right]
\]
which implies $\left\Vert A^{\dagger}A-I\right\Vert \leq\delta.$ As
we are working in finite dimensions, we know $\left\Vert AA^{\dagger}-I\right\Vert \leq\delta$
and $\left\Vert XX^{\dagger}-I\right\Vert \leq\delta.$ From 
\[
XX^{\dagger}-I=\left[\begin{array}{cc}
A & C\\
0 & A^{\mathrm{T}}\end{array}\right]\left[\begin{array}{cc}
A^{\dagger} & 0\\
C^{\dagger} & \overline{A}\end{array}\right]-I=\left[\begin{array}{cc}
AA^{\dagger}+CC^{\dagger}-I & C\overline{A}\\
A^{\mathrm{T}}C^{\dagger} & A^{\mathrm{T}}\overline{A}-I\end{array}\right]\]
we derive\[
\left\Vert AA^{\dagger}+CC^{\dagger}-I\right\Vert \leq\delta\]
so\[
\left\Vert C\right\Vert ^{2}=\left\Vert CC^{\dagger}\right\Vert \leq\left\Vert AA^{\dagger}+CC^{\dagger}-I\right\Vert +\left\Vert AA^{\dagger}-I\right\Vert \leq2\delta.
\]
\end{proof}

\subsection{Factorization of skew-symmetric matrices}

We need to calculate the Pfaffian of a Hermitian, skew-symmetric matrix
$Y$ in $\mathbf{M}_{2N}(\mathbb{C}).$ Rather than factor $Y$ as
in Theorem~\ref{thm:factorSkeySymm} we set $X=-iY,$ we compute the Pfaffian of the real
matrix $X$ and use formula
\[
\mathrm{Pf}\left(iX\right)=i^{N}\mathrm{Pf}\left(X\right)\quad X\in\mathbf{M}_{2N}(\mathbb{C}).
\]
For a real, skew-symetric matrix, we can combine
standard algorithms to get a factorization that exposes the spectrum
and Pfaffian at the same time.

\begin{thm}
\label{thm:diagSkewSymAlg} 
There is an order $N^{3}$ algorithm which,
for $X$ in $\mathbf{M}_{2N}(\mathbb{R})$ that is skew-symmetric,
calculates $\mathrm{Pf}(X),$ an orthogonal matrix $U$ and real numbers
$a_{1},\ldots,a_{N}$ so that
$
X=UDU^{\mathrm{T}}
$
with
\[
D=\left[\begin{array}{cccccccc}
0 & a_{1}\\
-a_{1} & 0 & 0\\
 & 0 & 0 & a_{2}\\
 &  & -a_{2} & 0 & 0\\
 &  &  & 0 & \ddots & \ddots\\
 &  &  &  & \ddots & 0 & 0\\
 &  &  &  &  & 0 & 0 & a_{N}\\
 &  &  &  &  &  & -a_{N} & 0\end{array}\right].
\]
\end{thm}

\begin{proof}
We start with a Hessenberg decomposition $X=QYQ^{\mathrm{T}}$ with
$Q$ orthogonal and $Y$ having all zeros below the first subdiagonal.
As $Y=Q^{\mathrm{T}}XQ$ must be real with $Y^{\mathrm{T}}=-Y,$ we
in fact have $Y$ is tridiagonal, indeed of the precise form
\begin{equation}
Y=\left[\begin{array}{cccccccc}
0 & c_{1}\\
-c_{1} & 0 & b_{1}\\
 & -b_{1} & 0 & c_{2}\\
 &  & -c_{2} & 0 & b_{2}\\
 &  &  & b_{2} & \ddots & \ddots\\
 &  &  &  & \ddots & 0 & b_{N-1}\\
 &  &  &  &  & -b_{N-1} & 0 & c_{N}\\
 &  &  &  &  &  & -c_{N} & 0\end{array}\right].
\label{eq:HessForm}
\end{equation}
The Hessenberg decomposition can computed in LAPACK by DGEHRD.

We now follow ideas of Golub \cite{golub1968least}. Let $P$ be the
permutation matrix corresponding to the shuffle so that
\[
P^{\mathrm{T}}YP=\left[\begin{array}{cc}
0 & C\\
-C^{\mathrm{T}} & 0\end{array}\right]
\]
where
\[
C=\left[\begin{array}{ccccc}
c_{1}\\
-b_{1} & c_{2}\\
 & -b_{2} & \ddots\\
 &  & \ddots & c_{N-1}\\
 &  &  & b_{N-1} & c_{N}\end{array}\right].
\]
Take a singular value decomposition $C=WEV$ with $W$ and $V$ orthogonal
and 
\[
E=\left[\begin{array}{ccccc}
a_{1}\\
 & a_{2}\\
 &  & \ddots\\
 &  &  & a_{N-1}\\
 &  &  &  & a_{N}\end{array}\right]
\]
for $a_{j} \geq 0.$ Then
\[
\left[\begin{array}{cc}
W^{\mathrm{T}}\\
 & V\end{array}\right]\left[\begin{array}{cc}
0 & C\\
-C^{\mathrm{T}} & 0\end{array}\right]\left[\begin{array}{cc}
W\\
 & V^{\mathrm{T}}\end{array}\right]=\left[\begin{array}{cc}
 & E\\
-E\end{array}\right]
\]
and so
\[
P\left[\begin{array}{cc}
W^{\mathrm{T}}\\
 & V\end{array}\right]\left[\begin{array}{cc}
0 & C\\
-C^{\mathrm{T}} & 0\end{array}\right]\left[\begin{array}{cc}
W\\
 & V^{\mathrm{T}}\end{array}\right]P^{\mathrm{T}}=P\left[\begin{array}{cc}
 & E\\
-E\end{array}\right]P^{\mathrm{T}}
=
D
\]
where
\[
D=\left[\begin{array}{cccccccc}
0 & a_{1}\\
-a_{1} & 0 & 0\\
 & 0 & 0 & a_{2}\\
 &  & -a_{2} & 0 & 0\\
 &  &  & 0 & \ddots & \ddots\\
 &  &  &  & \ddots & 0 & 0\\
 &  &  &  &  & 0 & 0 & a_{N}\\
 &  &  &  &  &  & -a_{N} & 0\end{array}\right].
\]
Next 
\[
P\left[\begin{array}{cc}
W^{\mathrm{T}}\\
 & V\end{array}\right]P^{\mathrm{T}}Q^{\mathrm{T}}XQP\left[\begin{array}{cc}
W\\
 & V^{\mathrm{T}}\end{array}\right]P^{\mathrm{T}}=D
 \]
so we are done with
\[
U=QP\left[\begin{array}{cc}
W\\
 & V^{\mathrm{T}}\end{array}\right]P^{\mathrm{T}}.
\]
\end{proof}

\subsection{Stuctured polar decomposition}

As argued by Higham \cite{highamFunctOfMat}, a fast and accurate
way to compute the polar of a matrix is via Newton's method. Suppose
$X$ is an invertible complex matrix. To compute $U$ unitary and
$P$ positive with $X=UP$ we set $X_{1}=X$ and
\[
X_{k+1}=\frac{1}{2}\left(X_{n}+\left(X_{n}^{-1}\right)^{\dagger}\right)
\]
and get $U=\lim_{k}X_{k}.$ The convergence is quadratic in the operator
norm, in in practice we used three iterations.

\begin{thm}
\label{polarSymmetries}
If $X^{\sharp}=\pm X$ then $X_{k}^{\sharp}=\pm X_{k}$ for all $k$
and 
\[
\mathrm{polar}\left(X\right)^{\sharp}=\pm X.
\]
If $X^{\sharp}=\pm X^{\dagger}$ then $X_{k}^{\sharp}=\pm X_{k}^{\dagger}$
for all $k$ and 
\[
\mathrm{polar}\left(X\right)^{\sharp}=\pm X^{\dagger}.
\]
If $X^{\mathrm{T}}=\pm X$ then $X_{k}^{\mathrm{T}}=\pm X_{k}$ for
all $k$ and 
\[
\mathrm{polar}\left(X\right)^{\mathrm{T}}=\pm X.
\]
If $X^{\mathrm{T}}=\pm X^{\dagger}$ then $X_{k}^{\mathrm{T}}=\pm X_{k}^{\dagger}$
for all $k$ and 
\[
\mathrm{polar}\left(X\right)^{\mathrm{T}}=\pm X^{\dagger}.
\]
\end{thm}

\begin{proof}
From the identities
$\left(X^{\sharp}\right)^{-1}=\left(X^{-1}\right)^{\sharp}$
and
$\left(X^{\sharp}\right)^{\dagger}=\left(X^{\dagger}\right)^{\sharp}$
we deduce
\[
X_{k+1}^{\sharp}=\frac{1}{2}\left(X_{k}^{\sharp}+\left(\left(X_{k}^{\sharp}\right)^{-1}\right)^{\dagger}\right)
\]
and
\[
\mathrm{polar}(X)^{\sharp}=\mathrm{polar}(X^{\sharp}).
\]
Similarly
\[
X_{k+1}^{\mathrm{T}}
=\frac{1}{2}\left(X_{k}^{\mathrm{T}}+\left(\left(X_{k}^{\mathrm{T}}\right)^{-1}\right)^{\dagger}\right)
\]
and
\[
\mathrm{polar}(X)^{\mathrm{T}}=\mathrm{polar}(X^{\mathrm{T}}).
\]
From $\left(X^{\dagger}\right)^{-1}=\left(X^{-1}\right)^{\dagger}$
and $\left(X^{\dagger}\right)^{\dagger}=X$ we deduce 
\[
X_{k+1}^{\dagger}
=\frac{1}{2}\left(X_{k}^{\dagger}+\left(\left(X_{k}^{\dagger}\right)^{-1}\right)^{\dagger}\right)
\]
and
\[
\mathrm{polar}(X)^{\dagger}=\mathrm{polar}(X^{\dagger}).
\]
Clearly 
\[
-X_{k+1}=\frac{1}{2}\left(-X_{n}-\left(X_{n}^{-1}\right)^{\dagger}\right)
\]
and so 
\[
\mathrm{polar}(-X)=-\mathrm{polar}(X).
\]
All eight versions of the theorem follow from some combination of
these formulas.
\end{proof}

The most accurate way to find the polar of a self-dual matrix $X$
is to apply Newton's method to $X.$ A faster, less accurate method,
which we utilized, is to apply Paige/van Loan to get symplectic unitary
$Q$ with 
\[
X=Q\left[\begin{array}{cc}
A & C\\
0 & A^{\mathrm{T}}\end{array}\right]Q^{\dagger},
\]
apply Newton's method to $A$ to find $U$ unitary with $U\approx\mathrm{polar}(A)$
and then use
\[
Q\left[\begin{array}{cc}
U & 0\\
0 & U^{\mathrm{T}}\end{array}\right]Q^{\dagger}
\]
which will be a self-dual unitary, close to the polar of $A$ when
$X^{\dagger}X\approx I.$ 

We can take logarithm of $U$ by using a Schur factorization, as we
discuss in the next subsections. This produces for us a Hermitian
matrix $Y$ with $e^{2\pi iY}\approx U,$ and from here form
\[
G=Q\left[\begin{array}{cc}
Y & 0\\
0 & Y^{\mathrm{T}}\end{array}\right]Q^{\dagger}
\]
which is Hermitian self-dual with $e^{2\pi iG}\approx X.$ We summarize
our ability to find self-dual logarithms of self-dual approximate
unitary matrices in Table, using test matrices that are very close
to being unitary. We compare our results to what results from using
an unstructured Schur factorization. Clearly we could do better in
calculating the log in the standard case, by attempting to pair the
eigenvalues and adjust the eigenvectors before taking log. The alarming
errors are a warning that something must be done when applying discontinuous
functions to normal self-dual matrices to
get any semblance of a self-dual matrix out.

\begin{table}
\begin{tabular}{|c|c|c|c|c|c|c|c|}
\hline 
\noalign{\vskip\doublerulesep}
$N$ & $\tfrac{1}{10^{12}}\left\Vert \strut U^{\dagger}U-I\right\Vert $ & \multicolumn{2}{c|}{Time (sec.)} & \multicolumn{2}{c|}{$\tfrac{1}{10^{12}}\left\Vert \strut U-QDQ^{\dagger}\right\Vert $} & \multicolumn{2}{c|}{$\left\Vert \strut H^{\sharp}-H\right\Vert $}\tabularnewline[\doublerulesep]
\hline 
\noalign{\vskip\doublerulesep}
 & %
\begin{minipage}[c][0.4in]{1.4in}%
samples are exactly self-dual%
\end{minipage} & stand. & struct. & stand. & struct. & stand. & struct.\tabularnewline[\doublerulesep]
\hline 
500 & 0.1028 & 25.002 & 16.799 & 0.0760 & 0.0743 & 2.8309  & $3.3329\times10^{-14}$\tabularnewline
\hline 
1000 & 0.2457 & 183.11 & 140.52 & 0.1674 & 0.1754 & 2.8049  & $5.1817\times10^{-14}$\tabularnewline
\hline 
1500 & 0.4298 & 641.89 & 465.04 & 0.2876 & 0.3079 & 2.8803  & $5.5677\times10^{-14}$\tabularnewline
\hline 
2000 & 0.6618 & 1480.1 & 1086.8 & 0.4381 & 0.4743 & 2.8320 & $6.5637\times10^{-14}$\tabularnewline
\hline
\end{tabular}
\caption{
\label{tab:self-dualUnitaryTest}
Comparison of eigensolvers/logarithm
for $U^{\dagger}=U^{-1}$, $U^{\sharp}=U^{-1}$ in $\mathbf{M}_{2N}(\mathbb{C}).$
In the standard case, using a Schur decomposition of $U$ found with
ZGEES, $D$ is a unitary diagonal, and $Q$ is unitary. In the structured
case, $Q$ is a symplectic unitary and $D$ is self-dual unitary diagonal.
In each case $H=Q^{\dagger}\tfrac{1}{2\pi i}\log\left(D\right)Q.$
Each reported average is over 10 test matrices.
}
\end{table}

\subsection{Matrix functions of structured Hermitian matrices}

Given a measurable function $f:\mathbb{C}\rightarrow\mathbb{R}$ and
a normal matrix $X$ the definition of the matrix $f(X)$ tells us
to compute it from any unitary diagonalization of $X,$
\begin{equation}
X=U\left[\begin{array}{ccc}
\lambda_{1}\\
 & \ddots\\
 &  & \lambda_{n}\end{array}\right]U^{\dagger}\implies f(X)=U\left[\begin{array}{ccc}
f(\lambda_{1})\\
 & \ddots\\
 &  & f(\lambda_{n})\end{array}\right]U^{\dagger}.
\label{eq:defBorelCalc}
\end{equation}
Notice $f(X)$ is Hermitian since we asked that $f$ be real-valued.
We defined $P$ this way from the Hamiltonian, using the function
$\chi_{(-\infty,E_{F}]},$ and saw in that case that $P$ will be
self-dual when the Hamiltonian is self-dual. This sort of preservation
of structure is true more generally. When $f$ is not continuous we
will generally need to numerically compute $f(X)$ using a structured
unititary factorization, as in the proof of the following. 

\begin{thm}
Suppose $f:\mathbb{C}\rightarrow\mathbb{R}$ is measurable function
$X$ is a normal matrix $X.$ If $X\in\mathbf{M}_{2N}(\mathbb{C})$
is self-dual then $f(X)$ is self-dual. If $X^{\mathrm{T}}=X$ then
$f(X)^{\mathrm{T}}=f(X)$ and so $f(X)$ is real symmetric.
\end{thm}

\begin{proof}
When $X$ is self-dual, we can find a symplectic unitary $U$ so that
the second half of the eigenvalues repeats the first. This implies
that \[
f(X)=U\left[\begin{array}{cccccc}
f(\lambda_{1})\\
 & \ddots\\
 &  & f(\lambda_{N})\\
 &  &  & f(\lambda_{1})\\
 &  &  &  & \ddots\\
 &  &  &  &  & f(\lambda_{N})\end{array}\right]U^{*}\]
is self-dual. An algorithm to find $U$ is discussed in \S\ref{sub:Factorization-of-self-dual}.

Notice that when $X^{\mathrm{T}}=X$ and $X$ is Hermitian, we are
assuming $X$ is a real matrix. In this special case, we can use standard
algorithems to find $U$ in (\ref{eq:defBorelCalc}) that is real
orthogonal. The fact that $f(X)$ is then real symmetric is now obvious.
The more general form of the second claim will arise when applying
the function $f,$ $g$ and $h$ as in Section .... to symmetric unitary
matrices, as wil be necessary in the study of GOE systems on the torus
in dimensions $4,$ $6$ and $7.$ This is a theoretical issue at
this time, so we don't worry about an algorithm. We write $X=A+iB$
for $A=\tfrac{1}{2}(X^{\dagger}+X)$ and $B=\tfrac{i}{2}(X^{\dagger}-X)$
and find that $A$ and $B$ commute (since $X$ is normal) and are
real symmetric since $X^{\mathrm{T}}=X.$ There is then a real orthogonal
matrix $U$ that diagonalizes both, and so diagonalizes $X.$ We again
can use (\ref{eq:defBorelCalc}) to show $f(X)$ is real symmetric.
\end{proof}

\subsection{Matrix functions of almost unitary matrices}

The definition of $f(X)$ corresponds to more naive notions of a function
of a matrix when $f$ is a polynomial, or even a Laurent function,
i.e. a polynomial in $X$ and $X^{-1}.$ Of course the latter will
require $X$ be invertible. For example, if
\[
f(x)=2x^{2}+x^{-1}
\]
then
\[
f(X)=2X^{2}+X^{-1}
\]
since given (\ref{eq:defBorelCalc}) we have
\[
X^{-1}=U\left[\begin{array}{ccc}
\lambda_{1}^{-1}\\
 & \ddots\\
 &  & \lambda_{n}^{-1}\end{array}\right]U^{\dagger}
\]
and
\[
X^{2}=U\left[\begin{array}{ccc}
\lambda_{1}^{2}\\
 & \ddots\\
 &  & \lambda_{n}^{2}\end{array}\right]U^{\dagger}.
\]
Unless high powers are involved, for such functions the diagonalization
will be slower than those more naive calculations. More importantly,
when $X$ is approximately unitary, we have $X^{\dagger}\approx X^{-1}.$ 

Define $X^{(n)}=X^{n}$ for $n$ nonnegative, and $X^{(-n)}=(X^{\dagger})^{n}.$
As long as $|k|$ is small, we have reasonable approximations $X^{(n)}\approx X^{n}$
and 
\[
f(X)\approx\sum_{k=-m}^{m}a_{n}X^{(n)}
\]
when $m$ is small and $f$ is the Laurent polynomial
\[
f(x)=\sum_{n=-m}^{m}a_{n}x^{n}.
\]

The functions we used when mapping the torus to the sphere were
\[
f(\lambda)
=\tfrac{150}{128}\sin(\arg(z))+\tfrac{25}{128}\sin(3\arg(z))+\tfrac{3}{128}\sin(5\arg(z))
\]
\[
g(z)=\begin{cases}
0 & \arg(z)\in[\tfrac{1}{4},\tfrac{3}{4}]\\
\sqrt{1-\left(f(z)\right)^{2}} & \arg(z)\notin[\tfrac{1}{4},\tfrac{3}{4}]\end{cases}
\]
\[
h(z)=\begin{cases}
\sqrt{1-\left(f(z)\right)^{2}} & \arg(z)\in[\tfrac{1}{4},\tfrac{3}{4}]\\
0 & \arg(z)\notin[\tfrac{1}{4},\tfrac{3}{4}]\end{cases}
\]
where now we are using a complex argument $\lambda$ or modulus one,
so that $\lambda^{(n)}=\lambda^{n}.$ These functions were selected
to be close to degree-five Laurent polynomials. Indeed $f$ is exactly
and order-five Laurent polynomial, while estimated the Fourier series
of $g$ and $h$ to arrive at
\[
f(z)=\sum_{n=-5}^{5}a_{n}z^{(n)},\ g(z)\approx\sum_{n=-5}^{5}b_{n}z^{(n)},\ h(z)\approx\sum_{n=-5}^{5}c_{n}z^{(n)}
\]
with the coefficients as in Table~\ref{tab:a_b_c}.

\begin{table}
\begin{centering}
\begin{tabular}{|c||c|c|c|c|c|c|}
\hline 
$n$ & 0 & 1 & 2 & 3 & 4 & 5\tabularnewline
\hline
\hline 
$a_{n}$ & 0 & $\frac{-150i}{256}$ & 0 & $\frac{-25i}{256}$ & 0 & $\frac{-3i}{256}$\tabularnewline
\hline 
$b_{n}$ & 0.20205 & -0.17994 & 0.125655 & -0.06601 & 0.023445 & -0.0038855\tabularnewline
\hline 
$c_{n}$ & 0.20205 & 0.17994 & 0.125655 & 0.06601 & 0.023445 & 0.0038855\tabularnewline
\hline
\end{tabular}
\par\end{centering}
\caption{
\label{tab:a_b_c}
The coefficients defining $f,$ $g,$ and $h,$
extended to negative indices by $a_{-n}=\overline{a_{n}}$ and $b_{-n}=b_{n}$
and $c_{-n}=c_{n}.$
}
\end{table}

Rather than computing
\[
f(\mathrm{polar}(U)),\ g(\mathrm{polar}(U)),\ h(\mathrm{polar}(U))
\]
we compute
\[
F=\sum_{n=-5}^{5}a_{n}U^{(n)},\ G=\sum_{n=-5}^{5}b_{n}U^{(n)},\ H=\sum_{n=-5}^{5}c_{n}U^{(n)}
\]
to save time. We give up exactness in expected relations, so get $F^{2}+G^{2}+H^{2}\approx I$
and $GH\approx0,$ but when $\left\Vert U^{\dagger}U-I\right\Vert $
is small, the errors $\left\Vert F-f(\mathrm{polar}(U))\right\Vert $
are small so we will not change any index. Computing a fifth power
of a matrix introduces little error, so when $U$ is self-dual within
machine precision, so will be $F,$ $G$ and $H.$

\section{Discussion}
We have given a general procedure for obtaining topological invariants of free fermion systems by mapping to a problem in $C^*$-algebra.
The general approach for systems in the 3 classical universality classes on a sphere is to form a projector $P$ onto the space of occupied states, and then
to project the coordinate matrices of the system into the occupied band.  These matrices approximately
commute and the sum of their squares is approximately equal to the identity.  We regard these matrices as describing a soft
sphere $S^d$.
We then look for topological
obstructions to approximating these projected matrices by exactly commuting matrices.  To compute these topological obstructions,
we form the operator $B$ in (\ref{Sdef}).  This operator $B$ approximately squares to the identity.  Thus, we can view this
operator $B$ as describing a soft zero dimensional sphere.  In this paper we have described an approach to computing topological
properties of $B$ in any dimension and six of the ten symmetry classes.

Previously, we implemented a special case of this procedure in two dimensions\cite{hastingsloring}.  The algorithm based on $C^*$-algebra allowed us
to study much larger systems than techniques based on considering
a flux torus\cite{essinmoore} in two dimensions.

In this paper,
we have implemented this procedure numerically for three dimensional time-reversal invariant insulators, considering systems up to
$12^3$ sites, with 4 states per site, for a total of a $6912$-dimensional Hilbert space.  We have found numerical evidence for an unexpected transition in the thermodynamic
limit.  In the presence of disorder, the system has localized states close to the band edge and close to zero energy, with delocalized states
at intermediate energies.  The invariant, averaged over disorder, also appears to display transitions in the thermodynamic limit, being
equal to $+1$ for all samples for $E_F$ near the band edge and
equal to $-1$ in all samples for $E_F$ near zero, and it fluctuates from sample to sample over an intermediate range of energies.  That is, there
appear to be critical values of the energy, $E_c,E_c'$ with $E_c<E_c'0<0$ such that the invariant fluctuates from sample to sample for $E_c<E_F<E_c'$ but the invariant is not fluctuating for $E_F<E_c$ or $E_c'<E_f<0$.   However, while the energy at which the localization transition occurs near zero energy
appears to coincide with the energy $E_c'$ of the transition in the invariant, the transition at energy $E_c'$ appears to be unrelated to
the localization transition.  If so, this reveals an expected phase transition which has no signature in the localization properties.

Topological properties of the soft zero dimensional sphere are described
by one of three properties of the operator $B$.  These are the Pfaffian, the determinant, and the number of positive
eigenvalues of $B$.  In one specific example, in the two dimensional time-reversal  invariant case, we applied a similarly transformation
to $B$\cite{hastingsloring,loringhastings} to make $B$ anti-symmetric, and then computed the Pfaffian.  In the
three dimensional case, we chose a particular set of $\gamma$ matrices to make $B$ block off-diagonal, and then
computed the determinant of one block of $B$.

For systems in the chiral universality classes, Wannier functions are no longer the defining feature of whether or not a given
Hamiltonian can be connected to a trivial Hamiltonian.  Instead, we have a unitary (representing
the upper-right of the spectrally flattened Hamiltonian) which approximately commutes with the position matrices.
This alows us to construct invariants for 6 out of the 10 universality classes from invariants in $C^*$-algebra.  Some of
these invariants were known invariants in $C^*$-algebra, such as the invariants in the GUE case, while other invariants, such
as the invariants of self-dual matrices in two\cite{loringhastings} and three dimensions were new.  We leave the question of
the remaining 4  universality classes for the future, but we expect that these problems also have interpretations in terms of $C^*$-algebra
which will also lead to new algorithms..

\end{document}